\documentclass[12pt]{unbthesis}
\usepackage{graphicx}
\usepackage{subfig}
\usepackage[subfigure]{tocloft} % no number for Vita in ToC
\usepackage{fancyhdr}
\usepackage[english]{babel}
\usepackage{footmisc}
\usepackage{algorithm}
\usepackage{algorithmicx}
\usepackage[noend]{algpseudocode}
\usepackage{listings}
\usepackage{fancyvrb}
\usepackage{color}
\usepackage{hyperref}
\usepackage{verbatim}% to comment multiple lines
\title{Algorithmic and Geometric Aspects of Data Depth With Focus on $\beta$-skeleton Depth}
\author{Rasoul Shahsavarifar}
\predegree{Master of Science, University of Tabriz, Iran, 2008\\
Bachelor of Science, Razi University, Iran, 2006}
\degree{Doctor of Philosophy}
\gau{Computer Science}
\supervisor{David Bremner, Ph.D, Computer Science}
\examboard{Patricia Evans, Ph.D, Computer Science, Chair\\  &Suprio Ray, Ph.D, Computer Science\\ &Jeffrey Picka, Ph.D Mathematics and Statistics}
\externalexam{Pat Morin, Ph.D,
Computer Science, Carleton University}
\date{April, 2019}
\copyrightyear{2019}
\DeclareMathOperator{\HD}{\textit{HD}}
\DeclareMathOperator{\SD}{\textit{SD}}
\DeclareMathOperator{\LD}{\textit{LD}}
\DeclareMathOperator{\SkD}{\textit{SkD}}
\DeclareMathOperator{\SphD}{\textit{SphD}}
\DeclareMathOperator{\RNG}{\textit{RNG}}
\DeclareMathOperator{\Conv}{\textit{Conv}}
\DeclareMathOperator{\Sph}{\textit{Sph}}
\DeclareMathOperator{\AIC}{\textit{AIC}}
\DeclareMathOperator{\BIC}{\textit{BIC}}
\newtheorem{theorem}{Theorem}[section]
\newtheorem{definition}{Definition}[section]
\newtheorem{exmp}{Example}[section]
\newtheorem{note}{Note}[section]
\newtheorem{open problem}{Open problem}[section]
\newtheorem*{proof sketch}{$Proof\;sketch$}
\newtheorem{lemma}{Lemma}[section]
\newtheorem{proposition}{Proposition}[section]

\newtheorem{conj}{Conjecture}[section]
\begin{document}
\unbtitlepage
\setcounter{secnumdepth}{3} \setcounter{tocdepth}{3}
\pagenumbering{roman} \setcounter{page}{1}
%%-------------Abstract-----------------
\doublespacing
\chapter*{Abstract}
\addcontentsline{toc}{chapter}{Abstract} The statistical rank tests play important roles in univariate non-parametric data analysis. If one attempts to generalize the rank tests to a multivariate case, the problem of defining a multivariate order will occur. It is not clear how to define a multivariate order or statistical rank in a meaningful way. One approach to overcome this problem is to use the notion of data depth which measures the centrality of a point with respect to a given data set. In other words, a data depth can be applied to indicate how deep a point is located with respect to a given data set. Using data depth, a multivariate order can be defined by ordering the data points according to their depth values. Various notions of data depth have been introduced over the last decades. In this thesis, we discuss three depth functions: two well-known depth functions halfspace depth and simplicial depth, and one recently defined depth function named as $\beta$-skeleton depth, $\beta\geq 1$. The $\beta$-skeleton depth is equivalent to the previously defined spherical depth and lens depth when $\beta=1$ and $\beta=2$, respectively. Our main focus in this thesis is to explore the geometric and algorithmic aspects of $\beta$-skeleton depth. For the geometric part, we study the arrangement of circles and lenses as the influence regions of $\beta$-skeleton depth. The combinatorial complexity of such arrangement is computed in this study. For the algorithmic part, we develop some algorithms to solve the problem of computing the planar $\beta$-skeleton depth and halfspace depth. The lower bound for the complexity of computation of planar $\beta$-skeleton depth is also investigated. Finally, we determine some relationships among different depth functions, and propose a method to approximate one depth function using another one. The obtained theoretical results are supported by some experimental results provided in the last chapter of this thesis.

%% -----------Dedication----------------
\chapter*{Dedication}
\addcontentsline{toc}{chapter}{Dedication}This thesis is dedicated to the memory of:
\begin{itemize}
\item my beloved sister, Masoumeh.
\item my cousins, Mohsen and Alireza who were my best friends. 
\end{itemize}

%%-----------Acknowledgements---------------
\chapter*{Acknowledgements}
\addcontentsline{toc}{chapter}{Acknowledgements}
First of all, I wish to express my profound gratitude and appreciation to my supervisor, Dr. David Bremner, for his guidance, support and patience during my Ph.D studies. In the numberless meetings that we had together, he always took time to discuss our approaches and problems. I am always grateful to David for all that he did for me. I would also like to thank Diane Souvaine in Tufts University, Martin and Erik Demaine in MIT, and Reza Modarres in George Washington University for hosting me as a visiting student. Additionally, I thank people in CCCG2017 and CCCG2018 for useful discussions regarding the approximation of depth functions. Special thanks goes to all of my family members in Iran and Canada for always loving, encouraging, and supporting me. I am extremely fortunate for having such family. Finally, I would like to acknowledge my thesis committee members Dr. Patricia Evans, Dr. Huajie Zhang, Dr. Suprio Ray, Dr. Jeffrey Picka, and my external reviewer Dr. Pat Morin for reading this thesis carefully, and providing me with detailed and useful comments.

%%-----------Table of Contents------------------
\renewcommand{\contentsname}{Table of Contents}
\tableofcontents{}
\addcontentsline{toc}{chapter}{Table of Contents}
%%------------List of Tables----------------------
\newpage
\listoftables{}
\addcontentsline{toc}{chapter}{List of Tables}
%%------------List of Figures----------------------
\listoffigures{}
\addcontentsline{toc}{chapter}{List of Figures}
%%-----------List of Symbols, Nomenclature or Abbreviation--------
\chapter*{List of Symbols, Nomenclature or
Abbreviations} \addcontentsline{toc}{chapter}{Abbreviations}
\begin{table}[!ht]
\begin{tabular}{|l|l|}
\hline
$D$ & a data depth function\\
$D(q;S)$ &depth value of point $q$ with respect to data set $S$\\
$\Conv$ & convex hull\\
$I$ & indicator function\\
$inf$ & infimum\\
$int$ & interior\\
$H$ & halfspace\\
$\mathbb{H}$ & class of closed halfspaces\\
$\HD$ & halfspace depth \\
$\SD$ & simplicial depth\\
$\SkD_{\beta}$ & $\beta$-skeleton depth\\
$\SphD$ & spherical depth\\
$\LD$ & lens depth\\
$S_{\beta}(x_i,x_j)$ & $\beta$-influence region corresponding to points $x_i$ and $x_j$\\
$\Sph(x_i,x_j)$ & spherical influence region corresponding to points $x_i$ and $x_j$\\
$L(x_i,x_j)$ & lens influence region corresponding to points $x_i$ and $x_j$\\
$CC$& combinatorial complexity\\
$GG$ & Gabriel graph\\
$\RNG$ & relative neighbourhood graph\\
\hline

\end{tabular}
\end{table}

%%-------------change single space to double space--------
\doublespacing \pagenumbering{arabic} \setcounter{page}{1}
%%-----------Chapters start-----------------------
%%-----------Chapter 1----------------------------
\chapter{Introduction}
\label{ch:intro}
For a univariate data set of $n$ points with unknown distribution, consider the problem of computing a point (location estimator) which best summarizes the data set. One may answer this problem by introducing the mean of data set as the desired point. This answer is acceptable as the points are some minimum relative distance apart. However, in general, it is not the best choice because it is enough to move one of the data points very far away from the rest. This causes the mean to follow such single point, but not the majority of the data points. From this example, it can be deduced that the high level of robustness is an important characteristic that a location estimator should have \cite{basu1998robust}. This characteristic of an estimator can be measured by a factor called the \emph{breakdown point}, which is the portion of data points that can move to infinity before the estimator does the same \cite{lopuhaa1992highly,donoho1983notion}. By this definition, the breakdown point of the mean is $1/n$, whereas the breakdown point of the median is $1/2$. It is proven that the maximum breakdown point for any location estimator is $1/2$ \cite{lopuhaa1991breakdown}. As such, median is an appropriate location estimator for an ordered univariate data set \cite{bassett1991equivariant}.
\\\\If one attempts to generalize the concept of median to higher dimensions, an additional problem will occur. In this case, it is not clear how to define a multivariate order that can be applied to compute the median of data set. One approach is to use the notion of data depth which we study in this thesis.
\section{General Definition of Data Depth}
Generally speaking, a data depth is a measure in non-parametric multivariate data analysis that indicates how deep (central) a point is located with respect to a given data set in $\mathbb{R}^d$ ($d \in \mathbb{Z}^+$). In other words, data depth introduces a partial order relation on $\mathbb{R}^d$ because it assigns a corresponding rank (the depth of point with respect to a given data set) to every point in $\mathbb{R}^d$. As a result, applying a data depth on a data set generates a partial ordered set (\emph{poset})\footnote{A poset is a set together with a partial ordering relation which is reflexive, antisymmetric and transitive.} of the data points. Considering different depth values, each data depth determines a family of \emph{regions}. Each region contains all points in $\mathbb{R}^d$ with the same depth values. Among all regions, the \emph{central} region also known as the \emph{median set} is the one whose depth is maximum. Inside the central region a point in $\mathbb{R}^d$ (not necessarily from the data set) with the largest depth is called the \emph{median} of the data set.
\\\\As discussed above, defining a multivariate order and generalizing the concept of median are not straightforward. Over the last few decades, various notions of data depth have been introduced as powerful tools in non-parametric multivariate data analysis. Most of them have been defined to solve specific problems. A short list of the most important and well-known depth functions is as follows: \emph{halfspace depth}~\cite{hotelling1990stability,small1990survey,tukey1975mathematics}, \emph{simplicial depth}~\cite{liu1990notion}, \emph{Oja depth}~\cite{oja1983descriptive}, \emph{regression depth}~\cite{rousseeuw1999regression}, \emph{majority depth}~\cite{liu1993quality}, \emph{Mahalanobis depth}~\cite{mahalanobis1936generalized}, \emph{projection depth}~\cite{zuo2000general}, \emph{Zonoid depth}~\cite{dyckerhoff1996zonoid}, \emph{spherical depth}~\cite{elmore2006spherical}. These depth functions are different in application, definition, and geometry of their central regions. Each data depth has different properties and requires different time complexity to compute. In 2000, Zuo and Serfling \cite{zuo2000general} provided a general framework for statistical depth functions. In this framework, a data depth is a real valued function that possesses the following properties: \emph{affine invariance}, \emph{maximality at the center}, \emph{monotonicity on rays}, and \emph{vanishing at infinity}. However depending on the particular application, not every depth function needs to fit in this framework. For example, in a medical study that includes patient data such as height and weight perhaps the affine invariance is not an important requirement because the height and weight axes are meaningful. In fact, we need the data depth used in this medical study to be invariant under scaling of the axes. Given this property, the results would be independent from the measuring systems.
\\\\The concept of data depth is widely studied by statisticians and computational geometers. Some~directions that have been considered by researchers include defining new depth functions, developing new algorithms, improving the complexity of computations, computing both exact and approximate depth values, and computing depth functions in lower and higher dimensions. Two surveys by Aloupis \cite{aloupis2006geometric} and Small \cite{small1990survey} can be referred as overviews of data depth from a computational geometer's and a statistician's point of view, respectively. Some research on algorithmic aspects of planar depth functions can be found in~\cite{ aloupis2001computing,aloupis2003algorithms,bremner2008output,chan2004optimal,chen2013algorithms,christmann2006regression,liu2006data,matousek1991computing,rousseeuw1999regression}. 
\\\\The main focus of this thesis is on the geometric and algorithmic concepts of three depth functions: halfspace depth ($\HD$), simplicial depth ($\SD$), and $\beta$-skeleton depth ($\SkD_{\beta}$), briefly reviewed in Section \ref{sec:quick-review}. More detailed definitions and properties of these depth functions are explored in Chapter \ref{ch:different-DD}. The $\beta$-skeleton depth is not as well-known as the other depth functions. However, it is easy to compute even in higher dimensions.
\section{Quick Review of $\HD$, $\SD$, and $\SkD_{\beta}$}
\label{sec:quick-review}
In 1975, Tukey generalized the definition of univariate median and defined the halfspace median as a point in which the halfspace depth is maximized, where the halfspace depth is a multivariate measure of centrality of data points. Halfspace depth is also known as Tukey depth or location depth. In general, the halfspace depth of a query point $q$ with respect to a given data set $S$ is the smallest fraction of data points that are contained in a closed halfspace through $q$ \cite{barnett1976ordering,bremner2008output,shamos1977geometry,tukey1975mathematics}. The halfspace depth function has various properties such as vanishing at infinity, affine invariance, and decreasing along rays. These properties are proved in \cite{donoho1992breakdown}. Many different algorithms for the computation of halfspace depth in lower dimensions have been developed \cite{bremner2008output,bremner2006primal,chan2004optimal,rousseeuw1998computing}. The bivariate and trivariate case of halfspace depth can be computed exactly in $O(n\log n)$ and $O(n^2 \log n)$ time \cite{rousseeuw1996algorithm,struyf1999halfspace}, respectively. However, computing the halfspace depth of a query point with respect to a data set of size $n$ in dimension $d$ is an NP-hard problem if both $n$ and $d$ are part of the input \cite {johnson1978densest}. Due to the hardness of the problem, designing efficient algorithms to compute and approximate the halfspace
depth of a point remains an interesting task in the research area of data depth \cite{afshani2009approximate,aronov2010approximate,chen2013absolute,har2011relative}.
\\\\In 1990, another generalization of univariate median was presented by Liu \cite{liu1990notion}. This generalization is based on this fact that the median of a data set $S\subseteq \mathbb{R}$ is a point in $\mathbb{R}$ which is contained in the maximum number of closed intervals, formed by any pair of data points. Liu replaced the closed intervals by closed simplices\footnote{A simplex in $\mathbb{R}^1$ is a line segment, in $\mathbb{R}^2$ is a triangle, in $\mathbb{R}^3$ is a tetrahedron, etc.} in higher dimensions and presented the definition of \emph{simplicial median}. In some references (e.g. \cite{liu1995control}) open simplices are considered. In this thesis, we only consider closed simplices as originally used by Liu in the definition of simplicial median. For a data set $S\in \mathbb{R}^d$, the simplicial median set is a set of points in $\mathbb{R}^d$ which are contained in the maximum number of closed simplices formed by any $(d+1)$ data points from $S$. The definition of simplicial median provides another measure of centrality of point $q\in \mathbb{R}^d$ with respect to the data set $S \subseteq \mathbb{R}^d $. This measure is known as simplicial depth which is the proportion of all closed simplices obtained from $S$ that contain $q$. The straightforward algorithm to compute the simplicial depth takes $O(n^{d+1})$ time. The simplicial depth is affine invariant \cite{liu1990notion}. It is proved that the breakdown point of the simplicial median is worse than the breakdown point of the halfspace median \cite{chen1995bounds}. For a set of $n$ points in general position\footnote{A data set in $\mathbb{R}^2$ is in general position if no three points are collinear.} in $\mathbb{R}^2$, the depth of the simplicial median multiplied by $n \choose 3$ is $\Theta(n^3)$ \cite{aloupiscomputing}. The bivariate case of simplicial depth can be computed optimally in $\Theta(n\log n)$ time \cite{rousseeuw1996algorithm,aloupis2002lower}.
\\\\In 2006, Elmore, Hettmansperger, and Xuan~\cite{elmore2006spherical} defined another notion of data depth named \emph{spherical depth}. It is defined as the proportion of all closed hyperballs with the diameter $\overline{x_ix_j}$, where $x_i$ and $x_j$ are any pair of points in the given data set $S$. These closed hyperballs are known as influence regions of the spherical depth function. In 2011, Liu and Modarres~\cite{liu2011lens}, modified the definition of influence region, and defined \emph{lens depth}. Each lens depth influence region is defined as the intersection of two closed hyperballs $B(x_i,\Vert x_i,x_j\Vert )$ and $B(x_j,\Vert x_i,x_j\Vert )$. These influence regions of spherical depth and lens depth are the multidimensional generalization of \emph{Gabriel circles} and lunes in the definition of the \emph{Gabriel Graph}~\cite{gabriel1969new} and \emph{Relative Neighbourhood Graph}~\cite{supowit1983relative}, respectively. In 2017, Yang~\cite{yang2017beta}, generalized the definition of influence region, and introduced a familly of depth functions called $\beta$-skeleton depth, indexed by a single parameter $\beta\geq 1$. The influence region of $\beta$-skeleton depth is defined to be the intersection of two closed hyperballs given by $B(c_i,\frac{\beta}{2}\Vert x_i,x_j\Vert )$ and $B(c_j,\frac{\beta}{2}\Vert x_i,x_j\Vert )$, where $c_i$ and $c_j$ are some combinations of $x_i$ and $x_j$. Spherical depth and lens depth can be obtained from $\beta$-skeleton depth by considering $\beta=1$ and $\beta=2$, respectively. The $\beta$-skeleton depth has some nice properties including symmetry about the center, maximality at the centre, vanishing at infinity, and monotonicity. Depending on whether Euclidean distance or Mahalanobis distance is used to construct the influence regions, the $\beta$-skeleton depth can be orthogonally invariant or affinely invariant. All of these properties are explored in~\cite{elmore2006spherical,liu2011lens,yang2014depth,yang2017beta}.
A notable characteristic of the $\beta$-skeleton depth is that its time complexity is $O(dn^2)$ which grows linearly in the dimension $d$.
\section{Overview of this Thesis}
In this chapter we have introduced the notion of data depth, and reviewed the definitions of the halfspace depth, simplicial depth, and $\beta$-skeleton depth. In Chapter \ref{ch:bacground} we recall some concepts in computational geometry and statistics which are used in this thesis. Chapter \ref{ch:different-DD} includes a general framework for depth functions. The properties and the previous results related to the halfspace depth, simplicial depth, and $\beta$-skeleton depth are also explored in this chapter. The main contributions of this thesis are provided in the following chapters.
\\In Chapter \ref{ch:geometric}, we study the $\beta$-skeleton depth from a geometric perspectives. We obtain an exact bound $O(n^4)$ for the combinatorial complexity of the arrangement\footnote{An arrangement of some geometric objects is a subdivision of space formed by a collection of such objects.} of planar $\beta$-influence regions.
\\In chapter \ref{ch:algorithmic}, we present an optimal algorithm to compute the planar spherical depth (i.e. $\beta$-skeleton depth, $\beta=1$) of a query point. This algorithm takes $\Theta(n\log n)$ time. In Theorem \ref{thrm:q-skeleton}, we prove that computing the planar $\beta$-skelton depth of a query point can be reduced to a combination of at most $3n$ range counting problems. This reduction and the results in \cite{agarwal2013range,aronov2008approximating,har2011geometric,shaul2011range}, allow us to develop an algorithm for computing the planar $\beta$-skeleton depth. This algorithm takes $O(n^{(3/2)+\epsilon})$ query time. In the last section of this chapter, we present an $\Theta(n\log n)$ algorithm for computing the planar halfspace depth. Although there exist another algorithm with the same running time for this problem \cite{aloupiscomputing}, our algorithm is much simpler in terms of implementation. We use the specialized halfspace range query that is explored in Algorithm \ref{Alg:sph-pseudocode}.
\\Chapter \ref{ch:lowerbound} includes a proof of the lower bound for computing the planar $\beta$-skeleton depth for cases $\beta=1, 1<\beta<\infty$, and $\beta=\infty$, separately. In each case, we reduce the problem of Element Uniqueness\footnote{Element Uniqueness problem: Given a set $A=\{a_1,...,a_n\}$ of real numbers, is there a pair of indices $i,j$ with $i \neq j$ such that $a_i = a_j$?} which takes $\Omega(n\log n)$ to the problem of computing the planar $\beta$-skeleton depth of a query point.
\\In Chapter \ref{ch:experiments}, we study the relationships among different depth notions in this thesis. These relationships are explored in two different ways. First, we use the geometric properties of the influence regions, and prove that the simplicial depth has an upper bound in terms of a constant factor of $\beta$-skeleton depth (i.e. $\SD\leq 3/2\SphD$). Secondly, for every pair of depth functions, we employ a fitting function to approximate one data depth using another one. For example, considering a certain amount of error, we approximate the planar halfspace depth by a polynomial function of planar $\beta$-skeleton depth, $\beta\to \infty$. In the last Section \ref{sec:exp}, some experimental results are provided to support the the theoretical results in this chapter.
\\Finally, Chapter \ref{ch:conclusion} is the conclusion of this thesis.%intro
%-----------Chapter 2----------------------------
\chapter{Background}
\label{ch:bacground}
In this chapter, we present some background for the later chapters of this thesis. The background is provided in two sections: computational geometry review (Section \ref{sec:CG-review}) and non-computational geometry review (Section \ref{sec:nonCG-review}). 
\section{Computational Geometry Review}
\label{sec:CG-review}
Some concepts in computational geometry which are frequently referred to in this thesis are briefly reviewed in this section.
\subsection{Model of Computation}
One of the fundamental steps in any study of algorithms is to specify the adopted model of computation. In particular, it is not possible to determine the complexity of an algorithm without knowing the primitive operations\footnote{The primitive operations are those whose costs are constant for the fixed-length operands.} employed in such algorithm. Choosing an appropriate model of computation is closely related to the nature of given problem. The problems in computational geometry are of different types. However, they may be generally categorized as: \textit{Subset Selection}, \textit{Computation}, and \textit{Decision}. Naturally, for each of subset selection and computation problems, there is a decision problem. As an example, a YES/NO answer is needed for the question: "For a given set $S$, does the the set $S'\subset S$ satisfies a certain property $P$?". It can be discussed that we need to deal with the real numbers to answer some problems in the above categories. Suppose, for example, that in a given set of points with integer coordinates, the distance between a pair is requested. As such, we need a \textit{random-access machine} in which each storage location is capable to hold a single real number. The related model of computation is known as the \textit{real \textbf{RAM}} \cite{aho1974ullman,preparata2012computational}. The primitive operations with unit cost in this model are as follows:
\begin{itemize}
\item arithmetic operations ($+$,$-$,$\times$,$\div$).
\item comparison between two real numbers ($<$,$\leq$,$=$,$\neq$,$\geq$,$>$).
\item indirect addressing of memory
\item $k$-th root, trigonometric functions, exponential functions, and logarithmic functions.
\end{itemize} 
The execution of a decision making algorithm in the real \textbf{RAM} can be described as a sequence of operations of two types: arithmetic and comparisons. Depending on the outcome of the arithmetic, the algorithm has a branching option at each comparison. Hence, the computation which is executed by the real \textbf{RAM} can be considered as a path in a rooted tree. In the literature, this rooted tree is known as the \emph{algebraic decision tree} model of computation \cite{dobkin1979complexity,rabin1972proving,reingold1972optimality,steele1982lower}.    
An algebraic decision tree on a set of variables $\{x_1,...,x_k\}$ is a program with a finite sequence of statements. Each of these statements is obtained based on the result of $f(x_1,...,x_k)\diamond 0$, where $f(x_1,...,x_k)$ is an algebraic function and $\diamond$ denotes any comparison relation. For a set of $n$ real elements, Ben-Or proved that any algebraic computation tree that solves the element uniqueness problem must have the complexity of at least $\Omega(n\log n)$ \cite{ben1983lower}.
\subsection{Arrangements}
Suppose that $\Gamma=\{\gamma_1,...,\gamma_n\}$ is a finite collection of geometric objects in $\mathbb{R}^d$. The arrangement $\mathcal{A}(\Gamma$) is the subdivision of $\mathbb{R}^d$ into connected cells of dimensions $0,1,...,d$ induced by $\Gamma$, where each $k$-dimensional cell is a maximal connected subset of $\mathbb{R}^d$ determined by a fixed number of $\gamma_i\in \Gamma$. The arrangement $\mathcal{A}(\Gamma)$ is planar if $\gamma_i\subset\mathbb{R}^2$. In the planar arrangement $\mathcal{A}(\Gamma)$, a $0$-dimensional cell is a \emph{vertex}, a $1$-dimensional cell is an \emph{edge}, and a $2$-dimensional cell is a \emph{face}. As an illustration, Figure \ref{fig:lin-arrang} represents the arrangement of $5$ given lines in the plane.
\begin{figure}[!ht]
  \centering
    \includegraphics[width=0.6\textwidth]{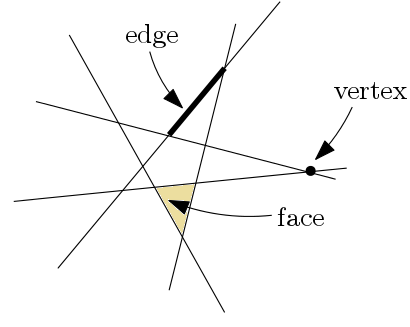}
  \caption{The planar arrangement of $5$ lines }
  \label{fig:lin-arrang}
\end{figure}
In studying arrangements, one of the fundamental questions that needs to be answered is how complex each arrangement can be. Computing the combinatorial complexity ($CC$) of each arrangement helps to answer this question, where the $CC$ of an arrangement $\mathcal{A}(\Gamma)$ is the total number of cells of $\mathcal{A}(\Gamma)$. For example, in Figure \ref{fig:lin-arrang}, $CC=42$ because there are $25$ edges, $10$ vertices, and $7$ faces in the arrangement. Every planar arrangement $\mathcal{A}(\Gamma)$ can be presented by a planar graph $G(V,E)$ such that $v_i\in V$ if and only if face $f_i\in \mathcal{A}(\Gamma)$. Furthermore, $e_{ij}\in E$ if and only if $f_i$ and $f_j$ are two adjacent faces of $\mathcal{A}(\Gamma)$. We recall Euler's Formula \cite{felsner2012geometric,ivanov2009making} which is one of the most useful facts regarding a planar graph.
\begin{theorem}
\label{thrm:Euler}
\emph{Euler's Formula} If $G$ is a planar graph with $C_v$ vertices, $C_e$ edges, and $C_f$ faces,
\begin{equation*}
C_v-C_e+C_f=2.
\end{equation*}
\end{theorem}
Because we study depth functions in this thesis, we define a nonstandard terminology that is frequently referred to in Chapter \ref{ch:geometric}.
\begin{definition}
In a planar arrangement $\mathcal{A}(\Gamma)$, we define a depth region to be the maximal connected union of faces with the same depth value. 
\end{definition}
Further and more detailed results on the arrangements of hyperplanes, spheres, simplices, line segments can be found in \cite{agarwal2000arrangements,berg2008computational,goodman1991discrete}.    
\subsection{Range Query}
Range query is among the central problems in computational geometry. In fact, many problems in computational geometry can be represented as a range query problem. In this thesis, computing the depth value of a query point in some cases is converted to a range query problem (see Chapter \ref{ch:algorithmic}). 
In a typical range query problem, we are given a data set $S$ of $n$ points, and a family $\Lambda$ of \emph{ranges} (i.e. subsets of $\mathbb{R}^d$). $S$ should be preprocessed into a data structure such that for a query range $\tau \in \Lambda$, all points in $S\cap\tau$ can be efficiently reported or counted. In addition to \emph{range counting} and \emph{range reporting}, another type of range query problems is \emph{range emptiness} (i.e. checking if $S\cap\tau=\emptyset$). Due to the nature of related problems in data depth among the three types of range query problems, range reporting is not used in this thesis. Typical example of ranges in range query problems include halfspaces, simplices, rectangle. However, in many applications we may need to deal with ranges bounded by nonlinear functions. In other words, the ranges are some \emph{semialgebraic sets}\footnote{A semialgebraic set is a subset of $\mathbb{R}^d$ obtained from a finite Boolean combination of $d$-variate polynomial equations and inequalities.}. In this case, the problem is known as semialgebraic range query \cite{agarwal1999geometric,agarwal2013range,goodman1991discrete,shaul2011range}. The recent results on some range query problems are reviewed in Chapter \ref{ch:algorithmic}.
\\In some problems such as halfspace and simplex range query, solving  the exact range counting is expensive which means that no exact algorithm with both near linear storage complexity and preprocessing time, and low query time exists \cite{afshani2008geometric}. This issue is a motivation for seeking some approximation techniques that can be applied to approximate a range counting problem. One way to achieve such technique is through the notion of $(\varepsilon,\tau)$-approximation, which is defined in the following. 
\begin{definition}  
\label{def:epsilon-approx}
Given a set $S=\{x_1,..., x_n\}\subset \mathbb{R}^2$, a parameter $\varepsilon>0$, and a semialgebraic range $\tau$ with constant description complexity, we say that algorithm $A$ computes an \textit{$(\varepsilon,\tau)$-approximation} $\eta_\tau$ if \begin{equation*}
(1-\varepsilon) \vert S \cap \tau \vert\leq \eta_\tau \leq (1+\varepsilon) \vert S \cap \tau \vert.
\end{equation*}
\end{definition}
\section{Non Computational Geometry Review}
\label{sec:nonCG-review}
In this section, we provide some background regarding the concepts of partially ordered set (poset), fitting function, and the measures of goodness of statistical models. These concepts are used in Sections \ref{sec:dissimilarity-measures} and \ref{sec:approx-half-space}, where we approximate depth functions using the idea of a fitting function.
\subsection{Poset}
For a finite set $S=\{x_1,...,x_n\}$, it is said that $P=(S,\preceq)$ is a poset if $\preceq$ is a partial order relation on $S$, that is, for all $x_i,x_j,x_t \in S$:
\begin{itemize}
\item[$i)$] $x_i \preceq x_i $.
\item[$ii)$] $x_i \preceq x_j$ and $x_j \preceq x_t$ implies that $x_i \preceq x_t$.
\item[$iii)$] $x_i \preceq x_j$ and $x_j \preceq x_i$ implies that $x_i \equiv_p x_j$, where $\equiv_p$ is the corresponding equivalence relation.
\end{itemize}
Poset $P=(S,\preceq)$ is called a chain if any two elements of $S$ are comparable, i.e., given $x_i ,x_j\in S$, either $x_i \preceq x_j $ or $x_j \preceq x_i $. If there is no comparable pair among the elements of $S$, the corresponding poset is an anti chain. Figure~\ref{fig:poset} illustrates different posets with the same elements.
\begin{figure}[!ht]
  \centering
    \includegraphics[width=0.8\textwidth]{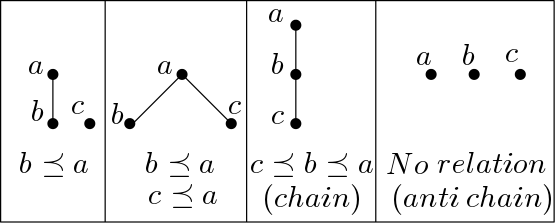}
  \caption{Different posets and relations among their elements}
  \label{fig:poset}
\end{figure}

\subsection{Fitting Function}
One method to represent a set of data is to use the notion of fitting function. This representation technique involves choosing values for the parameters in a function to best describe the data set. Depending on whether the number of parameters is known or unknown in advance, the method is called \emph{parametric fitting} or \emph{nonparametric fitting}, respectively \cite{gershenfeld1999nature}. In \emph {depth function approximation} that we study in this thesis, the number of parameters is known; therefore, we focus on the parametric fitting. Consider a set $\omega=\{(v_i,u_i);1\leq i\leq n\}$, where $V=(v_1,...,v_n)$ and $U=(u_1,...,u_n)$ are two given vectors. The goal is to obtain a function $f$ that best describe $\omega$ (i.e. $f$ captures the trend among the elements of $\omega$). For every element in $\omega$, $u_i=f(v_i)\pm\delta_i$, where $\delta_i\geq 0$ is the error corresponding to the approximation of $u_i$.
\\Obviously, there may be more than one candidate for $f$. The question is how to measure the goodness of fit in order to obtain the best fitting function.
\subsection{Evaluation of Fitting Function}
\label{sec:goodness-of-fitting}
In this section we review some methods and criteria that can be employed to measure how good a fitting function is. Based on the values of these methods, we can select the best fitting function for our data points.
\subsubsection{Coefficient of Determination}
For two vectors $U=(u_1,...,u_n)$ and $V=(v_1,...,v_n)$, suppose that $f$ is a function such that $u_i=f(v_i)\pm \delta_i$. Let $\xi_i=u_i-\overline{U}$, where $\overline{U}$ is the average of all $u_i; 1\leq i\leq n$. The coefficient of determination \cite{shafer2012beginning} of $f$ is indicated by $R^2$ (or $r^2$), and defined by:
\begin{equation}
\label{eq:r-squared}
r^2=\dfrac{\sum\limits_{i=1}^{n}(\xi_i^2-\delta_i^2)}{\sum\limits_{i=1}^{n}\xi_i^2}.
\end{equation} 
To avoid confusion between $\mathbb{R}^2$ and $R^2$, hereafter in this thesis, the coefficient of determination is represented by $r^2$ alone. In statistics and data analysis, $r^2$ is applied as a measure that assesses the ability of a model in predicting the real value of a parameter. From Equation \eqref{eq:r-squared}, it can be deduced that the value of $r^2$ is always within the interval $[0,1]$. Generally, a high value of $r^2$ related to a model fitted to a set of data indicates that the model is a good fit for such data. Obviously, the interpretations of fit depend on the context of analysis. For example, a model with $r^2=0.9$ explains $90\%$ of the variability of the response data around its mean. This percentage may be considered as a high value in a social study. However, it is not a high enough value in some other research areas (e.g.\;biochemistry, chemistry, physics), where the value of $r^2$ could be much closer to $100$ percent. 

\subsubsection{$\AIC$ and $\BIC$}
In statistics, \emph{Akaike Information Criterion ($\AIC$)} \cite{akaike1973maximum} and the \emph{Bayesian information criterion ($\BIC$)} \cite{schwarz1978estimating} are two commonly used criteria to select the best model among a collection of candidates that fit to a given data set. These two criteria provide a standardized way to achieve a balance between the goodness of fit and the simplicity of the model. The model with the lowest value of $\AIC$ (and/or the lowest value of $\BIC$) is preferred. Among a finite number of candidate models $M_k$ fitted to a data set $S$, both AIC and BIC involve choosing the model with the best penalized log-likelihood. The likelihood function corresponding to each model $M_k$ is a conditional probability given by $p(S|P_k)$, where $P_k$ is the parameter vector of $M_k$. In the literature, it is more common to use the log-likelihood $\mathfrak{L}(P_k)=\log p(S|P_k)$ \cite{myung2016model}. Considering the above notations, the $\AIC$ and $\BIC$ values of every model can be computed using the following equations.
\begin{align}
\label{eq:AIC}
\AIC=&-2\mathfrak{L}(P_k)+2\vert P_k \vert\\
\label{eq:BIC}
\BIC=&-2\mathfrak{L}(P_k)+ \vert P_k \vert\log(n),
\end{align}
where $\vert P_k \vert$ represents the number of parameters of $M_k$ and $n$ is the size of data set $S$. From Equations \eqref{eq:AIC} and \eqref{eq:BIC}, it can be seen that $\BIC$ penalizes the model complexity more heavily. In particular, $\BIC$ prefers the models with less parameters as the size of $S$ increases, whereas $\AIC$ does not penalize the model based on the size of $S$.

%background
%-----------Chapter3----------------------------
\chapter{Properties of Data Depth}
\label{ch:different-DD}
In this chapter we review general properties of a depth function. We also discuss different notions of data depth such as halfspace depth, simplicial depth, and $\beta$-skeleton depth that are studied in this thesis.
\section{General Framework}
\label{sec:framwork}
A data depth $D$, $(D:(\mathbb{R}^d,\mathbb{R}^{n\times d})\to\mathbb{R}^+)$ is a real-valued function defined at any arbitrary point $q \in \mathbb{R}^d $ with respect to a given data set $S\subset \mathbb{R}^d$. A typical depth function $D$ satisfies the following conditions which are known as the general framework for data depth, introduced in \cite{zuo2000general},.

\begin{itemize}
\item \textbf{Affine invariance} \cite{donoho1992breakdown,tanton2005encyclopedia}: For invertible matrix $A$ and constant vector $b$, $D$ is affinely invariant if 

\begin{equation*}
D(Aq+b;AS+b)= D(q;S); A \in \mathbb{R}^{d\times d}, b\in \mathbb{R}^d.
\end{equation*}

See Figure~\ref{affine}. If this equation holds for any orthogonal matrix $A$ (i.e.$A^TA=AA^T=I$), $D$ is \emph{orthogonally invariant} which is weaker than affine invariant. Transforming data points is commonly used in processing a data set; therefore, it is important for a data depth to be invariant in some sense.

\begin{figure}[!ht]
  \centering
    \includegraphics[width=0.65\textwidth]{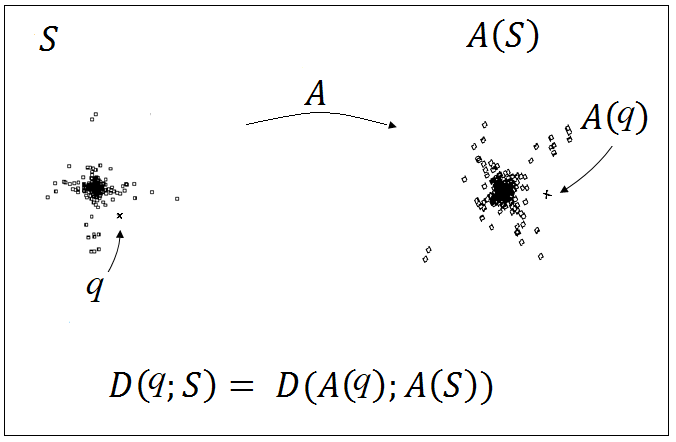}
  \caption{$D$ is invariant to affine transformation $A$}
  \label{affine}
\end{figure}

\item \textbf{Vanishing at infinity} \cite{liu1990notion,zuo2000general}: $D$ vanishes at infinity if for every sequence $\{q_k\}_{k\in\mathbb{N}}$ with $\lim_{k \to \infty}|| q_k||=\infty$, 
\begin{equation*}
\lim_{k \to \infty} D(q_{k};S)= 0.
\end{equation*}
Figure~\ref{vanish} is a visualization of this property. 

\begin{figure}[!ht]
  \centering
    \includegraphics[width=0.65\textwidth]{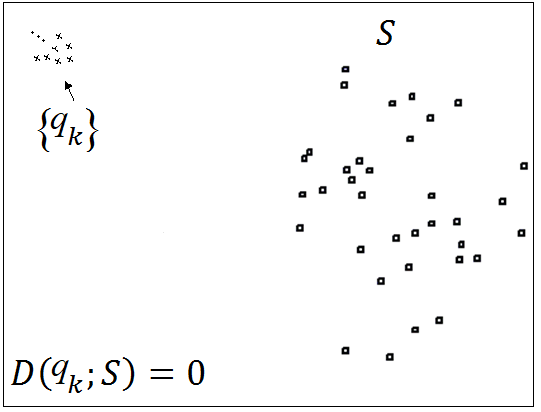}
  \caption{$D$ vanishes at infinity}
  \label{vanish}
\end{figure}

\item \textbf{Monotone on rays} \cite{bassett1991equivariant,lopuhaa1991breakdown,liu1993quality}: For $t\in \mathbb{R}^d$ as a center point (with maximal depth value), $D$ is monotone on the rays if

\begin{equation*}
D(q;S)\leq D(p;S)=D(t+r(q-t);S); r\in [0,1],
\end{equation*}

where the point $p$ is a convex combination of $q$ and $t$. As an illustration, see Figure~\ref{monotone}.

\begin{figure}[!ht]
  \centering
    \includegraphics[width=0.65\textwidth]{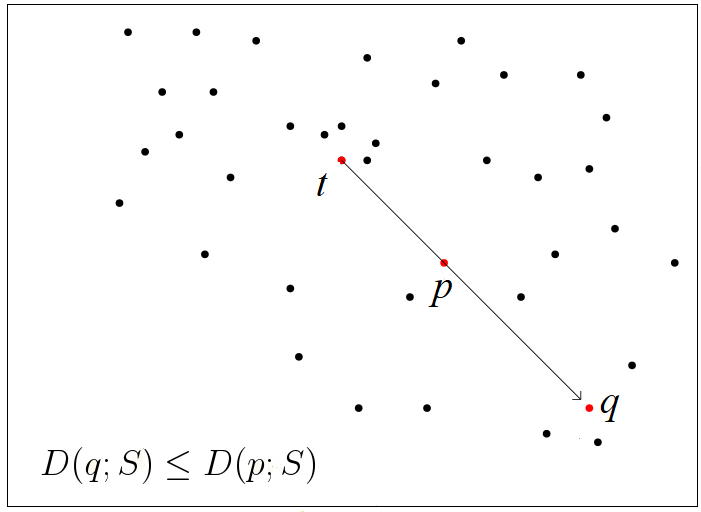}
  \caption{$D$ is monotone on the rays}
  \label{monotone}
\end{figure}

\item \textbf{Upper semi-continuity} \cite{durocher2009projection,mosler2012multivariate}: For $\alpha \in \mathbb{R}^+$, $D$ is upper semi-continuous if $D_\alpha(S)$ is a closed set, where

\begin{equation*}
D_\alpha(S)=\{q\in\mathbb{R}^d|D(q;S) \geq\alpha\}.
\end{equation*}

The outer boundary of $D_\alpha$ is called the \emph{contour} of depth $\alpha$ \cite{donoho1992breakdown}. See Figure~\ref{semicontinue}.

\begin{figure}[!ht]
  \centering
    \includegraphics[width=0.65\textwidth]{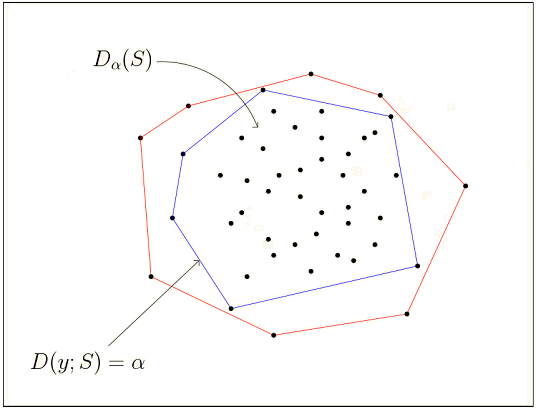}
  \caption{D is upper semi-continuous}
  \label{semicontinue}
\end{figure}
\end{itemize}
In addition to these four properties, some other conditions such as high breakdown point and level-convex are discussed for depth functions~\cite{donoho1992breakdown,donoho1983notion,lopuhaa1991breakdown}.
\begin{itemize}
\item \textbf{High breakdown point} \cite{donoho1983notion,lopuhaa1991breakdown}: The breakdown point of a location estimator $c=c(S)$ (i.e. the center point of $S=\{x_1,...,x_n\}$ in $\mathbb{R}^d$) is a number between zero and one, introducing the proportion of data points that must be moved to infinity before $c$ moves to infinity. In other words, the breakdown point $b(c)$ can be defined by
\begin{equation*}
b(c)= \min _{m}\left\{\frac{m}{n}; \sup_{S'}\Vert c(S)-c(S')\Vert=\infty \right\},
\end{equation*}
where $S'=(S\setminus S_m)\cup Y$, $S_m\subseteq S$, $\vert S_m\vert=m$, and $Y=\{y_1,...,y_m\}$ is an arbitrary set in $\mathbb{R}^d$.

\item \textbf{Level-convex} \cite{donoho1992breakdown}: Data depth $D$ is level-convex if all of its corresponding contours $D_{\alpha}$ are convex.
\end{itemize}

\section{Halfspace Depth}
\label{sec:halfspace}
The halfspace depth of a query point $q\in \mathbb{R}^d$ with respect to a given data set $S=\{x_1,...,x_n\} \subseteq \mathbb{R}^d$ is defined as the minimum portion of points of $S$ contained in any closed halfspace that has $q$ on its boundary. Using the notation of $\HD(q;S)$, the above definition can be presented by~\eqref{eq:halfspacedef}.
\begin{equation}
\label{eq:halfspacedef}
\HD (q;S)= \frac{2}{n}\min \{|S \cap H|: H\in \mathbb{H}\},
\end{equation}
where $2/n$ is the normalization factor\footnote{Instead of the normalization factor $1/n$ which is common in literature, we use $2/n$ in order to let the halfspace depth of $1$ to be achievable.}, $\mathbb{H}$ is the class of all closed halfspaces in $\mathbb{R}^d$ that pass through $q$, and $|S\cap H|$ denotes the number of points within the intersection of $S$ and $H$. As illustrated in Figure~\ref{fig:halfspace}, $\HD(q_1;S)=6/13$ and $\HD(q_2;S)=0$, where $S$ is a given set of points in the plane and $q_1,q_2$ are two query points not in $S$.
\begin{figure}[!ht]
  \centering
    \includegraphics[width=0.6\textwidth]{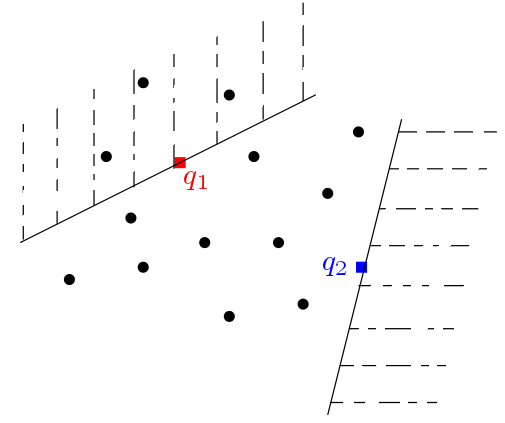}
  \caption{Two examples of halfspace depth in the plane}
  \label{fig:halfspace}
\end{figure}

For a given data set $S$ and a query point $q$, it can be verified that $\HD(q;S)$ is equal to zero if and only if $q$ is outside the convex hull\footnote{The convex hull of a data set $S \subset \mathbb{R}^d$ is the smallest convex set in $\mathbb{R}^d$ that contains $S$.} of $S$ (see Figure~\ref{fig:halfspace-depth0}).
\begin{figure}[!ht]
\centering
\includegraphics[width=0.6\textwidth]{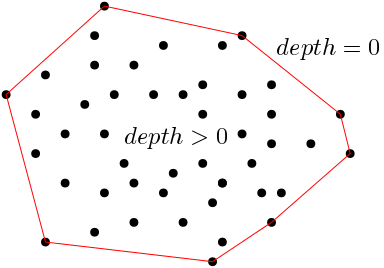}
\caption{halfspace depth}
\label{fig:halfspace-depth0}
\end{figure}
The halfspace depth satisfies all desirable properties of depth functions presented in Section \ref{sec:framwork} \cite{tukey1975mathematics,donoho1992breakdown}. Furthermore, for a data set in general position\footnote{It is said that a data set, in $\mathbb{R}^d$, is in general position if no $(d+1)$ points of the data points lie on a common hyperplane.} in $\mathbb{R}^d$, the breakdown point of the halfspace median is at least $\frac{1}{d+1}$, and at most $1/3$ for $d\geq 2$ \cite{donoho1992breakdown,donoho1982breakdown}. Another nice property of the halfspace depth is that the depth contours are all convex and nested, i.e.\;the contour of unnormalized depth $(k+1)$ is convex, and geometrically surrounded by the contour of unnormalized depth $k$ which is also convex \cite{donoho1992breakdown}.
\\\\Among all data depths, halfspace depth, perhaps, is the most extensively studied data depth. Many algorithms to compute, or approximate, the halfspace median have been developed in recent years \cite{afshani2009approximate, bremner2006primal, chan2004optimal, cuesta2008random, langerman2003optimization, matousek1991computing, rousseeuw1998computing, wilcox2003approximating}.
A summary of these results is as follows.
\begin{itemize}
\item A complicated $O(n^2\log n)$ algorithm to compute the halfspace depth of a point in $\mathbb{R}^2$ is implemented by Rousseeuw and Ruts in \cite{rousseeuw1998constructing}. %[RR98].
\item An optimal randomized algorithm for computing the halfspace median is developed by Chan \cite{chan2004optimal}. This algorithm requires $O(n\log n)$ expected time for $n$ non-degenerate data points.

\item An $O(n\log^5n)$ algorithm to compute the bivariate halfspace median is presented by Matousek in \cite{matousek1991computing}. The algorithm consists of two main steps: an $O(n\log^4n)$ algorithm to compute any point with depth of at least $k$, and a binary search on $k$ to find the median.% [.*]=[Mat91]
\item By improving the Matousek's algorithm, Langerman provided an algorithm which computes the bivariate halfspace median in $O(n\log^3n)$~\cite{langerman2001algorithms}.%[Lan01] in Ga thesis 
\item  For a set on $n$ non-degenerate points in $\mathbb{R}^d$, the halfspace median can be computed in $O(n \log n+n^{d-1})$ expected time (Theorem $3.2$ in \cite{chan2004optimal}).
\item The halfspace depth of $q\in \mathbb{R}^d$ with respect to $S\subseteq \mathbb{R}^d$ can be computed in $O((d+1)^k LP(n,d-1))$ time, where $k$ is the value of the output and $LP(n,d)$ denotes the running time for solving a linear program with $n$ constraints and $d$ variables (Theorem $5$ in \cite{bremner2008output}).
\item In a worst case scenario in $\mathbb{R}^2$, when the data points all are on the convex hull, it takes at least $O(n^2)$ to find all of the halfspace depth contours \cite{miller2001fast}. %[MRR+01].
\item A center point of $S\subseteq \mathbb{R}^d$,  a point whose halfspace depth is at least $[\frac{n}{d+1}]/n$, can be computed in $O(n)$ \cite{burr2011dynamic}. 
\item Computing the halfspace of a query point in $\mathbb{R}^2$ can be done in $\Theta(n\log n)$ time \cite{aloupiscomputing,aloupis2002lower}.
\end{itemize}

\section{Simplicial Depth} 
\label{sec:simplicial}
The simplicial depth of a query point $q\in \mathbb{R}^d$ with respect to $S=\{x_1,...,x_n\}\subset\mathbb{R}^d$ is defined as the total number of the closed simplices formed by data points that contain $q$. This definition can be given by \eqref{eq:simplicial}. 
\begin{equation}
\label{eq:simplicial}
\SD(q; S)= \frac{1}{{n \choose d+1}}\sum \nolimits_{(x_1,...,x_{d+1}) \in {n \choose d+1}} {I(q \in \Conv[x_1,...,x_{d+1}])},
\end{equation}
where $1/{n \choose d+1}$ is the normalization factor, the convex hull $\Conv[x_1,...,x_{d+1}]$ is a closed simplex formed by $d+1$ points of $S$, and $I$ is the indicator function. For $S=\{a,b,c,d,e\}$ in $\mathbb{R}^2$, Figure \ref{fig:simplicial} illustrates that $\SD(q_1;S)=3/10$ and $\SD(q_2;S)=4/10$.
\begin{figure}[!ht]
  \centering
    \includegraphics[width=0.8\textwidth]{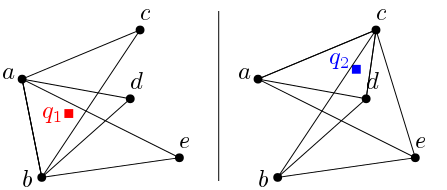}
  \caption{Two examples of simplicial depth in the plane}
  \label{fig:simplicial}
\end{figure}
\\Liu \cite{liu1995control} proved that the simplicial depth satisfies the affine invariance condition. Depending on the distribution of data points, the simplicial depth has completely different characteristics. For a \emph{Lebesgue-continuous} distribution \cite{hazewinkel2013encyclopaedia}, the simplicial depth changes continuously (Theorem $2$ in \cite{liu1990notion}), decreases monotonously on the rays, and has a unique central region \cite{liu1990notion,mosler2013depth}. Furthermore, the contours defined by simplicial depth are nested (Theorem $3$ in \cite{liu1990notion}). However, if the distribution is discrete, these characteristics are not necessary applicable \cite{zuo2000general}. Unlike the contours of halfspace depth which are convex, the contours of the simplicial depth are only starshaped (Section $2.3.3$ in \cite{mosler2013depth}). It has been proven that the breakdown point of the simplicial median is always worse than the breakdown point of halfspace median \cite{chen1995bounds}. The behaviour of simplicial depth contours is not as nice as the behaviour of half-space depth contours \cite{he1997convergence, miller2001fast}. As an example, Figure \ref{fig:simp-contour-not-nested} illustrates that the simplicial depth contours may not be nested. It can be seen that the contour enclosing all points of depth $10/20$ and up is not surrounded by the contour enclosing depth of $8/20$ and up.
\\\\Simplicial depth is widely studied in the literature. Some of the results regarding simplicial depth can be listed as follows.
\begin{itemize}
\item The simplicial depth of a query point in $\mathbb{R}^2$ can be computed using an optimal algorithm which takes $\Theta (n\log n)$ time \cite{aloupiscomputing,aloupis2002lower, langerman2000complexity}.
\item The simplicial depth of a point in $\mathbb{R}^3$ can be computed in $O(n^2)$, and in $\mathbb{R}^4$ the fastest known algorithm needs $O(n^4)$ time \cite{cheng2001algorithms}.
In the higher dimension $d$, no better algorithm is known than the brute force method with $O(n^{d+1})$ time.
\item The maximum simplicial depth in $\mathbb{R}^2$ can be computed in $O(n^4 \log(n))$~\cite{aloupiscomputing}.
\end{itemize}
\begin{figure}[!ht]
  \centering
    \includegraphics[width=0.65\textwidth]{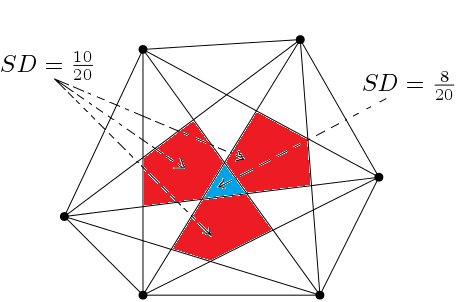}
  \caption{Simplicial depth contours}
  \label{fig:simp-contour-not-nested}
\end{figure}

\section{$\beta$-skeleton Depth}
\label{sec:beta-skeleton}
To introduce the $\beta$-skeleton depth, first, we need to define the $\beta$-influence region.
\begin{definition}
\label{def:beta-influence}
For $1 \leq \beta \leq \infty$, the $\beta$-influence region of vectors $x_i$ and $x_j$ ($S_{\beta}(x_i, x_j)$) is defined as follows:
\begin{equation}
\label{eq:beta-influence}
S_{\beta}(x_i, x_j)= B(c_i,r) \cap B(c_j,r),
\end{equation}
where $r=\frac{\beta}{2}\Vert x_i - x_j\Vert$, $c_i=\frac{\beta}{2}x_i + (1-\frac{\beta}{2})x_j$,  and $c_j=(1-\frac{\beta}{2})x_i + \frac{\beta}{2}x_j$. In the case of $\beta=\infty$, the $\beta$-skeleton influence region is defined as a slab determined by two halfspaces perpendicular to the line segment $\overline{x_ix_t}$ at the end points. Figure~\ref{fig:betasphericallens} shows the $\beta$-skeleton influence regions for different values of $\beta$.
\end{definition}
\begin{note}
For $1 \leq \beta \leq \infty$ in literature, the ball based version of $S_{\beta}(x_i, x_j)$ is also defined. In this case, the $S_{\beta}(x_i, x_j)$ is given by the union of the balls, instead of the intersection of them in equation \ref{eq:beta-influence}. For example, the hatching area in Figure \ref{fig:betasphericallens} denotes the ball based version of the $S_{2}(x_1, x_2)$. Since the definition of the $\beta$-skeleton depth is given based on the lune based $S_{\beta}(x_i, x_j)$ alone \cite{Modarress2015beta}, by $S_{\beta}(x_i, x_j)$ we only mean its lune based version hereafter in this study.  
\end{note}
  
\begin{figure}[!ht]
  \centering
    \includegraphics[width=0.8\textwidth]{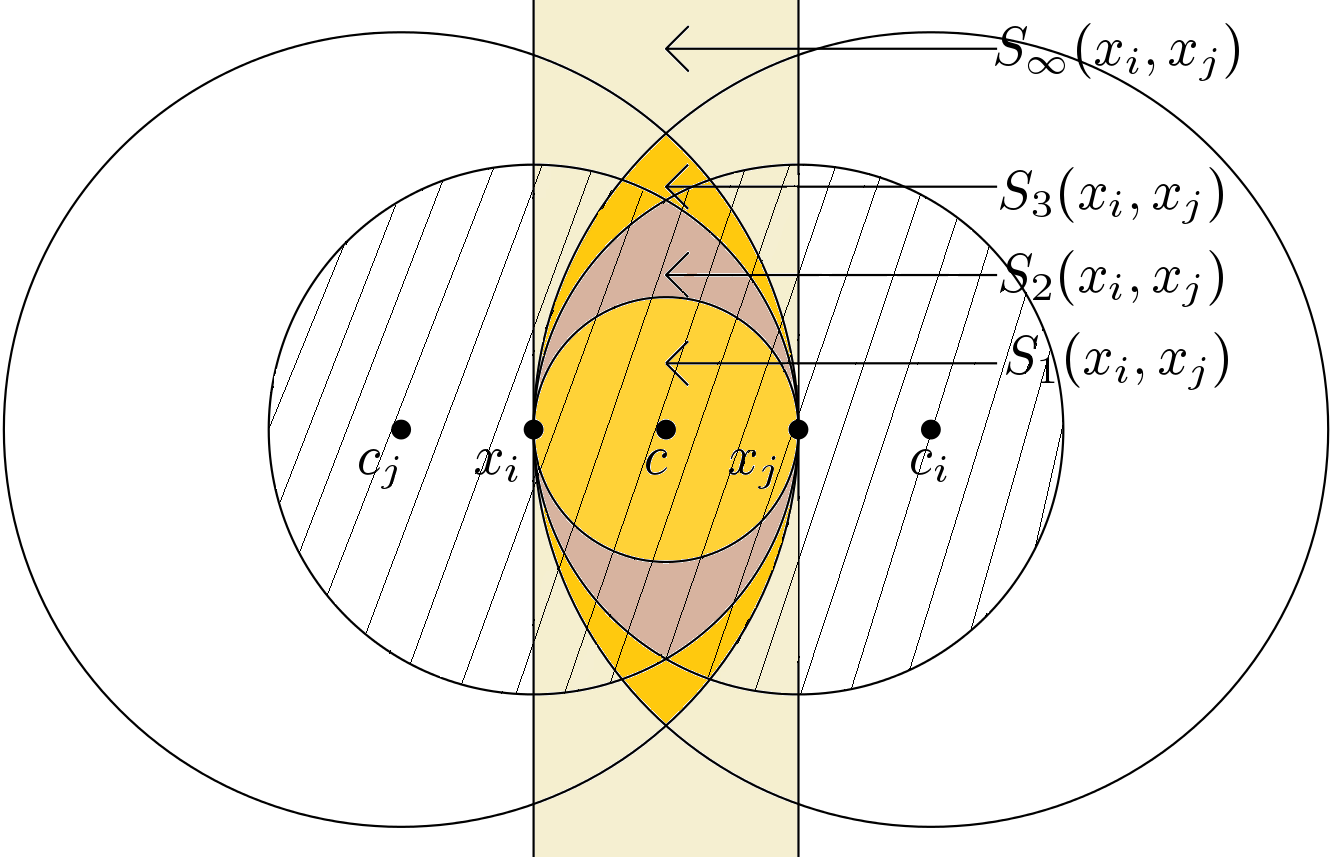}
  \caption{The $\beta$-influence regions defined by $x_i$ and $x_j$ for $\beta$=1, 2, 3, and $\infty$.}
  \label{fig:betasphericallens}
\end{figure}
\begin{definition}
For parameter $1\leq \beta\leq \infty$ and $S=\{x_1, ..., x_n\}\subset \mathbb{R}^d$, the $\beta$-skeleton depth of a query point $q \in \mathbb{R}^d$ with respect to $S$, is defined as the proportion of the $\beta$-influence regions that contain $q$. Using notation $\SkD_{\beta}$ for $beta$-skeleton depth, this definition can be presented by Equation \eqref{eq:beta}.
\begin{equation}
\label{eq:beta}
\SkD_{\beta}(q;S)= \frac{1}{{n \choose 2}}\sum_{1\leq i<j\leq n} {I(q \in S_{\beta}(x_i, x_j)}), 
\end{equation}
where $1/{n \choose 2}$ is the normalization factor.
\end{definition}
It can be verified that $q \in S_{\beta}(x_i, x_j)$ is equivalent to the inequality of $\frac{\beta}{2}\Vert x_i - x_j \Vert \geq \max \{ \Vert q - c_i \Vert , \Vert q - c_j \Vert\}$. The straightforward algorithm for computing the $\beta$-skeleton depth of $q \in \mathbb{R}^d$ takes $O(d n^2)$ time because the inequality should be checked for all $1\leq i,j\leq n$ \cite{yang2014depth, liu2011lens}.
The $\beta$-skeleton depth is a family of statistical depth functions  including the spherical depth when $\beta=1$ \cite{elmore2006spherical}, and the lens depth when $\beta=2$ \cite{liu2011lens}.
\\\\It is proved that the $\beta$-skeleton depth functions satisfy the data depth framework provided by Zuo and Serfling \cite{zuo2000general} because these depth functions are monotonic (Theorem $3$, \cite{yang2014depth}), maximized at the center, and vanishing at infinity (Theorem $5$, \cite{yang2014depth}). The $\beta$-skeleton depth functions are also orthogonally (affinely) invariant if the Euclidean (Mahalanobis) distance is used to construct the influence regions of $\beta$-skeleton depth influence regions (Theorem $1$, \cite{yang2014depth}). The breakdown point of $\beta$-skeleton median is at least $1-(1/\sqrt{2})$ (Theorem $12$, \cite{yang2014depth}). Regarding the geometric properties of $\beta$-skeleton depth, some of the results are as follows. The depth regions with the same depth value are not necessarily connected (see the regions with the depth of $5/6$ in Figures \ref{fig:sph-distinct-local-colored} and \ref{fig:lens-distinct-local-colored}). However, the depth regions are nested (Lemma $3$, \cite{yang2014depth}) which means that the contour with depth of $\alpha$ is geometrically surrounded by the contour with depth of $\gamma$, where $\alpha > \gamma$. For example, in Figure \ref{fig:all-beta-3points}, the contour with the depth of $2/3$ is surrounded by the contour with the depth of $1/3$. Another property is that the only depth regions which are convex are the central regions (see Theorem~\ref{thrm:cenral-region-convex}). We explore some other geometric properties related to the $\beta$-skeleton depth in Chapter \ref{ch:geometric}.
\begin{figure}[!ht]
  \centering
    \includegraphics[width=0.9\textwidth]{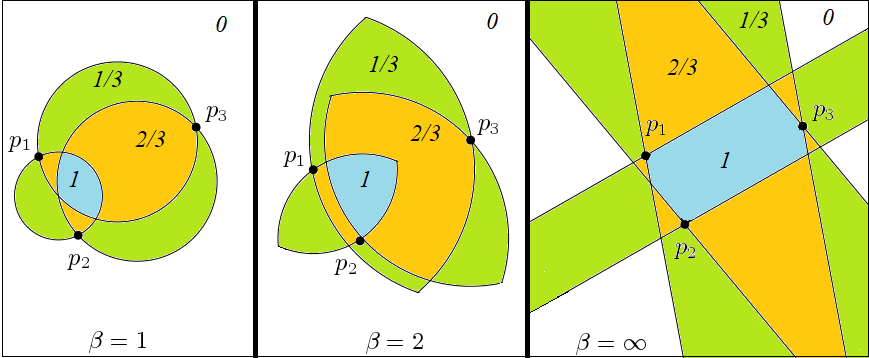}
  \caption{Partitions of the plane by $\beta$-skeleton depth with respect to $\{p_1, p_2, p_3\}$.}
  \label{fig:all-beta-3points}
\end{figure}
\subsection{Spherical Depth and Lens Depth}
\label{sec:sph-len}
As discussed above, the $\beta$-skeleton depth includes the spherical depth when $\beta=1$, and the lens depth when $\beta=2$. From  Equations~\eqref{eq:beta-influence} and~\eqref{eq:beta}, the definitions of spherical depth ($\SphD$) and lens depth ($\LD$) of a query point $q$ with respect to a given data set $S$ in $\mathbb{R}^d$ are as follows:
\begin{equation}
\label{eq:spherical}
\SphD(q;S)=\SkD_1(q;S)= \frac{1}{{n \choose 2}}\sum_{1\leq i<j\leq n} {I(q \in \Sph(x_i, x_j)})
\end{equation}
\begin{equation}
\label{eq:lens}
\LD(q;S)=\SkD_2(q;S)= \frac{1}{{n \choose 2}}\sum_{1\leq i<j\leq n} {I(q \in L(x_i, x_j)}),
\end{equation}
where the influence regions $\Sph(x_i,x_j)$ and $L(x_i,x_j)$ are equivalent to $S_1(x_i,x_j)$ and  $S_2(x_i,x_j)$, respectively. The influence region of the spherical depth is also know as the \emph{Gabriel sphere}.
%literature review
%%----------Chapter4--------------------------
\chapter{Geometric Results in $\mathbb{R}^2$}
\label{ch:geometric}
In this chapter we discuss the $\beta$-skeleton depth in $\mathbb{R}^2$ from a geometric point of view. The geometric results provide some guidance to the algorithms for computing $\beta$-skeleton depth. As an example, in Section \ref{sec:Combinatorial}, we explore that computing the entire arrangement of $\beta$-influence regions in $\mathbb{R}^2$ is not an efficient approach to compute the planar $\beta$-skeleton depth. Given a set of collinear points in $\mathbb{R}^2$, we compute the combinatorial complexity ($CC$) of the arrangement of $\beta$-influence regions. Some geometric properties of the $\beta$-skeleton depth are also explored in this chapter. 

\section{Combinatorial Complexity}\label{sec:Combinatorial}
For the $\beta$-influence regions obtained from a set of collinear points in $\mathbb{R}^2$, we present exact bounds for the number of edges, faces, and vertices in the corresponding arrangement.

\begin{definition}
The $CC$ of an arrangement in $\mathbb{R}^2$ is equal to the total number of faces, edges, and vertices (intersection points and data points) in the arrangement.
\end{definition}

\begin{exmp}
For a set of four collinear points $p_1, p_2, p_3,$ and $p_4$ in $\mathbb{R}^2$, consider the arrangement of corresponding $\beta$-influence regions ($1\leq \beta <\infty$). The $CC$ of this arrangement is equal to $34$ because the arrangement includes the total number of $12$ faces, $16$ edges, and $6$ vertices. For the case of $\beta=1$, see Figure \ref{fig:4collinearpoints}.  
\end{exmp}

\begin{figure}[!ht]
  \centering
    \includegraphics[width=0.6\textwidth]{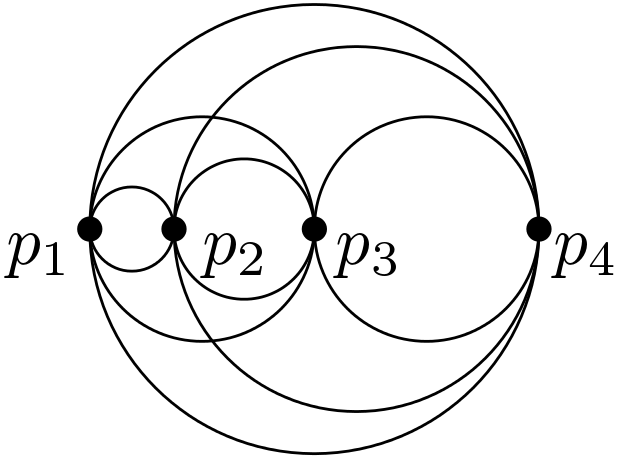}
  \caption{Arrangement of $\beta$-influence regions ($\beta=1$) for collinear points}
  \label{fig:4collinearpoints}
\end{figure}

\begin{theorem}
\label{thrm:Cob-cmplx-collinear-bound}
For the arrangement of all Gabriel circles obtained from $n$ distinct collinear points in $\mathbb{R}^2$,
\begin{equation*}
CC=\frac{1}{4}n^4 - \Theta(n^3).
\end{equation*}
\end{theorem}

\begin{proof}
We construct the arrangement of Gabriel circles incrementally, and define some strategies to count the number of faces, edges, and vertices in the arrangement. Starting with the two leftmost points in the data set, we have one Gabriel circle to consider. In this step, there are only two faces (inside the circle and outside the circle), two edges, and two vertices. Henceforward we add the data points one by one from left to right. We count the number of created faces, edges, and vertices after adding each Gabriel circle obtained from the new point and any previously added data points. We write the numbers of new faces, edges, and vertices in rows corresponding to faces, edges, and vertices, respectively. The new cells are obtained from the intersection of recently added circle and the cells in the previous arrangement. Finally, it is enough to sum up all obtained corresponding numbers in order to get the total number of all faces, edges, and vertices. These strategies are represented in the triangular forms in Tables \ref{tbl:triangle-faces}, \ref{tbl:triangle-edges}, and \ref{tbl:triangle-vertices}. These representations help to obtain a general formula for every element of the tables. Figure \ref{fig:collinear-conf} illustrates how to obtain the numbers for $n=3$ from the previous step (see the rows $n=2$ and $n=3$ in Tables \ref{tbl:triangle-faces}, \ref{tbl:triangle-edges}, and \ref{tbl:triangle-vertices}).
Note that we respectively define $f_k(n)$, $e_k(n)$, and $v_k(n)$ to be the number of recently created faces, edges, and vertices after including the new Gabriel circle $S_1(n^{th} point,\;(n-k)^{th} point)$ in the previously updated arrangement.

\begin{figure}[!ht]
  \centering
    \includegraphics[width=0.8\textwidth]{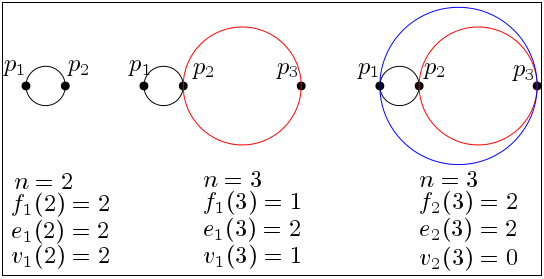}
  \caption{Visualization of rows $n=2$ and $n=3$ of Table \ref{tbl:triangle-faces}, Table \ref{tbl:triangle-edges}, and Table \ref{tbl:triangle-vertices}}
  \label{fig:collinear-conf}
\end{figure}

\begin{table}[!ht]
\centering
\begin{tabular}{|l||cccccccccccccc|}
\hline
$|S|$&\multicolumn{14}{|c|}{Faces added per data point in previous arrangement}\\
\hline
\hline
$2$&&&&&&&&2&&&&&&\\
$3$&&&&&&&1&&2&&&&&\\
$4$&&&&&&1&&4&&2&&&&\\
$5$&&&&&1&&6&&6&&2&&&\\
$6$&&&&1&&8&&10&&8&&2&&\\
$7$&&&1&&10&&14&&14&&10&&2&\\
$8$&&1&&12&&18&&20&&18&&12&&2\\
$n$&$f_{1}(n)$&&&...&&&&$f_{k}(n)$&&&...&&&$f_{n-1}(n)$\\
\hline
\end{tabular}
\caption{The number of faces in the incremental construction of the arrangement of Gabriel circles obtained from collinear data points.}
\label{tbl:triangle-faces}
\end{table}

\begin{table}[!ht]
\centering
\begin{tabular}{|l||cccccccccccccc|}
\hline
$|S|$&\multicolumn{14}{|c|}{Vertices added per data point in previous arrangement}\\
\hline
\hline
$2$&&&&&&&&2&&&&&&\\
$3$&&&&&&&2&&2&&&&&\\
$4$&&&&&&2&&4&&2&&&&\\
$5$&&&&&2&&6&&6&&2&&&\\
$6$&&&&2&&8&&10&&8&&2&&\\
$7$&&&2&&10&&14&&14&&10&&2&\\
$8$&&2&&12&&18&&20&&18&&12&&2\\
$n$&$e_{1}(n)$&&&...&&&&$e_{k}(n)$&&&...&&&$e_{n-1}(n)$\\
\hline
\end{tabular}
\caption{The number of edges in the incremental construction of the arrangement of Gabriel circles obtained from collinear data points.}
\label{tbl:triangle-edges}
\end{table}

\begin{table}[!ht]
\centering
\begin{tabular}{|l||cccccccccccccc|}
\hline
$|S|$&\multicolumn{14}{|c|}{Edges added per data point in previous arrangement}\\
\hline
\hline
$2$&&&&&&&&2&&&&&&\\
$3$&&&&&&&1&&0&&&&&\\
$4$&&&&&&1&&2&&0&&&&\\
$5$&&&&&1&&4&&4&&0&&&\\
$6$&&&&1&&6&&8&&6&&0&&\\
$7$&&&1&&8&&12&&12&&8&&0&\\
$8$&&1&&10&&16&&18&&16&&10&&0\\
$n$&$v_{1}(n)$&&&...&&&&$v_{k}(n)$&&&...&&&$v_{n-1}(n)$\\
\hline
\end{tabular}
\caption{The number of vertices in the incremental construction of the arrangement of Gabriel circles obtained from collinear data points.}
\label{tbl:triangle-vertices}
\end{table}
We define a function $\delta^r_{ij}$ which helps us to formulate $f_{k}(n)$, $e_{k}(n)$, and $v_{k}(n)$ in Tables \ref{tbl:triangle-faces}, \ref{tbl:triangle-edges}, and \ref{tbl:triangle-vertices}, respectively.
%\varsigma^r_{ij}=i[\frac{i\mod j}{i}]+r[\frac{i}{j}]+(i\mod j)
\[
\delta^r_{ij}=
\begin{cases}
r &\text{if $i=j$}\\
    i &\text{if $i<j$}\\
    0 &\text{otherwise}
\end{cases}
\]
The elements $f_{k}(n)$, $e_{k}(n)$, and $v_{k}(n)$ can be presented by following equations.
\begin{equation}
\label{eq:Fk}
f_{k}(n)=
\begin{cases}
f_{k}(n-1)+2\delta^1_{(k-1)(n-2)} &\text{if $2\leq k\leq n-1$ and $n\geq 3$}\\
    2 &\text{if $k=1$ and $n=2$}\\    
    1 &\text{if $k=1$ and $n\geq 3$}\\
    0 & \text{otherwise}
\end{cases}
\end{equation}
\begin{equation}
\label{eq:Ek}
e_{k}(n)=
\begin{cases}
e_{k}(n-1)+2\delta^1_{(k-1)(n-2)} &\text{if $2\leq k\leq n-1$ and $n\geq 3$}\\
    2 &\text{if $k=1$ and $n=2$}\\    
    2 &\text{if $k=1$ and $n\geq 3$}\\
    0 & \text{otherwise}
\end{cases}
\end{equation}
\begin{equation}
\label{eq:Vk}
v_{k}(n)=
\begin{cases}
v_{k}(n-1)+2\delta^0_{(k-1)(n-2)} &\text{if $2\leq k\leq n-1$ and $n\geq 3$}\\
    2 &\text{if $k=1$ and $n=2$}\\    
    1 &\text{if $k=1$ and $n\geq 3$}\\
    0 & \text{otherwise}
\end{cases}
\end{equation}
We employ the Telescoping Substitution to solve the recurrences in Equations \eqref{eq:Fk}, \eqref{eq:Ek}, and \eqref{eq:Vk} as follows. For $2\leq k \leq n-1$ and $n\geq 3$,
\begin{align*}
f_k(n)&=f_{k}(n-1)+2\delta^1_{(k-1)(n-2)}\\&=f_{k}(n-2)+2\delta^1_{(k-1)(n-3)}+2\delta^1_{(k-1)(n-2)}\\&=f_{k}(n-3)+2\delta^1_{(k-1)(n-4)}+2\delta^1_{(k-1)(n-3)}+2\delta^1_{(k-1)(n-2)}\\&\vdots\\&=f_k(k)+2\left(\delta^1_{(k-1)(k-1)}+\delta^1_{(k-1)(k)}+...+\delta^1_{(k-1)(n-3)}+\delta^1_{(k-1)(n-2)}\right)\\&=0+2\left(1+\sum_{i=1}^{n-k-1}\delta^1_{(k-1)(k+i-1)}\right)=2\left(1+(k-1)(n-k-1)\right).
\end{align*} 
Therefore,
\begin{equation}
\label{eq:Fk'}
f_{k}(n)=
\begin{cases}
2(1+(k-1)(n-k-1)) &\text{if $2\leq k\leq n-1$ and $n\geq 3$}\\
    2 &\text{if $k=1$ and $n=2$}\\    
    1 &\text{if $k=1$ and $n\geq 3$}\\
    0 & \text{otherwise}.
\end{cases}
\end{equation}
Similarly, we can solve $e_k(n)$ and $v_k(n)$ and obtain the following relations.
\begin{equation}
\label{eq:Ek'}
e_{k}(n)=
\begin{cases}
2(1+(k-1)(n-k-1)) &\text{if $2\leq k\leq n-1$ and $n\geq 3$}\\
    2 &\text{if $k=1$ and $n=2$}\\    
    2 &\text{if $k=1$ and $n\geq 3$}\\
    0 & \text{otherwise}
\end{cases}
\end{equation}
\begin{equation}
\label{eq:Vk'}
v_{k}(n)=
\begin{cases}
2(k-1)(n-k-1) &\text{if $2\leq k\leq n-1$ and $n\geq 3$}\\
    2 &\text{if $k=1$ and $n=2$}\\    
    1 &\text{if $k=1$ and $n\geq 3$}\\
    0 & \text{otherwise}
\end{cases}\hspace{1cm}
\end{equation}
To compute the $CC$ of the arrangement, we need to calculate $C_f$, $C_e$, and $C_v$ representing the total numbers of faces, edges, and vertices, respectively. These values can be similarly computed using Equations \eqref{eq:Fk'}, \eqref{eq:Ek'}, and \eqref{eq:Vk'}, respectively. We provide the computation of $C_f$ in the following. However, to avoid repetitions, we omit the computations of $C_e$ and $C_v$, and present only their final values.

\begin{align*}
C_f &=\sum \limits_{k,n}f_{k}(n)= 2+ \sum\limits_{m=3}^n (1+ \sum\limits_{k=2}^{m-1} f_{k}(m))
\\&=2+ \sum\limits_{m=3}^n (1+ \sum\limits_{k=2}^{m-1} 2((k-1)(m-k-1)+1)))
\\&=2+(n-2)+2\sum\limits_{m=3}^n \sum\limits_{k=2}^{m-1}(k-1)(m-k-1)+2\sum\limits_{m=3}^n \sum \limits_{k=2}^{m-1}1
\\&=n+\frac{2}{6}\sum\limits_{m=3}^n(m^3-6m^2+11m-6)+2\sum\limits_{m=3}^n(m-2)
\\&=n+(\frac{2}{6})(\frac{1}{4})(n^4-6n^3+11n^2-6n)+2\frac{(n-2)(n-1)}{2}
\\&=\frac{1}{12}\left(n^4-6n^3+23n^2-30n+24\right).
\end{align*}
\begin{equation*}
C_e=\sum \limits_{k,n}e_{k}(n)= 2+ \sum\limits_{m=3}^n (2+ \sum\limits_{k=2}^{m-1} e_{k}(m))=\frac{1}{12}(n^4-6n^3+23n^2-18n).
\end{equation*}
\begin{equation*}
C_v =\sum \limits_{k,n}v_{k}(n)=2+ \sum\limits_{m=3}^n (1+ \sum\limits_{k=2}^{m-1}v_{k}(m))=\frac{1}{12}(n^4-6n^3+11n^2+6n).
\end{equation*} 
The proof is complete because
\begin{equation*}
\label{eq:comb.cplx.colliear}
CC=C_f+C_e+C_v=\frac{1}{4}(n^4-6n^3+19n^2-14n+8)=\frac{1}{4}n^4-\Theta(n^3).
\end{equation*}
\end{proof}
\begin{note}
Recalling that the Gabriel circles and $\beta$-influence regions are equivalent when $\beta=1$, one can generalize Lemma \ref{thrm:Cob-cmplx-collinear-bound} by considering $\beta$-influence regions ($1\leq\beta<\infty$) instead of Gabriel circles. This generalization can be made because the corresponding combinatorics does not change for $\beta$-influence regions if $1\leq\beta<\infty$. However, the case of $\beta=\infty$ does not follow this generalization. In this case, $CC=\Theta(n)$ because $C_f=n$ (including one face outside all edges and $n-1$ faces among the edges), $C_e=n$ (including $n$ distinct parallel lines passing through data points), and $C_v=n+2$ (including $n$ data points and $2$ intersection points at infinity). We recall that for $n$ distinct collinear points in $\mathbb{R}^2$, the $\beta$-influence regions form some parallel slabs if $\beta=\infty$.
\end{note}
\begin{lemma}
\label{lm:influence-cut-4points}
Every distinct pair of $\beta$-influence regions cut each other in at most $4$ points. 
\end{lemma}
\begin{proof}
Every $\beta$-influence region can be represented by a pair of circular arcs. For $\beta=1$, these circular arcs are some half circles, whereas for $\beta>1$, they are smaller than half circles. We prove the lemma by considering two cases: $\beta=1$ and $\beta>1$. For the case $\beta=1$, the proof is trivial because two distinct $1$-influence regions (i.e. circles) cut each other in at most two points. For the case $\beta>1$, suppose that $S_{\beta}(a,b)$ and $S_{\beta}(c,d)$ are two arbitrary and distinct $\beta$-influence regions. Figure \ref{fig:influence-cut-4points} is an illustration of the $\beta$-influence regions ($\beta=3/2$) and their corresponding circular arcs. Suppose that $S_{\beta}(a,b)$, as a pair of circular arcs $EaF$ and $FbE$, arbitrarily cuts $S_{\beta}(c,d)$. Each one of $EaF$ and $FbE$ is smaller than its corresponding half circle. This implies that, after cutting the boundary of $S_{\beta}(c,d)$ in at most two points, none of $EaF$ and $FbE$ would turn back towards the boundary of $S_{\beta}(c,d)$. As such, two $\beta$-influence regions $S_{\beta}(a,b)$ and $S_{\beta}(c,d)$ cut each other in at most $4$ points.
\begin{figure}
\centering
\includegraphics[width=0.5\textwidth]{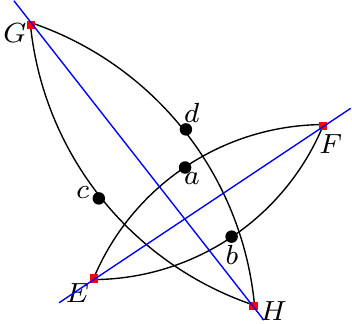}
\caption{Two arbitrarily crossed $\beta$-influence regions, $\beta=3/2$.}
\label{fig:influence-cut-4points}
\end{figure}
   
\end{proof}
\begin{lemma} The trivial upper bound for the $CC$ of the arrangement of all $\beta$-influence regions obtained from $n$ arbitrarily distributed data points in $\mathbb{R}^2$ is $O(n^4)$.
\label{lm:combinatorial-complexity-bound}
\end{lemma}
\begin{proof}
From Lemma~\ref{lm:influence-cut-4points}, the boundaries of every two $\beta$-influence regions intersect each other in at most $4$ points. It means that we have at most $4{{n\choose 2} \choose 2}$ intersection points in the arrangement of the $n\choose 2$ $\beta$-influence regions. Since every arrangement is a representation of a planar graph, we can use Euler's Formula $C_v-C_e+C_f=2$ (Theorem \ref{thrm:Euler}) in planar graphs to compute the number of faces, edges, and vertices in the arrangement. The number of vertices and the number of edges in the planar graphs are related to each other by the inequality $C_e\leq 3C_v-6$ (see Theorem $9.7$ in \cite{allenby2011count}). Considering the intersection points from the above discussion and the number of data points, $C_v$ and consequently $C_e$, $C_f$, and $CC$ can be computed as follows:
\begin{align*}
C_v&\leq 4{{n\choose 2} \choose 2}+n=\frac{1}{2}(n^4-2n^3-n^2+4n)\\
C_e&\leq 3C_v-6\leq \frac{3}{2}(n^4-2n^3-n^2+4n)-6=\frac{3}{2}(n^4-2n^3-n^2+4n-4)\\
C_f&=2+C_e-C_v\leq 2+(2C_v-6)=2C_v-4\leq n^4-2n^3-n^2+4n-4\\
CC&=C_f+C_e+C_v\leq (3n^4-6n^3-3n^2-12n-10)\in O(n^4).
\end{align*}
\end{proof}
\begin{note}
Lemma \ref{thrm:Cob-cmplx-collinear-bound} indicates that the trivial upper bound $O(n^4)$ for the $CC$ related to the arrangement of $\beta$-influence regions is achievable. 
\end{note}
\begin{conj}
\label{conj:minimum_intersection}
The collinear configuration of $n$ planar points minimizes the number of intersections (and consequently, edges and faces) among the corresponding $\beta$-influence regions.
\end{conj}
Proving Conjecture \ref{conj:minimum_intersection} would imply a lower bound for the $CC$ of the arrangement of $\beta$-influence regions for every arbitrary set of points in $\mathbb{R}^2$. Furthermore, it would also imply that the obtained upper bound in Lemma \ref{lm:combinatorial-complexity-bound} is optimal. Without loss of generality, it can be assumed that $\beta=1$.

\section{Geometric Properties of $\beta$-skeleton Depth}
In this section, some geometric properties of the $\beta$-skeleton depth are investigated. Among all of the $\beta$-influence regions, we only consider the case $\beta=1$ (i.e. Gabriel circles) related to the spherical depth. One can easily generalize the results to the other value of $\beta$.
\begin{figure}[!ht]
  \centering
    \includegraphics[width=0.7\textwidth]{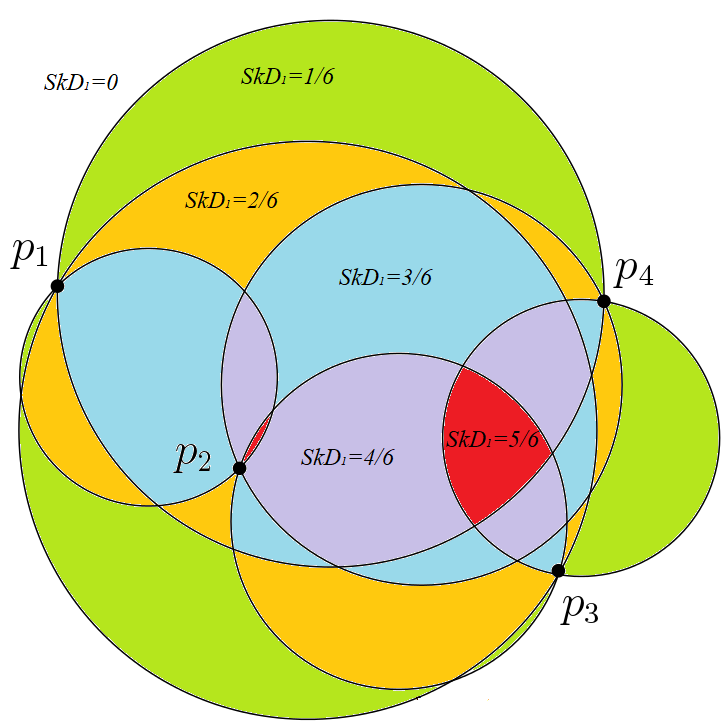}
  \caption{Spherical depth of planar points with respect to $\{p_1, p_2, p_3, p_4\}$}
    \label{fig:sph-distinct-local-colored}
\end{figure}
\begin{figure}[!ht]
\centering
    \includegraphics[width=0.7\textwidth]{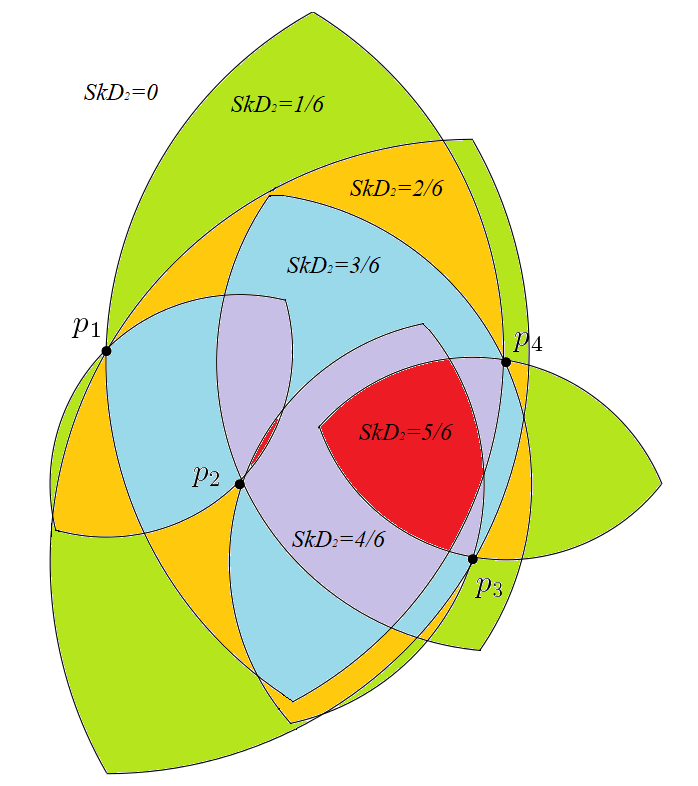}
  \caption{Lens depth of planar points with respect to $\{p_1, p_2, p_3, p_4\}$} 
  \label{fig:lens-distinct-local-colored}  
\end{figure}
\\As we discussed in Section \ref{sec:beta-skeleton}, similar to the property of halfspace depth in Section \ref{sec:halfspace}, the contours of $\beta$-skeleton depth are nested. However, we may have more than one contour per depth value. For example, in Figures \ref{fig:sph-distinct-local-colored} and \ref{fig:lens-distinct-local-colored}, it can be seen that there are two separate contours with the depth of $5/6$. If there exists more than one contour for a depth value, we have some locally deepest regions in the arrangement. Consequently, the central regions may not be connected (see Figures \ref{fig:sph-distinct-local-colored} and \ref{fig:lens-distinct-local-colored}). Another geometric property of $\beta$-skeleton depth is that the only convex depth regions are the locally deepest regions. These last two properties can be deduced from Lemma \ref{thrm:connectivity-in-points} and Theorem \ref{thrm:cenral-region-convex}.
\begin{lemma}
\label{thrm:connectivity-in-points}
For every arbitrary pair of neighboring faces\footnote{Two faces are neighbors if they have a common edge in their boundaries.} $A$ and $B$ in the arrangement of Gabriel circles obtained from a data set $S\subset \mathbb{R}^2$,
\begin{equation}
\label{eq:depth-neighboring-faces}
\forall(a\in A,b\in B); |\SkD_{1S}(a)-\SkD_{1S}(b)|\geq 1,
\end{equation}
where $\SkD_{1S}(a)={n\choose 2}\SkD_1(a;S)$.
\end{lemma}
\begin{proof}
The proof is immediate from the fact that every edge of the arrangement corresponds to some Gabriel circle $G$. One of the two faces bounded by this edge is entirely contained in $G$, and the other face is entirely outside of $G$.
\end{proof}
\begin{note}
Lemma \ref{thrm:connectivity-in-points} implies that the faces with the same depth can be connected only in discrete points.
\end{note}
\begin{note}
Considering two neighboring regions\footnote{Two regions are neighbors if they have a common border.} instead of two neighboring faces, one can generalize Lemma \ref{thrm:connectivity-in-points}. To obtain such generalization, it is enough to apply Lemma \ref{thrm:connectivity-in-points} on any arbitrary pair of neighboring faces in the neighboring regions.
\end{note} 
\begin{theorem}
\label{thrm:cenral-region-convex}
Consider the arrangement of Gabriel circles obtained from data set $S$. A face (region) in this arrangement is locally deepest if and only if it has no concave edge. In other words, a face is locally deepest if and only if it is a convex face.
\end{theorem}
\begin{proof}
$\Rightarrow$) Suppose that $A$ is a locally deepest face in the arrangement of Gabriel circles obtained from $S$. We prove that $A$ has no concave edge. To obtain a contradiction, assume that $e$ is a concave edge of $A$. Consequently, there exists a neighboring face $B$, in the arrangement, whose boundary contains $e$. Lemma \ref{thrm:connectivity-in-points} and the characteristics of concave edge in the arrangement of Gabriel circles imply that $\SkD_1(B;S)> \SkD_1(A;S)$. This result contradicts the assumption that $A$ is locally deepest.
\\$\Leftarrow$) Suppose that $A$ is a face that has no concave edge in the arrangement of Gabriel circles obtained from $S$. We prove that $A$ is locally deepest. Assume that the boundary of $A$ is composed of the convex edges $e_1,...e_k$, and thus $A$ has neighboring faces $B_1,...,B_m$ ($m\leq k$). Hence $A$ and every $B_i$ ($i=1,...,m$) have at least one common edge $e_j$ ($j=1,...,k$). From Lemma \ref{thrm:connectivity-in-points} and the characteristics of concave edge, it can be seen that 
\[ \forall B_i; \; \SkD_1(A;S) > \SkD_1(B_i;S).\]
This means that $A$ is a locally deepest because it is the deepest face among all of its neighboring faces $B_i$.
\end{proof}%geometric results
%%--------------------Chapter5-----------------
\chapter{Algorithmic Results in $\mathbb{R}^2$}
\label{ch:algorithmic}
The previous best algorithm for computing the $\beta$-skeleton depth of a point $q\in \mathbb{R}^d$ with respect to a data set $S=\{x_1..x_n\}\subseteq \mathbb{R}^d$ is the brute force algorithm \cite{yang2017beta}. This naive algorithm needs to check all of the $n \choose 2$ $\beta$-skeleton influence regions obtained from the data points to figure out how many of them contain $q$. Checking all of such influence regions causes the naive algorithm to take $\Theta(dn^2)$ time. In this chapter, we present an optimal algorithm for computing the planar spherical depth ($\beta=1$) and a subquadratic algorithm to compute the planar $\beta$-skeleton depth when $\beta>1$. In these algorithms, we need to solve some halfspace and some circular range counting problems, where all of the halfspaces have one common point. The circles also have the same characteristic. Furthermore, computing the planar $\beta$-skeleton depth is reduced to a combination of some range counting problems. In the special case $\beta=1$, we investigate a specialized halfspace range query method that leads to a $\Theta(n\log n)$ algorithm (Algorithm \ref{Alg:sph-pseudocode}) for $\beta$-skeleton depth. Finally, we present a simple and optimal algorithm (Algorithm \ref{Alg:halfspace-pseudocode}) that computes the planar halfspace depth in $\Theta(n\log n)$ time. This algorithm is similar to the Aloupis's algorithm in \cite{aloupiscomputing}. After sorting the points by angle, Aloupis employed the counterclockwise sweeping of a specific halfline. However, in our algorithm, we use the specialized halfspace range query that is explored in Algorithm \ref{Alg:sph-pseudocode}.
\section{Optimal Algorithm to Compute Planar Spherical Depth}
\label{sec:sph-alg}
Instead of checking all of the spherical  influence regions, we focus on the geometric aspects of such regions in $\mathbb{R}^2$. The geometric properties of these regions lead us to develop an $\Theta(n\log n)$ algorithm for the computation of planar spherical depth of $q\in \mathbb{R}^2$.

\begin{theorem}
For arbitrary points $a$, $b$, and $t$ in $\mathbb{R}^2$, $t \in \Sph(a,b)$ if and only if $\angle atb \geq \frac{\pi}{2}$.
\label{thrm:point-circle}
\end{theorem}

\begin{proof}
If $t$ is on the boundary of $\Sph(a,b)$, \emph{Thales' Theorem}\footnote{Thales' Theorem also known as the \emph{Inscribed Angle Theorem}: If $a$, $b$, and $c$ are points on a circle where $\overline{ac}$ is a diameter of the circle, then $\angle abc$ is a right angle.} suffices as the proof in both directions. For the rest of the proof, by $t \in \Sph(a,b)$ we mean  $t\in int\Sph(a,b)$.    
\\$\Rightarrow$) For $t \in \Sph(a,b)$, suppose that $\angle atb < \frac{\pi}{2}$ (proof by contradiction). We continue the line segment $\overline{at}$ to cross the boundary of the $\Sph(a,b)$. Let $t'$ be the crossing point (see the left figure in Figure~\ref{fig:point-circle}). Since $\angle atb < \frac{\pi}{2}$, then, $\angle btt'$  is greater than $\frac{\pi}{2}$. Let $\angle btt'=\frac{\pi}{2}+\epsilon_1; \epsilon_{1}>0 $. From Thales' Theorem, we know that $\angle at'b$ is a right angle. The angle $tbt'= \epsilon_{2}>0$ because $t\in \Sph(a,b)$. Summing up the angles in $\bigtriangleup tt'b$, as computed in~\eqref{eq:point-circle-pi-in}, leads to a contradiction. So, this direction of proof is complete.

\begin{equation}
\angle tt'b + \angle t'bt + \angle btt'\geq \frac{\pi}{2}+\epsilon_{2}+ \frac{\pi}{2}+\epsilon_{1} >\pi
\label{eq:point-circle-pi-in} 
\end{equation}
\\$\Leftarrow$) If $\angle atb = \frac{\pi}{2}+\epsilon_{1}; \epsilon_{1}>0$, we prove that $t \in \Sph(a,b)$. Suppose that $t \notin \Sph(a,b)$ (proof by contradiction). Since $t \notin \Sph(a,b)$, at least one of the line segments $\overline{at}$ and $\overline{bt}$ crosses the boundary of $\Sph(a,b)$. Without loss of generality, assume that $\overline{at}$ is the one that crosses the boundary of $\Sph(a,b)$ at the point $t'$ (see the right figure in Figure~\ref{fig:point-circle}). Considering Thales' Theorem, we know that $\angle at'b=\frac{\pi}{2}$ and consequently, $\angle bt't=\frac{\pi}{2}$. The angle $\angle t'bt=\epsilon_{2}>0$ because $t \notin \Sph(a,b)$. If we sum up the angles in the triangle $\bigtriangleup tt'b$, the same contradiction as in~\eqref{eq:point-circle-pi-in} will be implied.
\end{proof}

\begin{figure*}[!ht]
  \centering
    \includegraphics[width=0.85\textwidth]{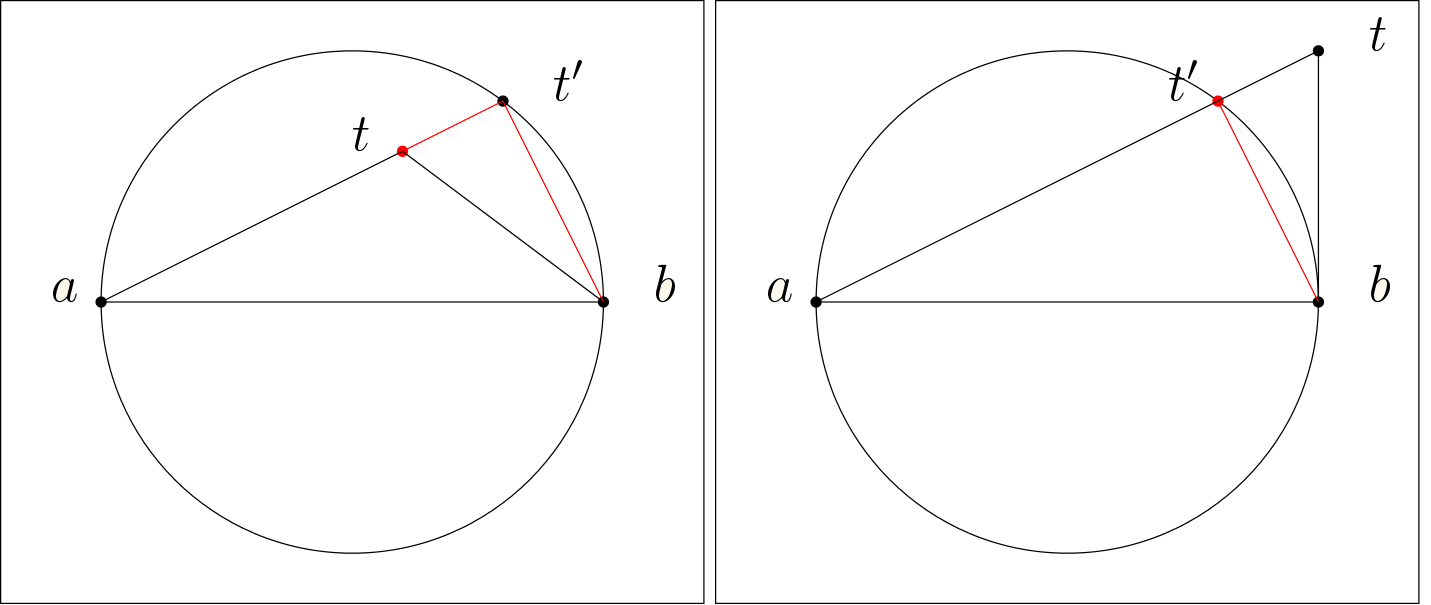}
  \caption{Point $t$ and spherical influence region $Sph(a,b)$}
  \label{fig:point-circle}
\end{figure*}

\paragraph{Algorithm~\ref{Alg:sph-pseudocode}:} Using Theorem~\ref{thrm:point-circle}, we present an algorithm to compute the spherical depth of a query point $q\in \mathbb{R}^2$ with respect to $S=\{x_1,...,x_n\} \subseteq \mathbb{R}^2$.  This algorithm is summarized in the following steps. The pseudocode of this algorithm is provided at the end of this chapter. 

\begin{itemize}

\item \textbf{Translating the points:}
Suppose that $T$ is a translation by $(-q)$. We apply $T$ to translate $q$ and all data points into their new coordinates. Obviously, $T(q)= O$.

\item \textbf{Sorting the translated data points:} In this step we sort the translated data points based on their angles in their polar coordinates. After doing this step, we have $S_T$ which is a sorted array of the translated data points.

\item \textbf{Calculating the spherical depth:} For the $i^{th}$ element in $S_{T}$, we define $O_i$ and $N_i$ as follows: 
\begin{equation}
\label{eq:O_i}
\begin{split}
&O_i=\left\{j: x_j\in S_T \: , \frac{\pi}{2} \leq |\theta(x_i) -\theta(x_j)|\leq \frac{3\pi}{2}\right\} \\& N_i= \{ 1,2,...,n\} \setminus O_i.
\end{split}
\end{equation}
Thus the spherical depth of $q$ with respect to $S$, can be computed by:
\begin{equation}
\label{eq:sph}
\SphD(q;S)=\SphD(O;S_T)= \frac{1}{2{n \choose 2}}\sum_{1\leq i\leq n}|O_i|.
\end{equation}
To present a formula for computing $|O_i|$, we define $f(i)$ and $l(i)$ as follows:
\[
f(i)=
\begin{cases}
\min N_i -1 &\text{if $\frac{\pi}{2}< \theta(x_i) \leq \frac{3\pi}{2}$}\\
    \min O_i & \text{otherwise}
\end{cases}
\]
\[
l(i)=
\begin{cases}
\max N_i +1 &\text{if $\frac{\pi}{2}< \theta(x_i) \leq \frac{3\pi}{2}$}\\
    \max O_i & \text{otherwise.}
\end{cases}
\]
Figure~\ref{fig:S-T-sorted} illustrates $O_i$, $N_i$,  $f(i)$, and $l(i)$ in two different cases. Considering the definitions of $f(i)$ and $l(i)$,
\begin{equation}
\label{eq:length-O_i}
|O_i|=
\begin{cases}
f(i)+(n-l(i)+1) &\text{if $\frac{\pi}{2}< \theta(x_i) \leq \frac{3\pi}{2}$}\\
    l(i)-f(i)+1 & \text{otherwise.}
\end{cases}
\end{equation}
\end{itemize}
This allows us to compute $|O_i|$ using a pair of binary searches. 
\begin{figure*}[!ht]
  \centering
    \includegraphics[width=0.95\textwidth]{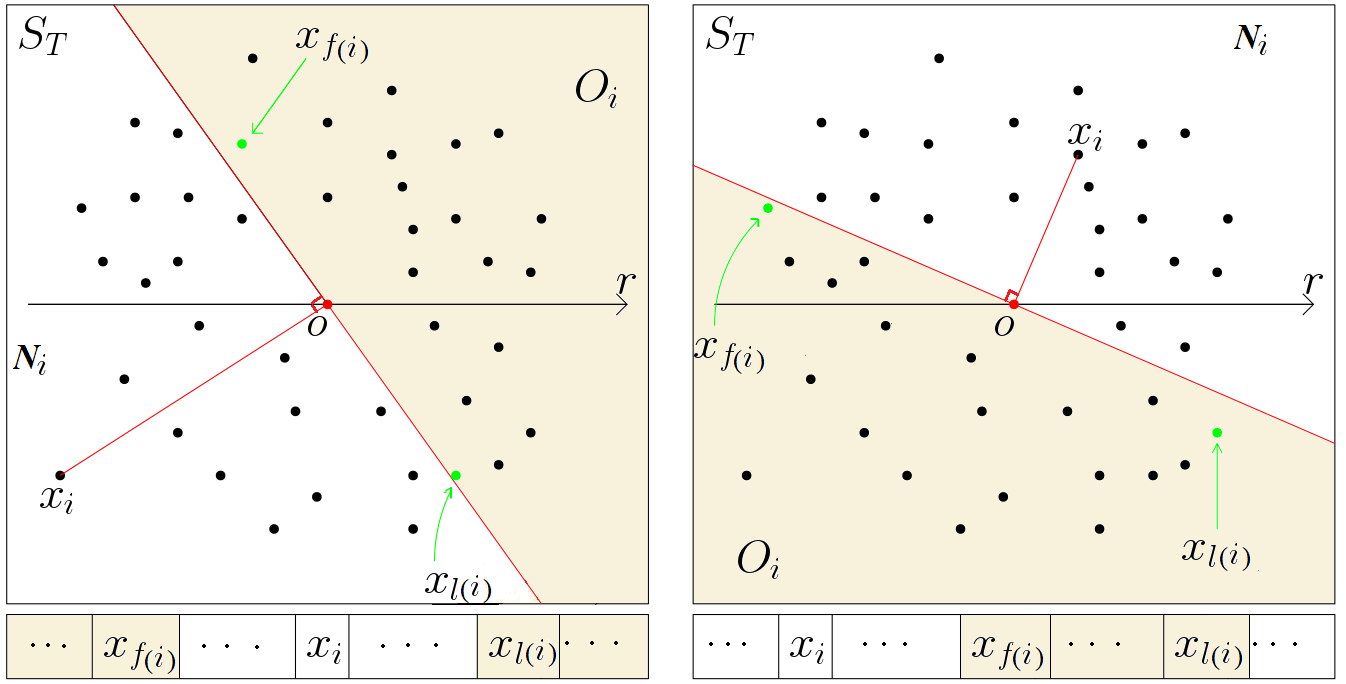}
  \caption{$\theta(x_i)\in (\frac{\pi}{2}, \frac{3\pi}{2}]$ (left figure), and $\theta(x_i)\notin (\frac{\pi}{2}, \frac{3\pi}{2}]$ (right figure)}
  \label{fig:S-T-sorted}
\end{figure*}
\paragraph{Time complexity of Algorithm~\ref{Alg:sph-pseudocode}:} The first procedure in the algorithm takes $\Theta(n)$ time to translate $q$ and all data points into the new coordinate system. The second procedure takes $\Theta(n\log n)$ time. Due to using binary search for every $O_i$, the running time of the last procedure is also $\Theta(n\log n)$. The rest of the algorithm contributes some constant time. In total, the running time of the algorithm is $\Theta(n\log n)$.
\begin{note}
\label{note:Coordinate-system}
\textbf{Coordinate system:} In practice it may be preferable to work in the Cartesian coordinate system. Sorting by angle can be done using some appropriate right-angle tests (determinants). Regarding the other angle comparisons, they can be done by checking the sign of dot products.
\end{note}
\section{Algorithm to Compute Planar $\beta$-skeleton Depth when $\beta> 1$}
\label{sec:beta-alg}
We recall from the definition of $\beta$-influence region that $S_{\beta}(x_i,x_j) ; \beta>1$ forms some lenses, and $S_{\infty}(x_i,x_j)$ forms some slabs for each pair $x_i$ and $x_j$ in $S\subseteq \mathbb{R}^2$. Using some geometric properties of such lenses and slabs, we prove Theorem~\ref{thrm:q-skeleton}. This theorem along with some results regarding the range counting problems in~\cite{agarwal2013range} help us to compute $\SkD_{\beta}(q;S); \beta>1$ in $O(n^{\frac{3}{2}+\epsilon})$ time, where $q$ and $S$ are in $\mathbb{R}^2$.
\begin{definition}
\label{def:line}
For an arbitrary non-zero point $a \in \mathbb{R}^2$  and parameter $\beta \geq 1$, $\ell(p)$ is a line that is perpendicular to $\overrightarrow{a}$ at the point $p=p(a,\beta)={(\beta -1)a}/{\beta}$. This line forms two halfspaces $H_o(p)$ and $H_a(p)$. The closed halfspace that includes the origin is $H_o(p)$ and the other halfspace which is open is $H_a(p)$.
\end{definition}

\begin{definition}
\label{def:ball}
For a disk $B(c,r)$ with the center $c=c(a,\beta)={\beta a}/{2(\beta -1)}$ and radius $r=\Vert c \Vert$, $B_o(c,r)$ is the intersection of $H_o(p)$ and $B(c,r)$,  and $B_a(c,r)$ is the intersection of $H_a(p)$ and $B(c,r)$, where $\beta>1$ and $a$ is an arbitrary non-zero point in $\mathbb{R}^2$. For the case of $\beta=1$, $B(c,r)= \emptyset$.
\end{definition}
Figure~\ref{fig:halfspace-balls} is an illustration of these definitions for different values of parameter $\beta$.

\begin{figure*}[!ht]
  \centering
    \includegraphics[width=1\textwidth]{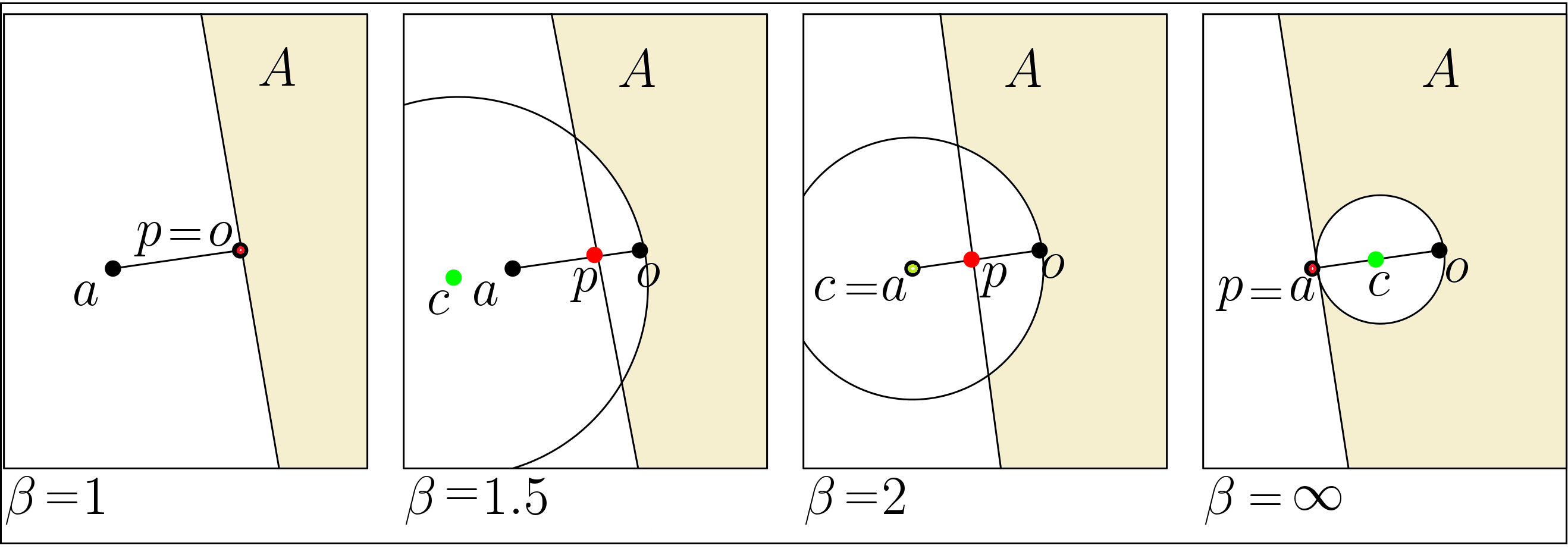}
  \caption{$H_o(p)$ and $B(c,r)$ defined by $a\in \mathbb{R}^2$ for $\beta=1,\;1.5,\;2,\;\text{and}\;\beta=\infty$, where $A=H_o(p)\setminus \{intB_o(c,r)\}$}
\label{fig:halfspace-balls}
\end{figure*}

\begin{theorem}
For non-zero points $a,b \in \mathbb{R}^2$, and parameter $\beta>1$, $b \in H_o(p)\setminus \{int B_o(c,r)\}$ if and only if the origin $O=(0,0)$ is contained in $S_{\beta}(a,b)$, where $c={\beta a}/{(2(\beta -1))}$, $r=\Vert c \Vert$, and $p={(\beta -1)a}/{\beta}$.
\label{thrm:q-skeleton}
\end{theorem}

\begin{proof} 
First, we show that $B_o(c,r)$ is a well-defined set meaning that $\ell(p)$ intersects $B(c,r)$. We compute $d(c,\ell(p))$, the distance of $c$ from $\ell(p)$, and prove that this value is not greater than $r$. It can be verified that $d(c,\ell(p))= d(c,p)$. Let $k={\beta}/{(2(\beta -1))}$; the following calculations complete this part of the proof.
\begin{align*}
d(c,p)&=d(\frac{\beta a}{2(\beta -1)},\frac{(\beta -1)a}{\beta})= d(ka,\frac{1}{2k}a)\\&=(k-\frac{1}{2k})\sqrt{({a_x}^2+{a_y}^2)}=(\frac{2k^2-1}{2k})\Vert a \Vert\\& < \frac{2k^2}{2k}\Vert a \Vert= k\Vert a \Vert =r.
\end{align*}
We recall the definition of $\beta$-influence region given by $S_{\beta}(a,b)=B(c_a,\frac{\beta}{2}\Vert a-b \Vert) \cap B(c_b,\frac{\beta}{2}\Vert a-b \Vert)$, where $c_a=\frac{\beta}{2}a +(1-\frac{\beta}{2})b$ and $c_b=\frac{\beta}{2}b +(1-\frac{\beta}{2})a$. Using this definition, following equivalencies can be derived from $O \in S_{\beta}(a,b)$.
\begin{align*}
O \in S_{\beta}(a,b)&\Leftrightarrow \frac{\beta \Vert a-b \Vert}{2} \geq max\{\Vert c_a \Vert , \Vert c_b \Vert\}\\&\Leftrightarrow \beta \Vert a-b \Vert \geq max\{\Vert \beta(a-b)+2b\Vert , \Vert \beta(b-a)+2a \Vert\}\\&\Leftrightarrow \beta ^2 \Vert a-b \Vert ^2 \geq max\{\Vert \beta(a-b)+2b\Vert ^2 , \Vert \beta(b-a)+2a \Vert ^2\}\\&\Leftrightarrow 0 \geq max\{b^2(1-\beta)+\beta\langle a,b\rangle , a^2(1-\beta)+\beta \langle a,b\rangle\}.
\end{align*}
By solving these inequalities for $(\beta -1)/\beta$ which is equal to $1/2k$, we obtain Equation \eqref{eq:equivalencies}.  
\begin{equation}
\label{eq:equivalencies}
\frac{1}{2k} \geq max \left\{ \frac{\langle a,b\rangle}{\Vert a \Vert ^2}, \frac{\langle a,b\rangle}{\Vert b \Vert ^2} \right\}
\end{equation}
For a fixed point $a$, the inequalities in Equation \eqref{eq:equivalencies} determine one halfspace and one disk given by \eqref{eq:halfspacerange} and \eqref{eq:disk}, respectively.
\begin{equation}
\label{eq:halfspacerange}
\frac{1}{2k} \geq \frac{\langle a,b\rangle}{\Vert a \Vert ^2} \Leftrightarrow\langle a,b\rangle \leq\frac{1}{2k} \Vert a \Vert ^2.
\end{equation}
\begin{equation}
\label{eq:disk}
\begin{split}
\frac{1}{2k} \geq \frac{\langle a,b\rangle}{\Vert b \Vert ^2} &\Leftrightarrow b^2-2k\langle a,b\rangle \geq 0 \\& \Leftrightarrow
b^2-2k\langle a,b\rangle+k^2a^2 \geq k^2a^2 \\& \Leftrightarrow \left( b- ka\right)^2 \geq \left( k \Vert a \Vert\right)^2.
\end{split}
\end{equation}
The proof is complete because for a point $a$, the set of all points $b$ containing in the feasible region defined by Equations \eqref{eq:halfspacerange} and \eqref{eq:disk} is equal to \\$H_o(p)\setminus \{intB_o(c,r)\}$.  
\end{proof}
\begin{proposition}
\label{lm:halfspace-circle-center}
For $\beta=2+\sqrt{2}$, the boundary of the given halfspace in \eqref{eq:halfspacerange} passes through the center of the given disk in \eqref{eq:disk}.
\end{proposition}
\begin{proof}
It is enough to substitute $b$ in the given halfspace with $ka$ which is the center of the given disk.
\begin{align*}
\langle a,(ka)\rangle\leq \frac{1}{2k}\Vert a \Vert ^2 &\Rightarrow k\Vert a \Vert^2 \leq \frac{1}{2k}\Vert a \Vert^2 \Rightarrow 2k^2\leq 1\\ &\Rightarrow 2(\frac{\beta}{2(\beta-1)})^2 \leq 1 \Rightarrow \beta =2+\sqrt{2}
\end{align*} 
\end{proof}
\paragraph{Algorithm~\ref{Alg:betaskeleton-pseudocode}:} Using Theorem~\ref{thrm:q-skeleton}, we present an algorithm to compute the $\beta$-skeleton depth of $q\in \mathbb{R}^2$ with respect to $S=\{x_1,...,x_n\} \subseteq \mathbb{R}^2$. This algorithm is summarized in two steps. The pseudocode of this algorithm can be found at the end of this chapter.
\begin{itemize}
\item \textbf{Translating the points:}
This step is exactly the same step as in Algorithm~\ref{Alg:sph-pseudocode}.
\item \textbf{Calculating the $\beta$-skeleton depth:} Suppose that $a=(a_x,a_y)$ is an element in $S'$ (translated $S$). We consider a disk and a line as follows:
\[
B(c,r):\left(x- ka_x \right)^2+ \left(y- ka_y\right)^2 = \left(k\Vert a \Vert\right)^2
\]
\[
\ell(p):a_xx+ a_yy= \frac{1}{2k}\Vert a \Vert,
\]
where $k$, $c$, $r$, and $p$ are defined in Theorem~\ref{thrm:q-skeleton}. From Theorem $1.2$ in~\cite{agarwal2013range}, we can compute $\vert H_o(p)\vert$ with $O(n)$ storage, $O(n\log n)$ expected preprocessing time, and $O(n^{\frac{1}{2}+\epsilon})$ query time, where $\vert H_o(p)\vert$ is the number of all elements of $S'$ that are contained in $H_o(p)$. For the elements of $S'$, $\vert intB_o(c,r) \vert$ which is defined as the number of elements containing in the interior of $B_o(c,r)$ can also be computed with the same storage, expected preprocessing time, and query time. We recall that $B_o(c,r)$ is the intersection of halfspace $H_o(p)$ and disk $B(c,r)$, where $p$, $c$, and $r$ are some functions of $a$.  Finally, $\SkD_\beta(q,S)$ which is equal to $\SkD_{\beta}(O;S')$ can be computed by Equation~\eqref{eq:skeleton-algorithm}.
 \begin{equation}
 \label{eq:skeleton-algorithm}
 \SkD_\beta(q,S)=\frac{1}{2{n \choose 2}}\left(\sum_{p}\vert H_o(p)\vert -\sum_{c} \vert intB_o(c,r) \vert\right)
 \end{equation}
Referring to Definitions~\ref{def:line} and~\ref{def:ball}, $H_o(p)$ and $B_o(c,r)$ can be computed in constant time.   
\end{itemize}
Theorem~\ref{thrm:q-skeleton} and Algorithm~\ref{Alg:betaskeleton-pseudocode} are valid for $\beta>1$. However, the case $\beta=1$ (Algorithm~\ref{Alg:sph-pseudocode} for spherical depth) can also be included in this result. In this case, $\forall c; B(c,r)=\emptyset$, and consequently, $\vert intB_o(c,r)\vert=0$. Therefore, $\SkD_1(q;S)$ which is equal to $\SkD_1(O;S')$ can be computed by:  
\begin{equation}
\label{eq:sphD-skeleton-algorithm}
\SkD_1(q;S)=\frac{1}{2{n \choose 2}}\sum_{p}\vert H_o(p)\vert.
\end{equation} 
\paragraph{Time complexity of Algorithm~\ref{Alg:betaskeleton-pseudocode}:} The translating procedure as is discussed in Algorithm~\ref{Alg:sph-pseudocode}, takes $O(n)$ time. With the $O(n \log n)$ expected preprocessing time, the second procedure takes $O(n^{\frac{3}{2}+\epsilon})$ time. In this procedure, the loop iterates $n$ times, and the range counting algorithms take $O(n^{\frac{1}{2}+\epsilon})$ time. The expected preprocessing time $O(n \log n)$ is required to obtain a data structure for the aforementioned range counting algorithms. The rest of the algorithm takes some constant time per loop iteration, and therefore the total expected running time of the algorithm is $O(n^{\frac{3}{2}+\epsilon})$. 
\section{Algorithm to Compute Planar Halfspace Depth}
\label{sec:tukey-alg}
To compute the halfspace depth of query point $q$ with respect to data set $S$, we need to find the minimum portion of data points separated by a halfspace through $q$. In this section, we develop an optimal algorithm to compute the planar halfspace depth of a query point. Identical to Aloupis' \cite{aloupiscomputing} algorithms, our algorithm takes $\Theta(n\log n)$ time. After sorting data points by angle, Aloupis employed the counterclockwise sweeping of a specific halfline. However, to obtain our algorithm, we reuse most of Algorithm \ref{Alg:sph-pseudocode}, and employ the specialized halfspace range counting that is explored in Section~\ref{sec:sph-alg}.   

\paragraph{Algorithm~\ref{Alg:halfspace-pseudocode}:} Suppose that $S=\{x_1,...,x_n\}\subset \mathbb{R}^2$ and $q\in \mathbb{R}^2$  are given. We summarize the algorithm in the following steps.

\begin{itemize}
\item \textbf{Translating the points:} This step is the same as the first step in Algorithm~\ref{Alg:sph-pseudocode}.

\item \textbf{Computing the halfspaces} In Equation \eqref{eq:halfspacedef}, it is not practical to compute $|S\cap H|$ for all halfspaces $H$ that pass through the query point. Instead of considering all of the halfspaces, we define $\mathbb{H}$ to be a finite set of the desired halfspaces such that we can obtain all possible values of $|S\cap H|$ if $H\in \mathbb{H}$. Computation of $\mathbb{H}$ can be done as follows:
\begin{itemize}
\item[$(i)$] Project all of the nontrivial\footnote{$a\in \mathbb{R}^2$ is nontrivial if $\Vert a\Vert> 0.$} elements of $S'$ (translated $S$) on the unit circle $C(O,1)$.
\item[$(ii)$] Construct $P=\{\pm x'_{ip}\}$, where $x'_{ip}$ is generated in $(i)$.
\item [$(iii)$] Using an $O(n\log n)$ sorting algorithm, sort the elements of $P$ by angle in counterclockwise order.
\item[$(iv)$] Remove the duplicates in the sorted $P$. 
\item[$(v)$] Let $m_i$ be the middle point of each pair of successive elements in the sorted array in $(iii)$. Suppose that $\ell _i$ is the line that passes through the points $O=(0,0)$ and $m_i$. Each line $\ell_i$ forms two halfspaces $h_{i1}$ and $h_{i2}$. As such, $\mathbb{H}$ can be defined by:
\[
\mathbb{H}=\{h_{ij};1\leq i\leq n, j\in\{1,2\}\}.
\]  
\end{itemize}
\item \textbf{Calculating the halfspace depth:} Similar to the computation of $\vert O_i\vert$ in Equation \eqref{eq:length-O_i}, a pair of binary searches can be applied to compute each of $\vert h_{ij}\cap S'\vert$, where $1\leq i\leq n$ and $j\in\{1,2\}$. Therefore, $\HD(q;S)$ which is equal to $\HD(O;S')$ can be computed by:
\begin{equation*}
\HD(q;S)=\frac{2}{n}\left(\min_{ij}\left\{\vert h_{ij}\cap S'\vert; h_{ij}\in \mathbb{H}\right\}+\vert trivials\vert\right).
\end{equation*}
\end{itemize}
\paragraph{Time complexity of Algorithm~\ref{Alg:halfspace-pseudocode}:} Referring to the analysis of Algorithm~\ref{Alg:sph-pseudocode}, the Translating procedure takes $\Theta(n)$ time. Sorting the elements of $P$ causes the second procedure to take $\Theta(n\log n)$ time. The rest of work in the second procedure takes some linear time. Finally, the depth calculation procedure takes $O(n\log n)$ because for every element in $\mathbb{H}$, two binary searches are called. Since $\mathbb{H}$ contains at most $2n$ elements, the outer loop takes $O(n)$. The binary searches take $O(\log n)$. The rest of the algorithm run in some constant time per loop iteration. From the above analysis, the overall running time of this algorithm is $O(n\log n)$.
\section{Pseudocode}
In this section we provide the pseudocode of the presented algorithms in sections \ref{sec:sph-alg}, \ref{sec:beta-alg}, and \ref{sec:tukey-alg}. First, we define some procedures which are used in the depth calculation algorithms.
\begin{algorithm}
\renewcommand{\algorithmicrequire}{\textbf{Input:}}
\renewcommand{\algorithmicensure}{\textbf{Output:}}
\newcommand{\Break}{\State \textbf{break}}
\caption{Translation$(S,q)$}
\label{Alg:Translating.procedure}
\begin{algorithmic}[1]
\Require Data set $S$ and Query point $q$
\Ensure Translated data set $S'$
\item[]
\For{each $x_i \in S$ }
\State $x_i \gets (x_i-q)$ 
\EndFor
\State \Return $S'$ 
\end{algorithmic}
\end{algorithm}
\begin{algorithm}
\renewcommand{\algorithmicrequire}{\textbf{Input:}}
\renewcommand{\algorithmicensure}{\textbf{Output:}}
\newcommand{\Break}{\State \textbf{break}}
\caption{Sorting-by-angle$(S)$}
\label{Alg:Sorting.procedure}
\begin{algorithmic}[1]
\Require Data set $S$
\Ensure Sorted array $S_T$
\item[]
\For{each $x_i=(a_i,b_i)\in S$}
\State $\theta(x_i)\gets \arctan(b_i/a_i)$
\EndFor
\State Using an $O(n\log n)$ sorting algorithm, sort $x_i$ based on $\theta(x_i)$
\State \Return $S_{T}$
\end{algorithmic}
\end{algorithm}
\begin{note}
Instead of using polar coordinates in Algorithm \ref{Alg:Sorting.procedure}, sorting can also be done by applying the suggested method in Note \ref{note:Coordinate-system}.
\end{note}
\begin{algorithm}
\renewcommand{\algorithmicrequire}{\textbf{Input:}}
\renewcommand{\algorithmicensure}{\textbf{Output:}}
\newcommand{\Break}{\State \textbf{break}}
\caption{Trivial-and-Nontrivial$(S)$}
\label{Alg:TrivialNontrivial.procedure}
\begin{algorithmic}[1]
\Require Data set $S$
\Ensure Partition $S$ into two sets $Trival$ and $NonTrivial$
\item[]
\State Initialize $Trivial=\emptyset$, $NonTrivial=\emptyset$
\For{each $x_i\in S$}
\State \textbf{if} \; $(\Vert x_i\Vert=0)$ 
\State $\hspace{0.5cm} Trivial\gets Trivial\cup \{x_i\}$
\State \textbf{else}
\State $\hspace{0.5cm} NonTrivial\gets NonTrivial\cup\{x_i\}$
\EndFor
\State \Return $(Trivial,NonTrivial)$
\end{algorithmic}
\end{algorithm}
\begin{algorithm}
\renewcommand{\algorithmicrequire}{\textbf{Input:}}
\renewcommand{\algorithmicensure}{\textbf{Output:}}
\newcommand{\Break}{\State \textbf{break}}
\caption{Projection$(S)$}
\label{Alg:Projection.procedure}
\begin{algorithmic}[1]
\Require Data set $S$
\Ensure Projection of $S$ on circle $C(O,1)$
\item[] 
\State Initialize $S_p=\emptyset$
\For{each $x_i\in S$}
\State $x_{ip}\gets x_i/\Vert x_i\Vert$
\State $S_p\gets S_p\cup\{\pm x_{ip}\} $
\EndFor
\State \Return $S_p$
\end{algorithmic}
\end{algorithm}
\begin{algorithm}
\renewcommand{\algorithmicrequire}{\textbf{Input:}}
\renewcommand{\algorithmicensure}{\textbf{Output:}}
\newcommand{\Break}{\State \textbf{break}}
\caption{Computing the planar $\beta$-skeleton depth, $\beta=1$ (spherical depth)}
\label{Alg:sph-pseudocode}
\begin{algorithmic}[1]
\Require Data set $S$ and Query point $q$
\Ensure $\SphD(q;S)$
\item[]
\State Initialize $\SphD(q;S)=0$
\State $S'\gets$ Translation$(S,q)$
\State $S_T \gets$ Sorting-by-angle$(S')$
\For{each $x_{i}\in S_{T}$}
\State Initialize $a=0$ and $b=n+1$
\State Using two \emph{binary search calls}, update the values of $a$ and $b$
\State \textbf{if} \; ($0< \theta(x_i) \leq \frac{\pi}{2}$)
\State \hspace{0.5cm} $a=min\{j:\: \theta(x_j)-\theta(x_i)\geq \frac{\pi}{2}\}$
\State \hspace{0.5cm} $b=max\{j:\: \frac{\pi}{2} \leq \theta(x_j)-\theta(x_i)\leq \frac{3\pi}{2}\}$
\State \textbf{else-if} \; ($\frac{\pi}{2}< \theta(x_i) \leq \frac{3\pi}{2}$)
\State \hspace{0.5cm} $a=max\{j:\: \theta(x_i)-\theta(x_j)\geq \frac{\pi}{2}\}$
\State \hspace{0.5cm} $b=min\{j:\: \theta(x_j)-\theta(x_i)\geq \frac{\pi}{2}\}$ 
\State \textbf{else}
\State \hspace{0.5cm} $a=min\{j:\: \theta(x_i)-\theta(x_j)\leq \frac{3\pi}{2}\}$
\State \hspace{0.5cm} $b=max\{j:\: \frac{\pi}{2} \leq \theta(x_i)-\theta(x_j)\leq \frac{3\pi}{2}\}$
\State \textbf{if} \; $(a=0 \;\;\text{and
}\;\;b=n+1)$ 
\State \hspace{0.5cm} $\vert O_i\vert=0$
\State \textbf{else}
\State \hspace{0.5cm} $\vert O_i\vert= \begin{cases}
a+(n-b+1)\; ; \frac{\pi}{2}< \theta(x_i) \leq \frac{3\pi}{2}\\
    b-a+1 \hspace{1.35cm}; \text{otherwise.}
\end{cases}
$
\State $\SphD(q;S)\gets \SphD(q;S)+ \frac{1}{2{n \choose 2}} |O_i|$ 
\EndFor
\State \Return $\SphD(q;S)$ 
\end{algorithmic}
\end{algorithm}
\begin{note}
To avoid unusual notations in Algorithm~\ref{Alg:sph-pseudocode}, we use the variables $a$ and $b$ instead of $f(i)$ and $l(i)$, respectively in the text.
\end{note}
\begin{algorithm}
\renewcommand{\algorithmicrequire}{\textbf{Input:}}
\renewcommand{\algorithmicensure}{\textbf{Output:}}
\newcommand{\Break}{\State \textbf{break}}
\caption{Computing the planar $\beta$-skeleton depth, $\beta>1$}
\label{Alg:betaskeleton-pseudocode}
\begin{algorithmic}[1]
\Require Data set $S$, Query point $q$, Parameter $\beta > 1$
\Ensure $\SkD_{\beta}(q;S)$
\item[]
\State Initialize $\SkD_{\beta}(q;S)=0$
\State $S'\gets$ Translation$(S,q)$
\For{each $a \in S'$}
\State Using two $O(n^{\frac{1}{2}+\epsilon})$ range counting algorithms, compute $\vert H_o(p)\vert$ 
\item[\hspace{1.35cm}and $\vert intB_o(c,r) \vert$]
\item[\hspace{1.35cm}(The computations of $H_o(p)$ and $intB_o(c,r)$ take constant time)] %
\State $\SkD_{\beta}(q;S)\gets \SkD_{\beta}(q;S)+\frac{1}{2{n \choose 2}}\left(\vert H_o(p)\vert -\vert intB^o(c,r)\vert\right)$
\EndFor
\State \Return $\SkD_{\beta}(q;S)$ 
\end{algorithmic}
\end{algorithm}
\begin{algorithm}
\renewcommand{\algorithmicrequire}{\textbf{Input:}}
\renewcommand{\algorithmicensure}{\textbf{Output:}}
\newcommand{\Break}{\State \textbf{break}}
\caption{Computing the planar halfspace depth}\label{Alg:halfspace-pseudocode}
\begin{algorithmic}[1]
\Require Data set $S$, Query point $q$
\Ensure $\HD(q;S)$
\item[]
\State Initialize $\HD_q=0$
\State $S' \gets$ Translation$(S,q)$
\State $(Triv,NTriv)\gets$ Trivial-and-Nontrivial$(S')$
\State $P \gets$ Projection$(NTriv)$ 
\State $P_T \gets$ Sorting$(P)$
\State $A_P \gets$ Removing-Duplicates$(P_T)$
\For{each $a_i$ in $A_P$}
\State $\ell_i \gets$ line passing through $(0,0)$ and $m_i=(a_i+a_{i+1})/2$
\EndFor
\State $\mathbb{H} \gets$ the set of all halfspaces $h_{ij}$ defined by $\ell_i$, where $1\leq i\leq n, j\in\{1,2\}$
\For{each $h_{ij}$ in $\mathbb{H}$}
\State Using a pair of binary searches, compute $|h_{ij}\cap S'|$ (similar to the
\item[\hspace{1.35cm}computation of $|O_i|$ in Algorithm \ref{Alg:sph-pseudocode})]
\State $\HD_q\gets \min\{\HD_q,|h_{ij}\cap S'|\}$
\EndFor
\State $\HD(q;S)\gets \frac{2}{n}\left(\HD_q+|Triv|\right)$
\State \Return $\HD(q;S)$  
\end{algorithmic}
\end{algorithm}
\begin{algorithm}
\renewcommand{\algorithmicrequire}{\textbf{Input:}}
\renewcommand{\algorithmicensure}{\textbf{Output:}}
\newcommand{\Break}{\State \textbf{break}}
\caption{Removing-Duplicates$(A)$}
\label{Alg:dupicatesRemoval.procedure}
\begin{algorithmic}[1]
\Require Sorted array $A$
\Ensure Sorted array $Anc$ with unique elements
\item[]
\State Define $Anc$ with dynamic size 
\State Initialize $Anc[1]=A[1]$, $j=2$
\For{$i=2$ to $\vert A\vert$}
\State \textbf{if}\;$(A[i]\neq A[i-1])$ 
\State $\hspace{0.5cm} Anc[j]\gets A[i]$
\State $\hspace{0.5cm} j\gets j+1$
\EndFor
\State \Return $Anc$
\end{algorithmic}
\end{algorithm}%algorithmic results
%%-----------------Chapter 6------------------------
\chapter{Lower Bounds}
\label{ch:lowerbound}
As discussed in Chapter \ref{ch:different-DD}, computing each of simplicial depth and halfspace depth in $\mathbb{R}^2$ requires $\Omega(n\log n)$ time. In this chapter we prove that computing the planar $\beta$-skeleton depth also requires $\Omega(n \log n)$, $\beta\geq 1$. We reduce the problem of Element Uniqueness to the problem of computing the $\beta$-skeleton depth in three cases $\beta=1$, $1< \beta < \infty$, and $\beta=\infty$. It is known that the question of Element Uniqueness has a lower bound of $\Omega (n\log n)$ in the algebraic decision tree model of computation proposed in~\cite{ben1983lower}.  
\section{Lower Bound for the Planar $\beta$-skeleton Depth, $\beta=1$}
\begin{theorem}
Computing the planar spherical depth ($\beta$-skeleton depth, $\beta=1$) of a query point in the plane takes $\Omega (n\log n)$ time.
\label{thrm:lowebound}
\end{theorem} 
\begin{proof}
We show that finding the spherical depth allows us to answer the question of Element Uniqueness. Suppose that $A=\{a_1,...,a_n\}$, for $n\geq 2$ is a given set of real numbers. We suppose all of the numbers to be positive (negative), otherwise we shift the points onto the positive $x$-axis. For every $a_i \in A$ we construct four points $x_i$, $x_{n+i}$, $x_{2n+i}$, and $x_{3n+i}$ in the polar coordinate system as follows:
\[
x_{(kn+i)}=\left(r_i,\theta_i+\frac{k\pi}{2}\right); \; 0\leq k \leq 3, 
\] 
where $r_i= \sqrt{1+{a_i^2}}$ and $\theta_i=\tan^{-1}(1/a_i)$. Thus we have a set $S$ of $4n$ points $x_{kn+i}$, for $1\leq i \leq n$. See Figure~\ref{fig:lower-bound-S}. The Cartesian coordinates of the points can be computed by:
\[
x_{(kn+i)}=
  \left[ {\begin{array}{cc}
   0 & -1\\
   1 & 0\\
  \end{array} } \right]^k \left( {\begin{array}{cc}
   a_i\\
   1 \\
  \end{array} } \right) ;\; k=0,1,2,3.
\]
We select the query point $q=(0,0)$, and present an equivalent form of Equation~\eqref{eq:O_i} for $O_j$ as follows:
\begin{equation}
\label{eq:O-for-lowebound}
O_j=\left\{x_k\in S\mid \angle x_jqx_k \geq \frac{\pi}{2} \right\}, \: 1\leq j\leq 4n,
\end{equation}
We compute $\SphD(q;S)$ in order to answer the Element Uniqueness problem. Suppose that every $x_j\in S$ is a unique element. In this case, $|O_j|=2n+1$ because, from~\eqref{eq:O-for-lowebound}, it can be figured out that the expanded $O_j$ is as follows:
\[
O_j=\begin{cases}
\{x_{n+1},...,x_{n+j}, x_{2n+1},...,x_{3n}, x_{3n+j},...,x_{4n}\} ;&j\in \{1,...,n\}\\
\{x_{2n+1},...,x_{n+j}, x_{3n+1},...,x_{4n}, x_{j-n},...,x_{n}\} ;&j\in \{n+1,...,2n\}\\
\{x_{3n+1},...,x_{n+j}, x_{1},...,x_{n}, x_{j-n},...,x_{2n}\} ;&j\in \{2n+1,...,3n\}\\
\{x_{1},...,x_{j-3n}, x_{n+1},...,x_{2n}, x_{j-n},...,x_{3n}\} ;&j\in \{3n+1,...,4n\}.
\end{cases}
\]
Let $\SphD_S(q)$ be the unnormalized form of $\SphD(q;S)$. Referring to Theorem~\ref{thrm:point-circle} and Equation~\eqref{eq:sph}, 
\[
\SphD_S(q)={\vert S\vert \choose 2}\SphD(q;S)=\frac{1}{2}\sum_{1\leq j\leq 4n}(2n+1)=4n^2+2n.
\]
Now suppose that there exist some $i\neq j$ such that $x_i= x_j$ in $S$. In this case, from Equation ~\eqref{eq:O-for-lowebound}, it can be seen that: 
\[
|O_{(kn+i)\bmod 4n}|=|O_{(kn+j)\bmod 4n}|=2n+2,  
\]
where $k=0,1,2,3$ (see Figure~\ref{fig:lower-bound-S}). As an example, for $k=0$, $|O_j|=|O_i|=2n+2$ because the expanded form of these two sets is as follows: (without loss of generality, assume $i<j<n$)
\[
O_i=O_j=\{x_{n+1}..x_{n+j},x_{2n+1}..x_{3n},x_{3n+i},x_{3n+j}, x_{3n+j+1}..x_{4n}\}.
\]
Theorem~\ref{thrm:point-circle} and Equation~\eqref{eq:sph} imply that:
\[
\SphD_S(q)\geq \frac{1}{2}(8+\sum_{1\leq j\leq 4n}(2n+1))= 4n^2+2n+4.
\]
Therefore the elements of $A$ are unique if and only if the spherical depth of $(0,0)$ with respect to $S$ is $4n^2+2n$. This implies that the computation of spherical depth requires $\Omega (n\log n)$ time. It is necessary to mention that the only computation in the reduction is the construction of $S$ which takes $O(n)$ time. Finally, we mention that the reduction does not depend on the sorted order of the elements.
\end{proof}
\begin{figure}[!ht]
  \centering
    \includegraphics[width=0.7\textwidth]{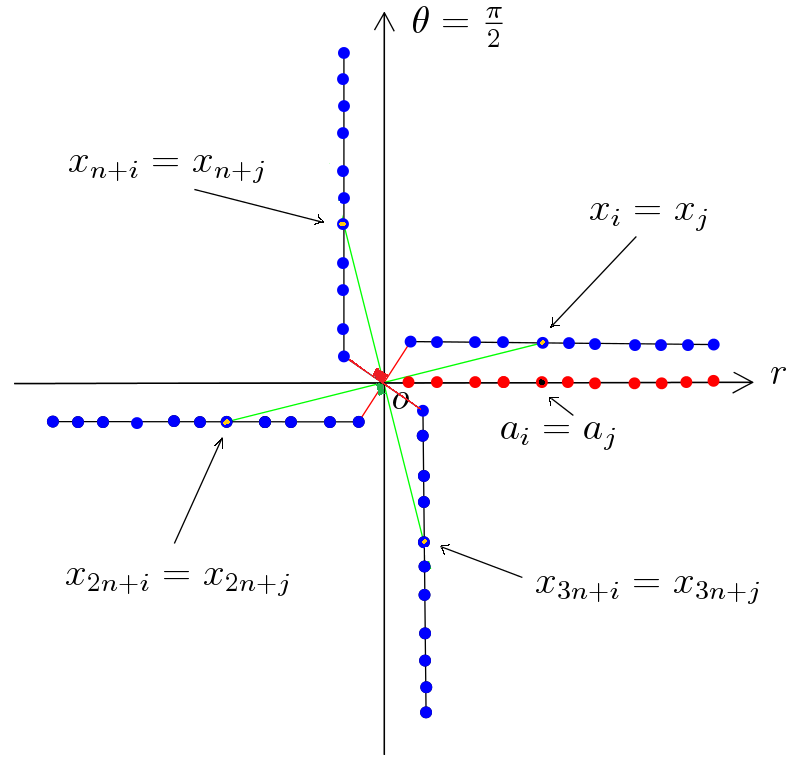}
  \caption{An illustration of $A$ and $S$ for $\beta=1$}
  \label{fig:lower-bound-S}
\end{figure}
\begin{note}
Instead of four copies of the elements of $A$, we could consider two copies of such elements to construct $S$ as used in section \ref{sec:lower-beta-1to-infty}. However, the depth calculation becomes more complicated in this case.
\end{note}
\section{Lower Bound for the Planar $\beta$-skeleton Depth, $1<\beta<\infty$}
\label{sec:lower-beta-1to-infty}
First, we prove the lower bound for the planar lens depth where $\beta=2$ in $\beta$-skeleton depth. Using the same reduction technique, we generalize the result to all values of $1<\beta<\infty$.
\begin{lemma}
For $1\leq i \leq n$ and $b_i\in \mathbb{R}^+$, suppose that $S$ and $L_j$ are two sets of polar coordinates as follows:
\[
S=\{x_i=(b_i,0),x_{n+i}=(b_i,\pi/3)\}
\]
\[ 
L_j=\{x_k\in S\mid O=(0,0)\in L(x_j,x_k)\}\; ,1\leq j\leq 2n.
\]
For a unique element $x_j\in S$, $L_j=\{x_{(n+j)\bmod 2n}\}$.
\label{lm:lens-unique}
\end{lemma}
\begin{proof}
Suppose that $x_k\in L_j; k\notin\{j,(n+j)\bmod 2n\}$. We prove that such $x_k$ does not exist. If $\angle x_jOx_k =0$, it is obvious that $O \notin L(x_j,x_k)$ and $x_k$ cannot be an element of $L_j$. For the case $\angle x_jOx_k =\pi/3$, by definition of $L_j$, $O \in L(x_j,x_k)$. From Definition \ref{def:beta-influence} for $\beta=2$,
\begin{align}
O\in L(x_j,x_k)&\Leftrightarrow O\in B(x_j,\Vert x_j-x_k\Vert)\cap B(x_k,\Vert x_j-x_k\Vert)\\&\Leftrightarrow\Vert x_j-O\Vert\leq\Vert x_j-x_k\Vert\; \&\; \Vert x_k-O\Vert\leq\Vert x_j-x_k\Vert\\& \label{eq:dxjxk}\Leftrightarrow \Vert x_j-x_k\Vert\geq \max\{b_j,b_k\}.
\end{align}
From the cosine formula\footnote{Cosine formula: For a triangle $\bigtriangleup abc$,
\[
\Vert ab\Vert ^2=\Vert ac\Vert ^2+\Vert bc\Vert ^2-2\Vert ac\Vert\Vert bc\Vert\cos(\angle bca) .
\]
} in triangle $\bigtriangleup x_iOx_j$, we have

\begin{equation}
\label{eq:cosin}
\Vert x_j-x_k\Vert^2= b^2_j + b^2_k -2 b_j b_k\cos(\pi/3).
\end{equation}
Equations~\eqref{eq:dxjxk} and~\eqref{eq:cosin} imply that:
\begin{align*}
b^2_j + b^2_k - b_j b_k\geq\max\{b_j^2,b_k^2\}&\Leftrightarrow\begin{cases}
b_j^2+b_k^2-b_jb_k\geq b_j^2\\
b_j^2+b_k^2-b_jb_k\geq b_k^2
\end{cases}\\
&\Leftrightarrow\begin{cases}
b_k^2-b_jb_k\geq 0\\
b_j^2-b_jb_k\geq 0
\end{cases}
\Leftrightarrow \begin{cases}
b_k\geq b_j\\
b_j\geq b_k
\end{cases}\Leftrightarrow b_k=b_j.
\end{align*}
This result contradicts the assumption of $x_k \in L_j; k\notin \{j,(n+j)\bmod 2n\}$
\end{proof}
\begin{theorem}
Computing the lens depth of a query point in the plane takes $\Omega (n\log n)$ time.
\label{thrm:lowebound-L}
\end{theorem}
\begin{figure}[!ht]
  \centering
    \includegraphics[width=0.6\textwidth]{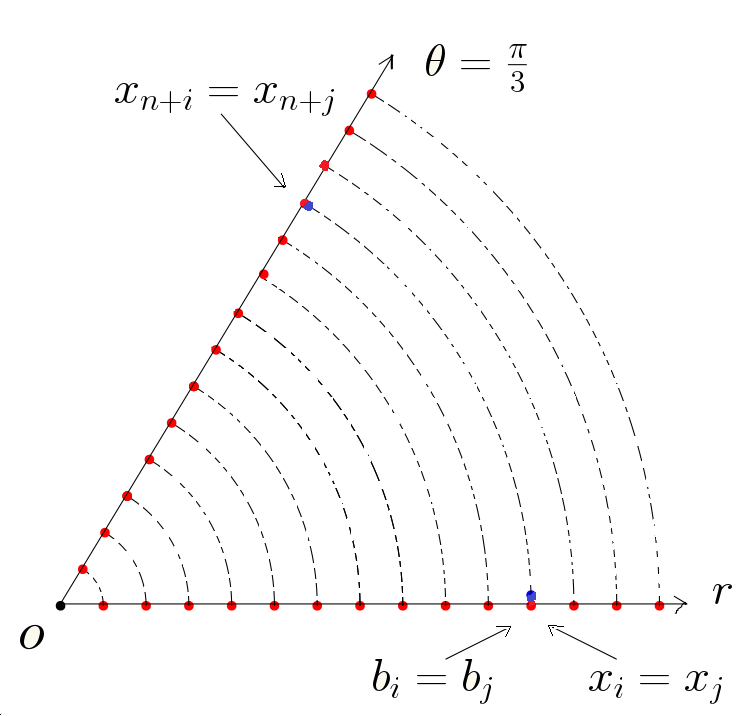}
  \caption{An illustration of $B$ and $S$ for $\beta=2$}
  \label{fig:lower-bound-L}
\end{figure}
\begin{proof}
Suppose that $B=\{b_1,...,b_n\}$, for $n\geq 2$ is a given set of real numbers. Without loss of generality, we let these numbers to be positive (see the proof of Theorem~\ref{thrm:lowebound}). For $1\leq i\leq n$, we construct set $S=\{x_i, x_{n+i}\}$ of $2n$ points in the polar coordinate system such that $x_i=(b_i,0)$ and $x_{n+i}=(b_i, \pi/3)$. See Figure~\ref{fig:lower-bound-L}. We select the query point $q=(0,0)$, and define $L_j$ as follows:
\begin{equation}
\label{eq:L-for-lowebound}
L_j=\left\{x_k\in S\mid q \in L(x_j,x_k) \right\}, \: 1\leq j\leq 2n.
\end{equation}
Using Equation~\eqref{eq:L-for-lowebound}, the unnormalized form of Equation~\eqref{eq:lens} can be presented by:
\begin{equation}
\label{eq:lens-lj}
\LD_S(q)= {\vert S\vert \choose 2}\LD(q;S)=\frac{1}{2}\sum_{1\leq j\leq 2n}\vert L_j\vert.
\end{equation}
We solve the problem of Element Uniqueness by computing $\LD_S(q)$. Suppose that every $x_j \in S$ is a unique element. In this case, Lemma~\ref{lm:lens-unique} implies that $L_j=\{x_{(n+j) \bmod 2n}\}$. From Equation~\eqref{eq:lens-lj}, we have
\[
\LD_S(q)=\frac{1}{2}\sum_{1\leq j\leq 2n} 1 = n.
\]
Now assume that there exists some $i\neq j$ such that $x_i = x_j$ in $S$. In this case, 
\begin{align*}
&L_j=L_i=\{x_{(n+i)\bmod 2n},x_{(n+j)\bmod 2n}\}\\&L_{(n+i)\bmod 2n}=L_{(n+j)\bmod 2n}=\{x_{i},x_{j}\}.
\end{align*} 
As such, for $1\leq t\leq 2n$ and $t\notin \{i,j,(i+n)\bmod 2n,(j+n)\bmod 2n\}$,
\begin{align*}
\LD_S(q)&=\frac{1}{2}\sum_t\vert L_t\vert+\frac{1}{2}(\vert L_i\vert+\vert L_j\vert+\vert L_{(i+n)\bmod 2n}\vert+\vert L_{(j+n)\bmod 2n}\vert)\\&=\frac{1}{2}(2(n-2))+\frac{1}{2}(2+2+2+2)=n+2.
\end{align*}
For the case of having more duplicated elements in $S$,
\begin{equation}
\label{eq:lens-lower}
\LD_S(q)=n+2c,
\end{equation}
where $c$ is the number of duplicates. Therefore the elements of $S$ are unique if and only if $c=0$ in Equation~\eqref{eq:lens-lower}. This implies that the computation of lens depth requires $\Omega (n\log n)$ time. Note that all of the other computations in this reduction take $O(n)$.  
\end{proof}
\begin{lemma}
\label{lm:angle-beta-influence}
For $a,b\in\mathbb{R}^2$, suppose that $u_{\beta}$ is a fixed intersection point between the two disks constructing the $S_{\beta}(a,b)$. $\theta=\angle au_{\beta}b=\cos^{-1}(1-1/\beta)$.
\end{lemma}
\begin{proof}
From Definition \ref{def:beta-influence},
\[S_{\beta}(a,b)=B(c_a,(\beta/2)\Vert a-b\Vert)\cap B(c_b,(\beta/2)\Vert a-b\Vert),\] where $c_a=(\beta/2)a+(1-\beta/2)b$ and $c_b=(\beta/2)b+(1-\beta/2)a$. It can be verified that $\Vert c_a-c_b\Vert=(\beta-1)\Vert a-b\Vert$. Suppose that $m$ is the middle point of $\overline{ab}$ and $h=\Vert u_{\beta}-m\Vert$.  See Figure \ref{fig:beta-influence-theta}. The value of $\theta$ can be computed as follows.
\begin{figure}[!ht]
  \centering
    \includegraphics[width=0.6\textwidth]{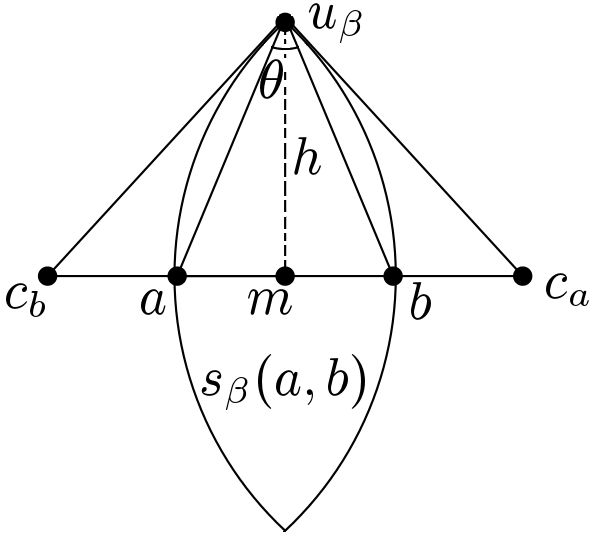}
  \caption{The $\beta$-influence region for $a,b\in \mathbb{R}^2$.}
  \label{fig:beta-influence-theta}
\end{figure}
\begin{align*}
h^2=\Vert u_{\beta}-a\Vert^2-\Vert a-m\Vert^2&=\Vert u_{\beta}-c_b\Vert^2-\Vert c_b-m\Vert^2\\&=(\frac{\beta}{2}\Vert a-b\Vert)^2-(\frac{\beta-1}{2}\Vert a-b\Vert)^2\\&=\frac{2\beta-1}{4}\Vert a-b\Vert^2\\\Rightarrow \Vert u_{\beta}-a\Vert^2&=\frac{2\beta-1}{4}\Vert a-b\Vert^2+\Vert a-m\Vert^2=\frac{\beta}{2}\Vert a-b\Vert^2
\end{align*}
The cosine formula in triangle $\bigtriangleup au_{\beta}b$ implies that
\begin{align*}
\Vert a-b\Vert^2&=\Vert u_{\beta}-a\Vert^2+\Vert u_{\beta}-b\Vert^2-2\Vert u_{\beta}-a\Vert\Vert u_{\beta}-b\Vert\cos(\theta)\\&=2\Vert u_{\beta}-a\Vert^2-2\Vert u_{\beta}-a\Vert^2\cos(\theta)\\&=\beta\Vert a-b\Vert^2-\beta\Vert a-b\Vert^2\cos(\theta)\Rightarrow \theta=\cos^{-1}(1-\frac{1}{\beta}).
\end{align*}
\end{proof}
\begin{theorem}
For $1<\beta<\infty$, computing the $\beta$-skeleton depth of a query point in the plane requires $\Omega(n\log n)$ time.
\label{thrm:lowebound-beta}
\end{theorem} 
\begin{proof sketch}
It is enough to generalize the reduction technique in Theorem \ref{thrm:lowebound-L}. As Lemma \ref{lm:angle-beta-influence} suggests, we need to choose $\theta=cos^{-1}(1-1/\beta)$ to construct $S=\{x_i=(b_i,0),x_{n+i}=(b_i,\theta)\}$, where $1\leq i\leq n$ and $b_i\in B$ defined in the proof of Theorem \ref{thrm:lowebound-L}. Figure \ref{fig:lower-bound-beta-other} illustrates that for every unique element $x_i\in S$, there exists only one element in $S$ such that the corresponding $\beta$-influence region contains $O$. As can be seen in this figure, $O$ is not contained in the $\beta$-influence region $S_{\beta}(x_i,y)$. Similar to the proof of Theorem \ref{thrm:lowebound-L}, it can be deduced that $\SkD_{\beta}(O;S)=n$ if every element in $S$ is unique. However, $\SkD_{\beta}(O;S)=n+2c$ if there exist $c$ duplicates among the elements of $S$. Note that we use the real \textit{RAM} model of computation in order to calculate $\theta$, where we need the square root of a real number to be computed in constant time.
\begin{figure}[!ht]
  \centering
    \includegraphics[width=0.6\textwidth]{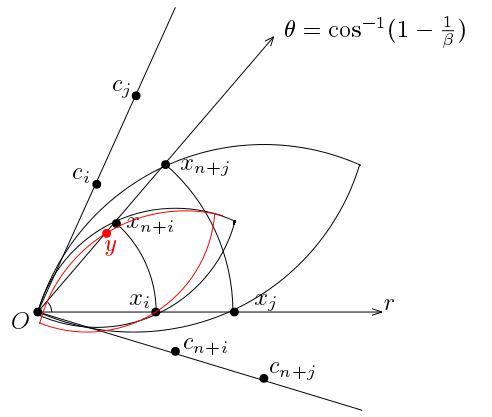}
  \caption{The $\beta$-influence regions for some elements of $S$ when $1<\beta<\infty.$}
  \label{fig:lower-bound-beta-other}
\end{figure}
\end{proof sketch}
\section{Lower Bound for the Planar $\beta$-skeleton Depth, $\beta=\infty$}
Suppose that $B$ is a set of positive real numbers as introduced in the proof of Theorem \ref{thrm:lowebound-L}. From the proof of Theorem \ref{thrm:lowebound-beta}, the rotation angle $\theta=cos^{-1}(1-1/\beta)$ is equal to $0$ if $\beta=\infty$. It means that there is not a proper rotation angle to make the second copy of the data points. However, it is enough to shift up the points by some constant (e.g. $\max\{b_i\}$), and construct $S=\{x_i=(b_i,0),x_{n+i}=(b_i,\max\{b_i\})\}$ (see Figure \ref{fig:lower-bound-inf}).
\begin{figure}[!ht]
  \centering
    \includegraphics[width=0.6\textwidth]{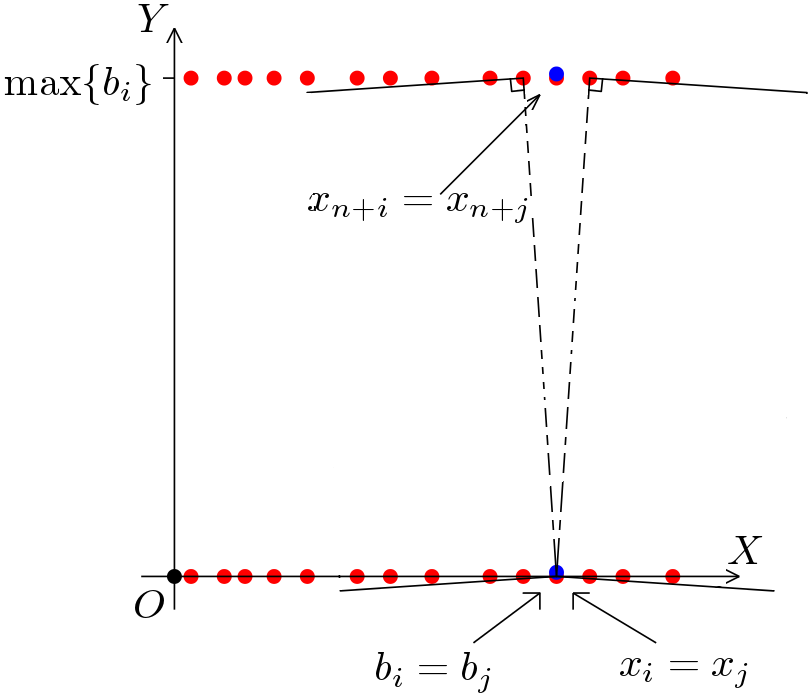}
  \caption{An illustration of $B$ and $S$ when $\beta=\infty$}
  \label{fig:lower-bound-inf}
\end{figure}

We select $q=(0,0)$, and define 
\begin{equation}
\label{eq:S-inf-for-lower-inf}
S_{\infty j}=\left\{x_k\in S\mid q \in S_{\infty}(x_j,x_k) \right\}, \: 1\leq j\leq 2n.
\end{equation}
From Definition \ref{def:beta-influence}, it can be verified that
\begin{equation*}
q\in S_{\infty}(x_j,x_k)\Leftrightarrow \max\{\angle qx_kx_j,\angle qx_jx_k\}\leq \frac{\pi}{2}.
\end{equation*} 
Suppose that every $x_i\in S; 1\leq i\leq n$ is a unique element. In this case,
\begin{equation}
\label{eq:s-inf-xi-unique}
S_{\infty i}=\{x_{n+1},x_{n+2},...,x_{n+i}\},\;S_{\infty(n+i)}=\{x_i,x_{i+1},...,x_n\}.
\end{equation}
From Equation \eqref{eq:S-inf-for-lower-inf}, the unnormalized form of Equation \eqref{eq:beta} for $\beta=\infty$ can be presented by:
\begin{equation}
\label{eq:SkD-inf-unnormalized}
\SkD_{\infty S}(q)={\vert S\vert\choose 2}\SkD_{\infty}(q;S)=\frac{1}{2}\sum_{1\leq j\leq 2n}\vert S_{\infty j}\vert.
\end{equation}
Equations \eqref{eq:s-inf-xi-unique} and \eqref{eq:SkD-inf-unnormalized} imply that
\begin{align*}
\label{eq:SkD-inf-unique-elements}
\SkD_{\infty S}(q)&=\frac{1}{2}\sum_{1\leq i\leq n}(\vert S_{\infty i}\vert+\vert S_{\infty (n+i)}\vert)\\&=\frac{1}{2}\sum_{1\leq i\leq n}(i+(n-i+1))=\frac{n(n+1)}{2}={n+1\choose 2}.
\end{align*}
Now assume that for some $i\neq k$, $x_i = x_k$ in $S$. Without loss of generality, suppose that $1\leq i<k\leq n$. In this case, 
\begin{equation}
\label{eq:s-inf-xi-non-unique1}
S_{\infty i}=S_{\infty k}=\{x_{n+1},x_{n+2},...,x_{n+i},x_{n+k}\}
\end{equation}
\begin{equation}
\label{eq:s-inf-xi-non-unique2}
S_{\infty (n+i)}=S_{\infty (n+k)}=\{x_k,x_i,x_{i+1},...,x_n\}.
\end{equation} 
As such, for $1\leq t\leq n$ and $t\notin \{i,k\}$, Equations \eqref{eq:s-inf-xi-unique}, \eqref{eq:SkD-inf-unnormalized}, \eqref{eq:s-inf-xi-non-unique1}, and \eqref{eq:s-inf-xi-non-unique2} can be used to compute $\SkD_{\infty S}(q)$ as follows.

\begin{align*}
&\SkD_{\infty S}(q)=\frac{1}{2}\sum_{1\leq j\leq 2n}\vert S_{\infty j}\vert\\&=\frac{1}{2}\left(\sum_t(\vert S_{\infty t}\vert+S_{\infty (n+t)}\vert)+\vert S_{\infty i}\vert+\vert S_{\infty(n+i)}\vert+\vert S_{\infty k}\vert+\vert S_{\infty(n+k)}\vert\right)\\&=\frac{1}{2}\left(\sum_t(t+n-t+1)+(i+1)+(n-i+2)+(k+1)+(n-k+2)\right)\\&=\frac{1}{2}\left(\sum_t(n+1)+2(n+3)\right)=\frac{1}{2}\left((n-2)(n+1)+2(n+1)+4\right)\\&=\frac{1}{2}(n(n+1)+4)={n+1\choose 2}+2
\end{align*}
For the general case of having $c$ duplicates among the elements of $S$,
\begin{equation}
\label{eq:SkD-inf-lower}
\SkD_{\infty S}(q)={n+1\choose 2}+2c.
\end{equation}
Therefore the elements of $S$ are unique if and only if $c=0$ in Equation~\eqref{eq:SkD-inf-lower}. The above results imply that the computation of $\beta$-skeleton depth, where $\beta=\infty$, requires $\Omega (n\log n)$ time.
%lower bound
%%---------------Chapter 7-------------------------
\chapter{Relationships and Experiments}
\label{ch:experiments}
In this chapter we study the relationships among different depth functions such as $\beta$-skeleton depth, halfspace depth, and simplicial depth in two different ways. First, we focus on the geometric properties of the influence regions. Second, the idea of fitting function is applied to approximate one data depth using another one. Our main motivation to study the relationships among different depth functions is derived from the complexity of computations, especially in higher dimensions. For example, computing the $\beta$-skeleton depth using brute force algorithm is much easier and relatively faster than computing most of the other depth functions such as halfspace depth and simplicial depth. Unlike halfspace depth and simplicial depth, the time complexity of $\beta$-skeleton depth grows linearly in the dimension. Recall that the time complexity of $\beta$-skeleton depth using a brute force algorithm in dimension $d$ is $O(dn^2)$. Whereas, the best known algorithm for computing the simplicial depth in the higher dimension $d$ is brute force which takes $O(n^{d+1})$ time. Computing the halfspace depth is an NP-hard problem when the dimension $d$ is a part of input. See Sections \ref{sec:halfspace}, \ref{sec:simplicial}, and \ref{sec:beta-skeleton}.
\section{Geometric Relationships}
\label{sec:geo-relation}
Some geometric properties related to the $\beta$-influence regions and simplices are explored in this section. These properties help to bound each one of $\beta$-skeleton depth and simplicial depth in terms of  the other one.  
\subsection{Convergence of $\beta$-skeleton Depth}
\begin{lemma}
\label{lm:beta-influence-containment}
For $\beta'>\beta\geq 1$ and $a,b\in \mathbb{R}^2$, $S_{\beta}(a,b)\subseteq S_{\beta'}(a,b)$, where the $\beta$-influence region $S_{\beta}(a,b)$ is the intersection of two disks $B(C_{ab\beta},R_{ab\beta})$ and $B(C_{ba\beta},R_{ba\beta})$, $C_{ab\beta}=(\beta/2)(a-b)+b$, and  $R_{ab\beta}=(\beta/2)d(a,b)$.
\end{lemma}
\begin{proof}
To prove $S_{\beta}(a,b)\subseteq S_{\beta'}(a,b)$ which is equivalent with Equation \eqref{eq:containment-beta-betaprime},
\begin{equation}
\label{eq:containment-beta-betaprime}
B(C_{ab\beta},R_{ab\beta})\cap B(C_{ba\beta},R_{ba\beta}) \subseteq B(C_{ab\beta'},R_{ab\beta'})\cap B(C_{ba\beta'},R_{ba\beta'})
\end{equation} it suffices to prove both inclusion relationships $B(C_{ab\beta},R_{ab\beta})\subseteq B(C_{ab\beta'},R_{ab\beta'})$ and $B(C_{ba\beta},R_{ba\beta})\subseteq B(C_{ba\beta'},R_{ba\beta'})$. We only prove the first one, and the second one can be proved similarly.
Suppose that $\beta<\beta'=\beta+\varepsilon; \varepsilon>0$. It is trivial to check that two disks $B(C_{ab\beta},R_{ab\beta})$ and $B(C_{ab\beta'},R_{ab\beta'})$ meet at $b$. See Figure~\ref{fig:beta-beta'-disks}. Let $t\neq b$ be an extreme point of $B(C_{ab\beta},R_{ab\beta})$. This implies that
\begin{align*}
d(t,C_{ab\beta})&=R_{ab\beta} \Leftrightarrow d(t,\frac{\beta(a-b)}{2}+b)=\frac{\beta(d(a,b))}{2}\\&\Leftrightarrow \left\Vert \frac{\beta(a-b)}{2}+(b-t)\right\Vert^2=\left(\frac{\beta(a-b)}{2}\right)^2\\&
\Leftrightarrow \Vert b-t\Vert^2-\beta(b-a)\cdot (b-t)=0.
\end{align*}
The last equality means that $(b-a)\cdot (b-t)\geq 0$. Hence,
\begin{align*}
&\Vert b-t\Vert^2-\beta(b-a)\cdot (b-t)=0\\
&\Leftrightarrow \Vert b-t\Vert^2-(\beta+\varepsilon)(b-a)\cdot (b-t)<0\\
&\Leftrightarrow \Vert b-t\Vert^2-\beta'(b-a)\cdot (b-t)<0\\
&\Leftrightarrow d(t,C_{ab\beta'})-R_{ab\beta'}<0\\&\Leftrightarrow t\in intB(C_{ab\beta'},R_{ab\beta'}).
\end{align*}
From the above calculations, every extreme point of $B(C_{ab\beta},R_{ab\beta})$ is an interior point of $B(C_{ab\beta'},R_{ab\beta'})$; therefore, $B(C_{ab\beta},R_{ab\beta})\subseteq B(C_{ab\beta'},R_{ab\beta'})$.  
\end{proof}
\begin{figure}[!ht]
  \centering
    \includegraphics[width=0.8\textwidth]{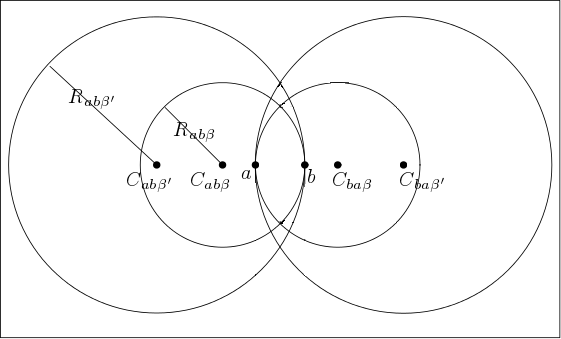}
  \caption{$S_{\beta}(a,b)$ and $S_{\beta'}(a,b)$ .}
  \label{fig:beta-beta'-disks}
\end{figure}
\begin{lemma}
Suppose that $\beta'>\beta\geq 1$. For a query point $q$ and given data set $S$ in $\mathbb{R}^2$, $\SkD_{\beta'}(q;S)\geq \SkD_{\beta}(q;S)$.
\label{lm:beta-betaprime} 
\end{lemma}
\begin{proof}
For any arbitrary pair of points $x_i$ and $x_j$ in $S$, Lemma \ref{lm:beta-influence-containment} implies that $S_{\beta}(x_i,x_j)\subseteq S_{\beta'}(x_i,x_j)$. Hence, Equation \eqref{eq:beta-betaprime} is sufficient to complete the proof.
\begin{equation}
\label{eq:beta-betaprime}
\SkD_{\beta'}(q;S)= \sum _{x_i,x_j \in S} {I(q \in S_{\beta'}(x_i, x_j))} \geq \sum _{x_i,x_j \in S} {I(q \in S_{\beta}(x_i, x_j))}= \SkD_{\beta}(q;S)
\end{equation} 
\end{proof}

\begin{definition}
\label{def:generic}
A query point $q$ is generic with respect to a data set $S=~\{x_1, x_2,\dots, x_n\}$ if for all $x_i,x_j\in S$, $q$ does not lie on the boundary of $S_{\infty}(x_i,x_j)$ or, on the line segment $\overline{x_ix_j}$. 
\end{definition}

\begin{lemma}
\label{lm:slab-sub-beta-influence}
Suppose that $S\subset\mathbb{R}^2$ is a given data set and $Q\subset\mathbb{R}^2$ is a set of generic query points. Assuming that $\mathcal{R}\subset\mathbb{R}^2$ is a large enough finite range that contains $S$ and $Q$, for two distinct elements $x_i$ and $x_j$ in $S$, and $q\in Q$,
\begin{equation}
\label{eq:slab-sub-beta-influence}
\exists\beta_{ij}<\infty; q\in S_{\infty}(x_i,x_j)\Rightarrow q\in S_{\beta_{ij}}(x_i,x_j).
\end{equation}
There also exists a $\beta^*<\infty$ that satisfies \eqref{eq:slab-sub-beta-influence} for all $\{x_i,x_j\}\subset S, x_i\neq x_j$. 
\end{lemma}
\begin{figure}[!ht]
  \centering
    \includegraphics[width=0.8\textwidth]{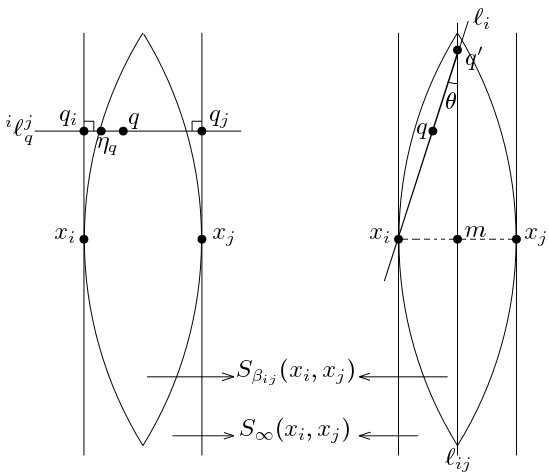}
  \caption{Slab and $\beta_{ij}$-influence region containing point $q$}
  \label{fig:slab-lens-containment-second-proof}
\end{figure}
\begin{proof}
Suppose that $q$ is a generic query point, $m$ is the middle point of $\overline{x_ix_j}$, and $\ell_{ij}$ is a line parallel to the boundaries of $S_{\beta_{ij}}(x_i,x_j)$, through $m$. We consider two lines $\ell{i}$, through $q$ and $x_i$, and $\ell_{j}$, through $q$ and $x_j$. Obviously, each of $\ell_{i}$ and $\ell_{j}$ intersects $\ell_{ij}$. Between these two intersections, let $q'$ be the farthest one. See Figure \ref{fig:slab-lens-containment-second-proof} (right). The desired value of $\beta_{ij}$ in Equation \eqref{eq:slab-sub-beta-influence} can be computed using Lemma \ref{lm:angle-beta-influence} as follows.
\begin{equation}
\label{eq:beta-another-solu-equation}
\theta=\cos^{-1}(1-\frac{1}{\beta_{ij}})\Rightarrow \cos(\theta)=1-\frac{1}{\beta_{ij}}\Rightarrow \beta_{ij}=\frac{1}{1-\cos(\theta)}
\end{equation}
where $\theta$ is the smaller angle at $q'$. Since $q$ is a generic query point, $\cos(\theta)\neq 1$. Equation \eqref{eq:beta-another-solu-equation} provides the desired value for $\beta_{ij}$ in Equation \eqref{eq:slab-sub-beta-influence}. The proof is complete because
\begin{equation}
\label{eq:beta-star}
\forall \{x_i,x_j\}\subset S; q\in S_{\infty}(x_i,x_j) \Rightarrow q\in S_{\beta
^*}(x_i,x_j),\; \beta^*=\max_{ij}\{\beta_{ij}\}.
\end{equation}
\paragraph{Another proof:} Suppose that $q$ is a generic query point, and $q\in S_{\infty}(x_i,x_j)$. Let $^i\ell^j_q$ be a line passing through $q$, and perpendicular to the boundaries of slab $S_{\infty}(x_i,x_j)$ in points $q_i$ and $q_j$. See Figure \ref{fig:slab-lens-containment-second-proof} (left). Every value of $\beta_{ij}$ that meets the following requirements can be considered as the desired $\beta_{ij}$.
\begin{itemize}
\item $\beta_{ij}$ is large enough such that $S_{\beta_{ij}}(x_i,x_j)$ is cut by $^i\ell^j_q$.
\item $\min\{\Vert q-q_i\Vert,\Vert q-q_j\Vert\}\geq \min\{\Vert q_i-\eta_q\Vert,\Vert q_j-\eta_q \Vert\}$, where $\eta_q$ is a fixed intersection point between $^i\ell^j_q$ and $S_{\beta_{ij}}(x_i,x_j)$.
\end{itemize}
The value of $\beta^*$ can be chosen as in Equation~\eqref{eq:beta-star}.
\end{proof}
\begin{theorem}
\label{thrm:convergence-SKD-beta}
For data set $S$ and a generic query point $q$ given in a large enough finite range $\mathcal{R}\subset\mathbb{R}^2$, the $\beta$-skeleton depth functions converge. In other words,
\begin{equation}
\label{eq:convergence-SKD-beta}
\exists\beta^*<\infty; \forall \beta\geq\beta^*,\; \SkD_{\beta}(q;S)= S_{\infty}(q;S). 
\end{equation}
\begin{proof}
It is enough to prove that there exists a $\beta^*<\infty$ such that
\begin{equation}
\label{eq:influnce-beta*-infty-equivalent}
\forall\{x_i,x_j\}\subset S, x_i\neq x_j, q\in S_{\beta^*}(x_i,x_j)\Leftrightarrow q\in S_{\infty}(x_i,x_j).
\end{equation}
$\Rightarrow)$ Lemma \ref{lm:beta-influence-containment} is obviously enough because
\[
\forall\{x_i,x_j\}\subset S,\forall \beta^*<\infty, S_{\beta^*}(x_i,x_j)\subset S_{\infty}(x_i,x_j).
\]
$\Leftarrow)$ To prove this direction, Lemma \ref{lm:slab-sub-beta-influence} suggests that it is enough to choose
\[\beta^*=\max_{ij}\{\beta_{ij}\}.\]
\end{proof}
\end{theorem}
\subsection{$\beta$-skeleton Depth versus Simplicial Depth}
First, we study the relationship between spherical ($\beta$-skeleton, $\beta=1$) depth and simplicial depth. From Lemma \ref{lm:beta-betaprime}, we can generalize the obtained results for every value of $\beta$.   
\begin{definition}
\label{def:bin-sin}
For a point $q \in \mathbb{R}^2$ and a data set $S=\{x_1,...,x_n\}\subset \mathbb{R}^2$, we define $B_{in}(q;S)$ to be the set of all closed spherical influence regions, out of $n \choose 2$ possible of them, that contain $q$. We also define $S_{in}(q;S)$ to be the set of all closed triangles, out of $n \choose 3$ possible defined by $S$, that contain $q$.
\end{definition}
\begin{lemma}
Suppose that $q$ is a point inside $\Conv(S)$, where $S\subset \mathbb{R}^2$ is a given data set. $q$ is covered by the union of spherical influence regions defined by $S$.
\label{lm:CH-q-GC}
\end{lemma}
\begin{proof}
Let $H=\Conv(S)$. By Caratheodory's theorem \cite{caratheodory1907variabilitatsbereich}, there is at least one triangle, defined by the vertices of $H$, that contains $q$. We prove that the union of the spherical influence regions defined by such triangle contains $q$. See Figure~\ref{fig:Triangleabc}. This statement can be proved by contradiction. Suppose that $q$ is covered by none of $\Sph(a,b)$, $\Sph(a,c)$, and $\Sph(b,c)$. Therefore, Theorem~\ref{thrm:point-circle} implies that none of the angles $\angle aqb$, $\angle aqc$, and $\angle bqc$ is greater than or equal to $\frac{\pi}{2}$ which is a contradiction because at least one of these angles should be at least $\frac{2\pi}{3}$ in order to get $2\pi$ as their sum.
\end{proof}
\begin{lemma} Suppose that $S=\{a,b,c\}$ is a set of points in $\mathbb{R}^2$. For every $q\in \mathbb{R}^2$, if $\vert S_{in}(q;S)\vert =1$, then $\vert B_{in}(q;S)\vert\geq 2$.
\\Another form of Lemma~\ref{lm:triangle-ball} is that if $q \in \bigtriangleup abc$, then $q$ falls inside at least two spherical influence regions out of $\Sph(a,b)$, $\Sph(c,b)$, and $\Sph(a,c)$. The equivalency between these two forms of the lemma is clear. We prove the first one.
\label{lm:triangle-ball}
\end{lemma}
\begin{proof}
From Lemma~\ref{lm:CH-q-GC}, $\vert B_{in}(q;S)\vert\geq 1$. Suppose that $\vert B_{in}(q;S)\vert =1$. If $q$ is one of the vertices of $\bigtriangleup abc$, it is clear that $\vert B_{in}(q;S)\vert \geq 2$. Without loss of generality, we suppose that $q$ falls in $int\Sph(a,b)$. For the rest of the proof, we focus on the relationships among the angles $\angle aqb$, $\angle cqa $, and $\angle cqb$ (see Figure~\ref{fig:Triangleabc}). Since $q$ is inside $\bigtriangleup abc$, $\angle aqb \leq \pi$. Consequently, at least one of $\angle cqa$ and $\angle cqb$ is greater than or equal to $\pi/2$. So, Theorem~\ref{thrm:point-circle} implies that $q$ is in at least one of $int\Sph(a,c)$ and $int\Sph(b,c)$.  Hence, $\vert B_{in}(q;S)\vert =1$ contradicts $\vert S_{in}(q;S)\vert =1$ which means that $\vert B_{in}(q;S)\vert \geq 2$. As an illustration, in Figure~\ref{fig:Triangleabc}, for the points in the hatched area $\vert B_{in}(q;S)\vert =3$.
\end{proof}
\begin{figure}[!ht]
  \centering
    \includegraphics[width=0.6\textwidth]{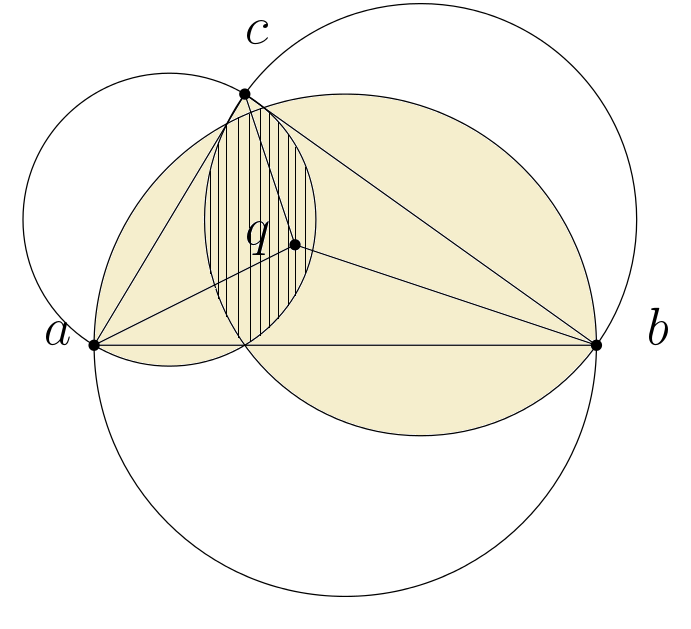}
  \caption{Triangle $abc$ and its corresponding spherical influence regions containing $q$}
  \label{fig:Triangleabc}
\end{figure}
\begin{lemma} For a data set $S=\{x_1,...,x_n\} \subset \mathbb{R}^2$, 
\[
\frac{\vert B_{in}(q;S)\vert }{\vert S_{in}(q;S)\vert}\geq \frac{2}{n-2}.
\]
\label{lm:bin-sin}
\end{lemma}
\begin{proof} Suppose that $\Sph(x_i,x_j)\in B_{in}(q;S)$. There exist at most $(n-2)$ triangles in $S_{in}(q;S)$ such that $\overline{x_ix_j}$ is an edge of them. We consider $ \bigtriangleup x_ix_jx_k$ to be one of such triangles (see Figure~\ref{fig:bin-sin} as an illustration). Referring to Lemma~\ref{lm:triangle-ball}, $q$ belongs to at least one of $\Sph(x_i,x_k)$ and $\Sph(x_j,x_k)$. Similarly, there exist at most $(n-2)$ triangles in $S_{in}(q;S)$ such that $x_ix_k$ (respectively $x_jx_k$) is an edge of them. In the process of computing $\vert S_{in}(q;S)\vert$, triangle $\bigtriangleup x_ix_jx_k$ is counted at least two times, once for $\Sph(x_i,x_j)$ and another time for $\Sph(x_i,x_k)$ (or $\Sph(x_j,x_k)$ ). Consequently, for every sphere area in $B_{in}(q;S)$, there exist at most $\frac{(n-2)}{2}$ distinct triangles, triangles with only one common side, in $S_{in}(q;S)$. As a result, Equation~\eqref{eq:bin-sin} can be obtained. 
\begin{equation}
\frac{(n-2)}{2}\vert B_{in}(q;S)\vert \geq \vert S_{in}(q;S)\vert\Rightarrow \frac{\vert B_{in}(q;S)\vert}{\vert S_{in}(q;S)\vert}\geq \frac{2}{(n-2)}
\label{eq:bin-sin}
\end{equation}
\end{proof}
 \begin{figure}[h!]
  \centering
  \includegraphics[width=0.55\textwidth]{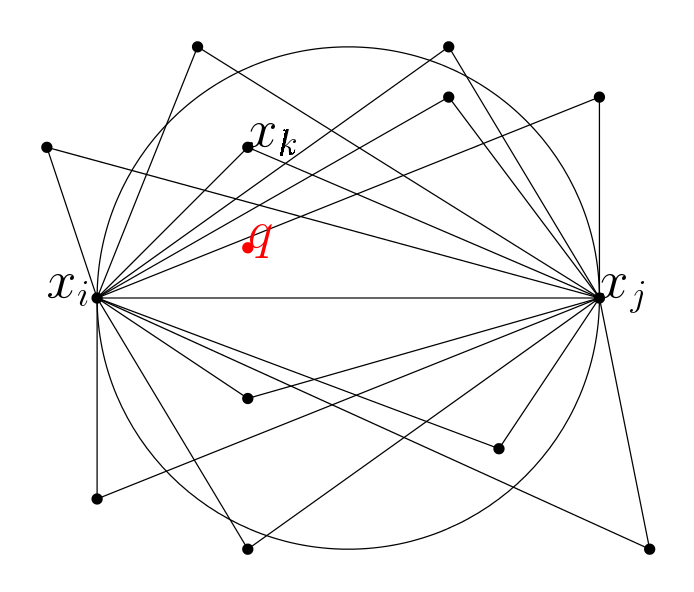}
  \caption{$\Sph(x_i,x_j)$ and all triangles, with edge $x_jx_j$, containing $q$}
  \label{fig:bin-sin}
\end{figure}
\begin{theorem}
For a data set $S=\{x_1,...,x_n\}$ and a query point $q$ in $\mathbb{R}^2$, $\SphD(q;S)\geq \frac{2}{3} \SD(q;S)$.
\label{thrm:Sph-Simp}
\end{theorem}
\begin{proof}
From Definition~\ref{def:bin-sin} the definitions of spherical depth and simplicial depth, \eqref{eq:def-sph-by-bin} and \eqref{eq:def-sim-by-sin} can be easily verified.
\begin{equation}
\label{eq:def-sph-by-bin}
\SphD(q;S)=\frac{1}{{n \choose 2}}\vert B_{in}(q;S)\vert
\end{equation}
\begin{equation}
\label{eq:def-sim-by-sin}
\SD(q;S)=\frac{1}{{n \choose 3}}\vert S_{in}(q;S)\vert
\end{equation}
Using these two equations, the ratio of spherical depth and simplicial depth can be calculated as follows: 
\begin{equation}
\label{eq:first-ratio1}
\frac{\SphD(q;S)}{\SD(q;S)}=\frac{{n \choose 3}}{{n \choose 2}} \frac{\vert B_{in}(q;S)\vert}{\vert S_{in}(q;S)\vert}= \frac{(n-2)\vert B_{in}(q;S)\vert}{3\vert S_{in}(q;S)\vert}.
\end{equation}
Equation~\eqref{eq:first-ratio1} and Lemma~\ref{lm:bin-sin} imply that 
\[
\frac{\SphD(q;S)}{\SD(q;S)}\geq \frac{2(n-2)}{3(n-2)}=\frac{2}{3} \Rightarrow \SphD(q;S) \geq \frac{2}{3} \SD(q;S).
\]
\end{proof}
\begin{theorem}
Suppose that $S$ is a given data set consisting of $n$ points in general position in $\mathbb{R}^2$. For $q \in \mathbb{R}^2$ and $\beta \geq 1$, $SkD_{\beta}(q;S)\geq \frac{2}{3} \SD(q;S)$.
\label{thrm:Beta-Simp}
\end{theorem}
\begin{proof}
The proof can be derived from Lemma \ref{lm:beta-betaprime} and Theorem \ref{thrm:Sph-Simp}.  
\end{proof}
\section{Relationships via Dissimilarity Measures}
\label{sec:dissimilarity-measures}
In this section we define two dissimilarity measures between every pair of depth functions, and approximate these depth functions by one another with a certain amount of error. A notable application of this approximation technique is that we can approximate the halfspace depth, which is an NP-hard problem in higher dimension $d$, by the $\beta$-skeleton depth which takes only $O(dn^2)$ time (see Section \ref{sec:approx-half-space}). The dissimilarity measures in this technique are defined based on the concepts of \emph{fitting function} and \emph{Hamming distance}. We train the halfspace depth function by the $\beta$-skeleton depth values obtaining from a given data set. The goodness of approximation can be determined using the dissimilarity measures and the sum of squares of error values.
\subsection{Fitting Function and Dissimilarity Measure}
To determine the dissimilarity between two vectors $U=(u_1 ,..., u_n)$ and $V=(v_1, ..., v_n)$, the idea of fitting functions can be applied. Considering the goodness measures of fitting functions in Section \ref{sec:goodness-of-fitting}, assume that $f$ is the best function fitted to $U$ and $V$. In other words, $u_i=f(v_i)\pm \delta_i$; $\delta_i\geq 0$ and $1\leq i\leq n$. Let $\xi_i=u_i-\overline{U}$, where $\overline{U}$ is the average of $u_i$ ($1\leq i\leq n$). In Equation \eqref{eq:e-distance}, we define $d_E(U,V)$, as a function of $\delta_i$ and $\xi_i$, to measure the dissimilarity between $U$ and $V$.
\begin{equation}
\label{eq:e-distance}
d_E(U,V)=1-r^2,
\end{equation}
where $r^2$ is the coefficient of determination. We recall from Equation \eqref{eq:r-squared} that
\begin{equation*}
r^2=\dfrac{\sum\limits_{i=1}^{n}(\xi_i^2-\delta_i^2)}{\sum\limits_{i=1}^{n}\xi_i^2}.
\end{equation*}  
Since $r^2\in[0,1]$, $d_E(U,V)\in [0,1]$. A smaller value of $d_E(U,V)$ represents more similarity between $U$ and~$V$. 
\subsection{Dissimilarity Measure Between two Posets}
The idea of defining the following distance comes from the proposed structural dissimilarity measure between posets in \cite{fattore2014measuring}. Let $\mathbb{P}=\{P_t=(S,\preceq_t)\vert t\in \mathbb{N}\}$ be a finite set of posets, where $S=\{x_1, ..., x_n\}$ is a given data set. For $P_k\in\mathbb{P}$ we define a matrix $M^k_{n\times n}$ by:
\[
M^k_{ij}= \begin{cases}
1\;\;\;\;\;x_i \preceq_k x_j \\
0\;\;\;\;\; \text{otherwise.}
\end{cases}
\]
We use the notation of $d_c(P_f,P_g)$ to define the dissimilarity between two posets $P_f,P_g\in \mathbb{P}$ as follows:
\begin{equation}
\label{eq:h-distance}
d_c(P_f,P_g)=\dfrac{\sum\limits_{i=1}^{n}\sum\limits_{j=1}^{n} \vert M^f_{ij}-M^g_{ij} \vert}{n^2-n}
\end{equation}
It can be verified that $d_c(P_f,P_g)\in [0,1]$, where the closer value to $1$ means the less similarity between $P_f$ and $P_g$. This measure of similarity is a metric on $\mathbb{P}$ because for all $P_f,P_g,P_h \in \mathbb{P}$,
\begin{itemize}
\item $d_c(P_f,P_g)\geq 0$
\item $d_c(P_f,P_g)=0 \Leftrightarrow P_f=P_g$
\item $d_c(P_f,P_g)=d_c(P_g,P_f)$
\item $d_c(P_f,P_h)\leq d_c(P_f,P_g)+ d_c(P_g,P_h).$
\end{itemize}
Proving these properties is straightforward. We prove the last property which is less trivial.
\begin{align*}
d_c(P_f,P_h)&=\dfrac{\sum\limits_{i=1}^{n}\sum\limits_{j=1}^{n} \vert M^f_{ij}-M^h_{ij} \vert}{n^2-n} = \dfrac{\sum\limits_{i=1}^{n}\sum\limits_{j=1}^{n} \vert (M^f_{ij}-M^g_{ij})+(M^g_{ij}-M^h_{ij})\vert}{n^2-n}
\\&\leq \dfrac{\sum\limits_{i=1}^{n}\sum\limits_{j=1}^{n} (\vert M^f_{ij}-M^g_{ij}\vert +\vert M^g_{ij}-M^h_{ij}\vert)}{n^2-n}
\\&= \dfrac{\sum\limits_{i=1}^{n}\sum\limits_{j=1}^{n} \vert M^f_{ij}-M^g_{ij} \vert}{n^2-n} + \dfrac{\sum\limits_{i=1}^{n}\sum\limits_{j=1}^{n} \vert M^g_{ij}-M^h_{ij} \vert}{n^2-n}\\&=d_c(P_f,P_g)+d_c(P_g,P_h).
\end{align*}
\section{Approximation of Halfspace Depth}
\label{sec:approx-half-space}
We use the proposed method in Section \ref{sec:dissimilarity-measures} to approximate the halfspace depth. Motivated by statistical applications and machine learning techniques, we train the halfspace depth function using the values of $\beta$-skeleton depth. Among all depth functions, the $\beta$-skeleton depth is chosen because it is easy to compute and its time complexity, i.e. $O(dn^2)$, grows linearly with the dimension $d$. 
\subsection{Approximation of Halfspace Depth and Fitting Function}
\label{sec:approx1}
Suppose that $S=\{x_1, ...,x_n\}$ is a given data set. By choosing some subsets of $S$ as training samples, we consider the problem of learning the halfspace depth function using the $\beta$-skeleton depth values. In particular, we use the cross validation and information criterion to obtain the best function $f$ such that $\HD(x_i;S)=f(\SkD_{\beta}(x_i;S))\pm \delta_i$. The function $f$ can be considered as an approximation function for halfspace depth. Finally, $d_E(\HD,\SkD_{\beta})$ as the error of approximation can be computed using Equation \eqref{eq:e-distance}.
\subsection{Approximation of Halfspace Depth and Poset Dissimilarity}
\label{sec:approx2}
In some applications, the structural ranking among the elements of $S$ is more important than the depth value of single points. Let $S=\{x_1, ..., x_n\}$ be a given data set and $D$ be a depth function. Applying $D$ on $x_i$ with respect to $S$ generates a poset. In fact, $P_D=(D(x_i;S),\leq)$ is a chain because for every $x_i, x_j \in S$, the values of $D(x_i;S)$ and $D(x_j;S)$ are comparable. For halfspace depth and $\beta$-skeleton depth, their dissimilarity measure of rankings can be obtained by Equation \eqref{eq:h-distance} as follows:
\[
d_c(\HD,\SkD_{\beta})=\dfrac{\sum\limits_{i=1}^{n}\sum\limits_{j=1}^{n} \vert M^{\HD}_{ij}-M^{\SkD_{\beta}}_{ij} \vert}{n^2-n}.
\] 
The smaller value of $d_c(\HD,\SkD_{\beta})$, the more similarity between $\HD$ and $\SkD_{\beta}$ in ordering the elements of $S$.
\\\\ In sections \ref{sec:approx1} and \ref{sec:approx2}, instead of $\beta$-skeleton, any other depth function can be considered to approximate halfspace depth. Considering any other depth function, we can compute the goodness of approximation using dissimilarity measures $d_E$ and $d_c$.     
\section{Experimental Results}
\label{sec:exp}
In this section we provide some experimental results to support Sections \ref{sec:geo-relation}, \ref{sec:dissimilarity-measures}, and \ref{sec:approx-half-space}. The results are summarized in some tables and graphs presented in Sections \ref{sec:approx1} and \ref{sec:approx2}. To obtain these results, we computed the depth functions and their relationships for three sets $Q_1$, $Q_2$, and $Q_3$ of planar query points with respect to data sets $S_1$, $S_2$, and $S_3$ of planar points, respectively. The cardinalities of $Q_i$ and $S_i$ are as follows: $|Q_1|=100$, $|S_1|=750$, $|Q_2|=1000$, $|S_2|=2500$, $|Q_3|=2500$, $|S_3|=10000$. The elements of $Q_i$ and $S_i$ ($i=1,2,3$) are some randomly generated points (double precision floating point) within the square $\{(x,y)| x,y \in [-10,10]\}$. The following lines of code, in MATLAB, are used to generate the elements of $Q_i$ and $S_i$.

\begin{verbatim}
% To Generate the Sets of Random Data Points and  Query Points:
n = 3; % the number of decimal places
d = 2; % dimension of points
kS = 10000; % number of data points
kQ = 2500; % number of query points
% [l_range,u_range] is the interval where the
% random points are generated within
l_range = -10; u_range = 10;
S = randi([l_range,u_range]*10^n,[kS,d])/10^n;
Q = randi([l_range,u_range]*10^n,[kQ,d])/10^n;

\end{verbatim}
The implementations are done in Java. The source codes and detailed results are publicly available at \url{https://github.com/RasoulShahsavari/Data-Depth-Source-Codes}.
\subsection{Experiments for Geometric Relationships}
\label{sec:exp.for.geo}
To support the obtained relationships in Section \ref{sec:geo-relation}, we compute the spherical depth, lens depth, and the simplicial depth of the points in three random sets $Q_1$, $Q_2$, and $Q_3$ with respect to data sets $S_1$, $S_2$, and $S_3$, respectively. The results of our experiments are summarized in Table~\ref{table:results-random-points}. Every cell in the table represents the corresponding depth of query points in $Q_i$ with respect to data set $S_i$. As can be seen, there are some gaps between obtained experimental bounds for random points and the theoretical bounds in Theorem~\ref{thrm:Beta-Simp} and Lemma~\ref{lm:beta-betaprime}. For example, the experiments suggests $1.2$ as a lower bound for $\LD/SphD$, whereas Theorem \ref{thrm:Beta-Simp} introduces the lower bound $1$. More research on this topic is needed to figure out if the real bounds are closer to the experimental bounds or to the current theoretical bounds.
\begin{table}[!ht]
\begin{center}
\begin{tabular}{|l||l|l||l|l||l|l|}
\hline
 &\multicolumn{2}{l||}{$(Q_1;S_1)$}&\multicolumn{2}{l||}{$(Q_2;S_2)$}&\multicolumn{2}{l|}{$(Q_3;S_3)$}\\
\cline{2-7}
 &Min& Max&Min&Max&Min&Max\\
\hline\hline
$\SD$&0.00&0.25&0.00&0.25&0.00&0.24\\
\hline
$\SphD$&0.01&0.50&0.00&0.50&0.00&0.50\\
\hline
$\LD$&0.05&0.61&0.05&0.61&0.04&0.61\\
\hline
$\frac{\SphD}{\SD}$&2.00&$\infty$&2.00&$\infty$&2.03&$\infty$\\
\hline
$\frac{\LD}{\SD}$&2.43&$\infty$&2.44&$\infty$&2.44&$\infty$\\
\hline
$\frac{\LD}{\SphD}$&1.21&8.11&1.22&23.16&1.22&157.16 \\
\hline
\end{tabular}
\end{center}
\caption{Summary of experiments for the geometric relationships among simplicial depth ($\SD$), spherical depth ($\SphD$), and lens depth ($\LD$)}
\label{table:results-random-points}
\end{table}

\subsection{Experiments for Approximations}
\label{sec:exp.for.approx}
In this section some experimental results are provided to support Theorem~\ref{thrm:convergence-SKD-beta}, and our proposed method of approximation in Section \ref{sec:approx-half-space}. We compute the planar halfspace depth and planar $\beta$-skeleton depth of $q\in Q_2$ with respect to $S_2$ for different values of $\beta$, where $Q_2$ and $S_2$ are introduced in Section \ref{sec:exp}. The results of our experiments are summarized in Tables~\ref{tbl:results-beta-converge}, \ref{tbl:results-halfspce-approx}, and \ref{tbl:results-halfspce-approx-power}. For each row of these three tables, a plot labelled by the corresponding fitting function is provided (Figures \ref{fig:1-inf} - \ref{fig:10000-h-pow}). As can be seen in the last two rows of Table \ref{tbl:results-beta-converge}, considering $\beta^*=1000$, Theorem \ref{thrm:convergence-SKD-beta} is supported by these experiments. From Table \ref{tbl:results-halfspce-approx}, it can be seen that the halfspace depth can be approximated by a quadratic function of the $\beta$-skeleton depth with relatively small values of $d_E(\SkD,\HD)\approxeq 0.0186$ and $d_C(\SkD,\HD)~\approxeq 0.0373$. In particular, $\HD\approxeq3.8883(\SkD_{\beta}-0.3004)^2$ if $\beta\to \infty$. The results in Table~\ref{tbl:results-halfspce-approx-power} indicate that the halfspace depth function can also be approximated via a power model. In particular, $\HD\approxeq(\SkD_{\beta}+0.2255)^{5.8850}$ if $\beta\to \infty$. Referring to the results in Tables \ref{tbl:results-halfspce-approx} and \ref{tbl:results-halfspce-approx-power} and Figures \ref{fig:1-h} - \ref{fig:10000-h-pow}, the power model suggests a better fit both theoretically and visually. For example, the values of $d_E$ in the Table \ref{tbl:results-halfspce-approx-power} are smaller than the corresponding values in Table \ref{tbl:results-halfspce-approx}. Furthermore, from Figures \ref{fig:10000-h} and \ref{fig:10000-h-pow}, it can be seen that the power model captures the curvature of the data better than the quadratic model does. We recall that $d_C$ represents the structural dissimilarity between the exact values of two depth functions; hence, the values of $d_C$ are not affected by the choice of fitting model. In short, our experimental results on approximating the halfspace depth by $\beta$-skeleton depth show that the power model behaves slightly better than the quadratic model. However, depending on the application, the quadratic model might be preferred because of its simplicity.  

\begin{table}[!ht]
\begin{center}
\begin{tabular}{|l|l|l|l|l|}
\hline
$x$&$\SkD_{\infty}=f(x)$&$d_E(x,\SkD_{\infty})$&$d_c(x,\SkD_{\infty})$&Figure\\
\hline\hline
$\SkD_{1}$&$0.9423x+0.2885$&$0.0086$&$0.012$&\ref{fig:1-inf}\\
\hline
$\SkD_{2}$&$0.8580x+0.2304$&$0.0034$&$0.008$&\ref{fig:2-inf}\\
\hline
$\SkD_{3}$&$0.8598x+0.1922$&$0.0016$&$0.005$&\ref{fig:3-inf}\\
\hline
$\SkD_{1000}$&$0.9985x+0.0016$&$0.0000$&$0.000$&\ref{fig:1000-inf}\\
\hline
$\SkD_{10000}$&$1.0000x+0.0000$&$0.0000$&$0.000$&\ref{fig:10000-inf}\\
\hline
\end{tabular}
\end{center}
\caption{Summary of experiments for the convergence of $\beta$-skeleton depth}
\label{tbl:results-beta-converge}
\end{table}

\begin{table}[!ht]
\begin{center}
\begin{tabular}{|l|l|l|l|l|}
\hline
$x$&$\HD=f(x)$&$d_E(x,\HD)$&$d_c(x,\HD)$&Figure\\
\hline\hline
$\SkD_{1}$&$2.9275(x+0.0152)^2$&$0.0127$&$0.0306$&\ref{fig:1-h}\\
\hline
$\SkD_{2}$&$2.5522(x-0.0608)^2$&$0.0141$&$0.0331$&\ref{fig:2-h}\\
\hline
$\SkD_{3}$&$2.6569(x-0.1117)^2$&$0.0156$&$0.0343$&\ref{fig:3-h}\\
\hline
$\SkD_{1000}$&$3.8765(x-0.2992)^2$&$0.0186$&$0.0373$&\ref{fig:1000-h}\\
\hline
$\SkD_{10000}$&$3.8883(x-0.3004)^2$&$0.0186$&$0.0373$&\ref{fig:10000-h}\\
\hline
\end{tabular}
\end{center}
\caption{Summary of experiments for the halfspace depth approximation via a quadratic model}
\label{tbl:results-halfspce-approx}
\end{table}

\begin{table}[!ht]
\begin{center}
\begin{tabular}{|l|l|l|l|l|}
\hline
$x$&$\HD=f(x)$&$d_E(x,\HD)$&$d_c(x,\HD)$&Figure\\
\hline\hline
$\SkD_{1}$&$(x+0.4622)^{4.8174}$&$0.0130$&$0.0306$&\ref{fig:1-h-pow}\\
\hline
$\SkD_{2}$&$(x-0.3465)^{4.4263}$&$0.0120$&$0.0331$&\ref{fig:2-h-pow}\\
\hline
$\SkD_{3}$&$(x+0.3089)^{4.5734}$&$0.0120$&$0.0343$&\ref{fig:3-h-pow}\\
\hline
$\SkD_{1000}$&$(x+0.2259)^{5.8737}$&$0.0122$&$0.0373$&\ref{fig:1000-h-pow}\\
\hline
$\SkD_{10000}$&$(x+0.2255)^{5.8850}$&$0.0123$&$0.0373$&\ref{fig:10000-h-pow}\\
\hline
\end{tabular}
\end{center}
\caption{Summary of experiments for the halfspace depth approximation via a power model}
\label{tbl:results-halfspce-approx-power}
\end{table}

\begin{figure}[!ht]
  \centering
    \includegraphics[width=0.9\textwidth]{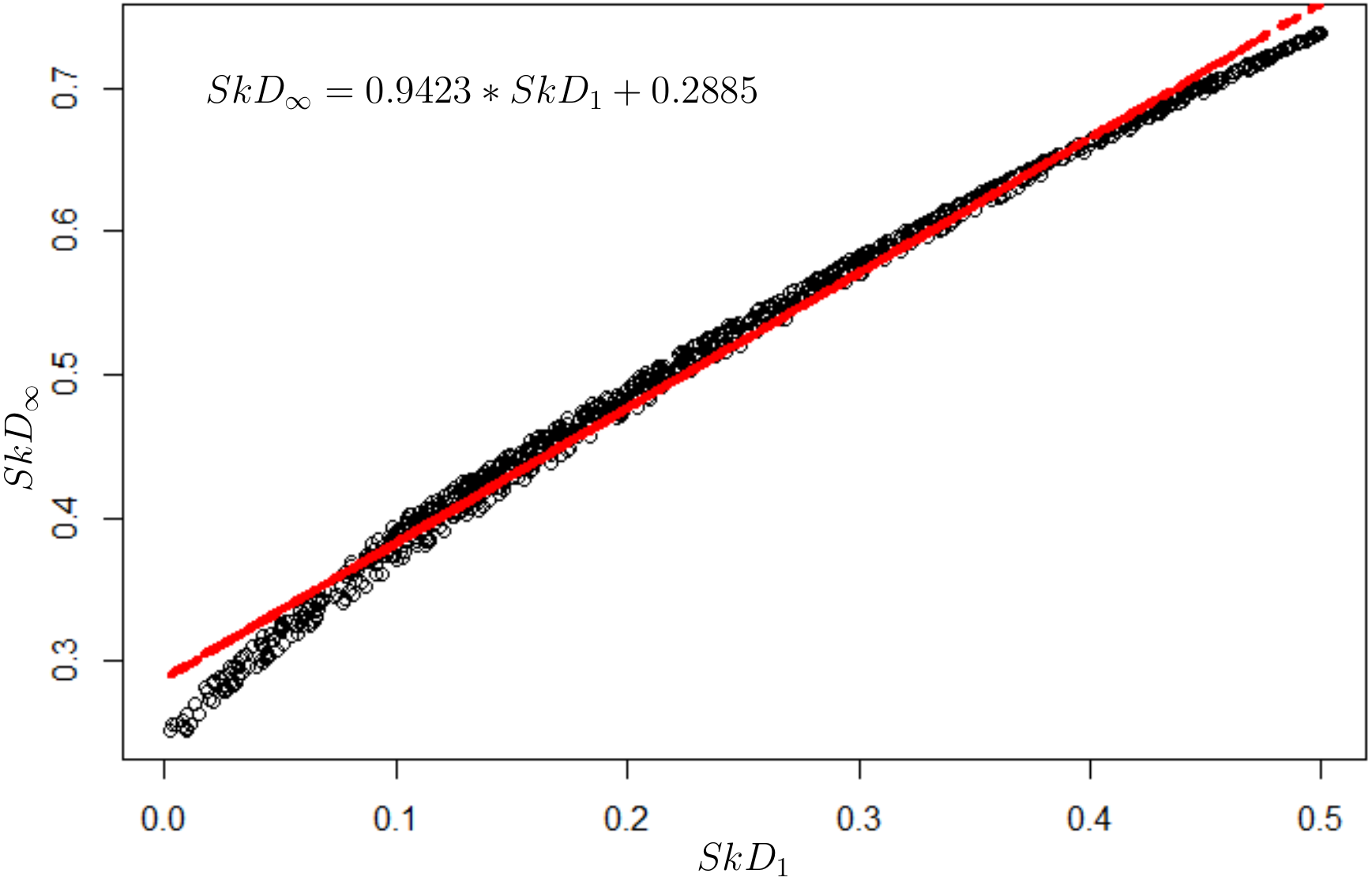}
  \caption{Fitting a linear model to $\SkD_{\infty}$ and $\SkD_1$.}
  \label{fig:1-inf}
\end{figure}

\begin{figure}[!ht]
  \centering
    \includegraphics[width=0.9\textwidth]{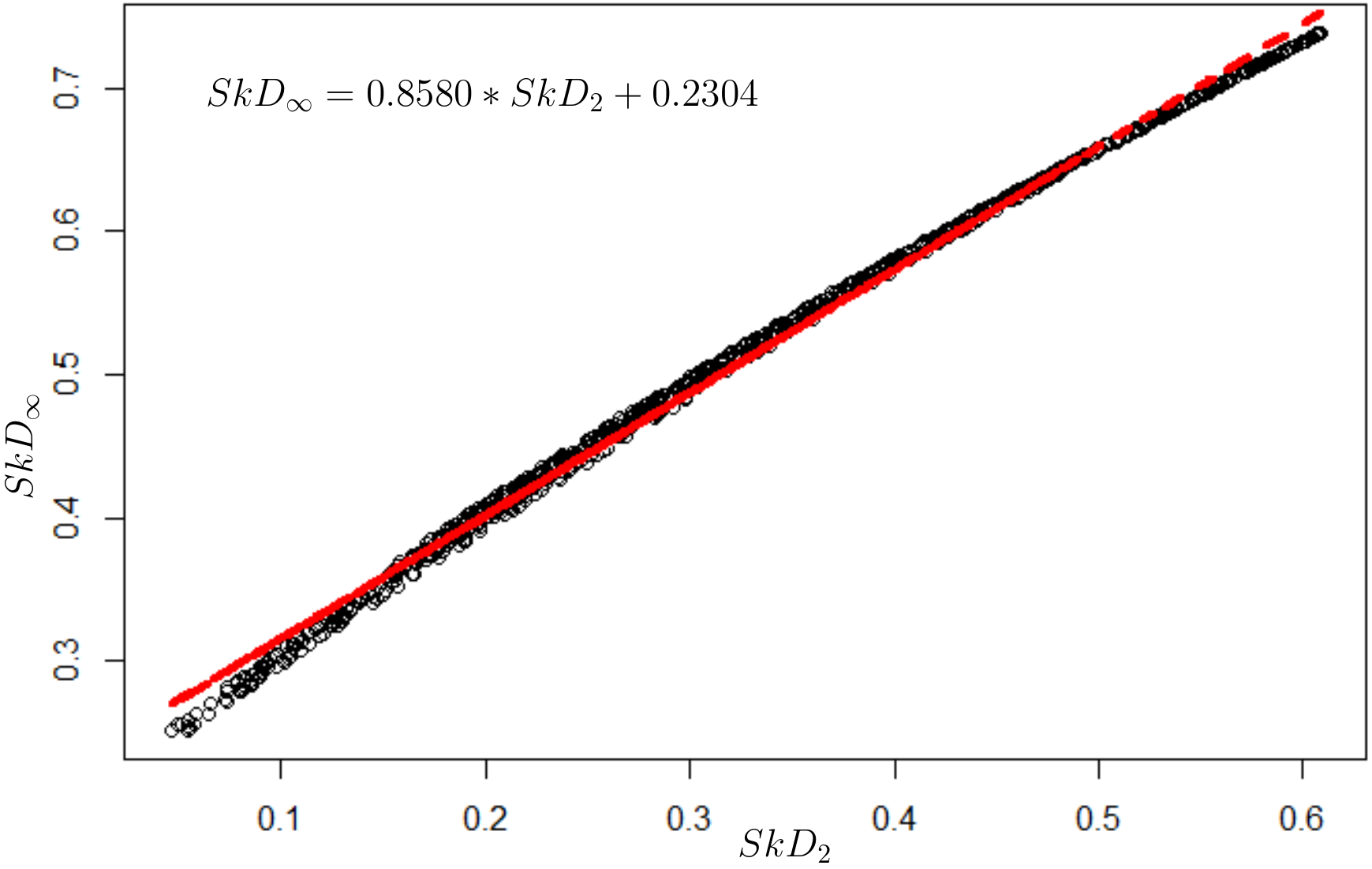}
  \caption{Fitting a linear model to $\SkD_{\infty}$ and $\SkD_2$ .}
  \label{fig:2-inf}
\end{figure}

\begin{figure}[!ht]
  \centering
    \includegraphics[width=0.9\textwidth]{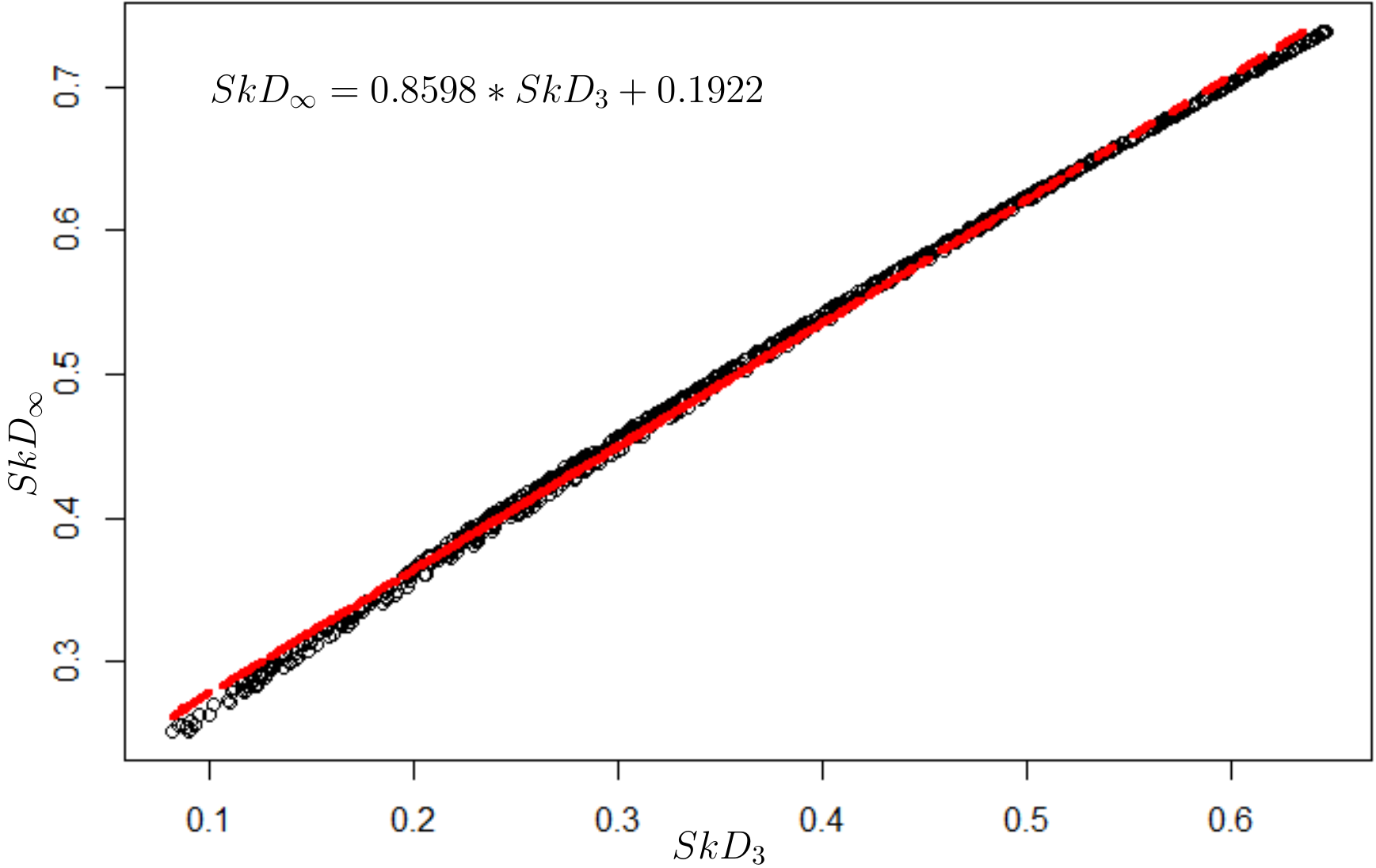}
  \caption{Fitting a linear model to $\SkD_{\infty}$ and $\SkD_3$.}
  \label{fig:3-inf}
\end{figure}

\begin{figure}[!ht]
  \centering
    \includegraphics[width=0.9\textwidth]{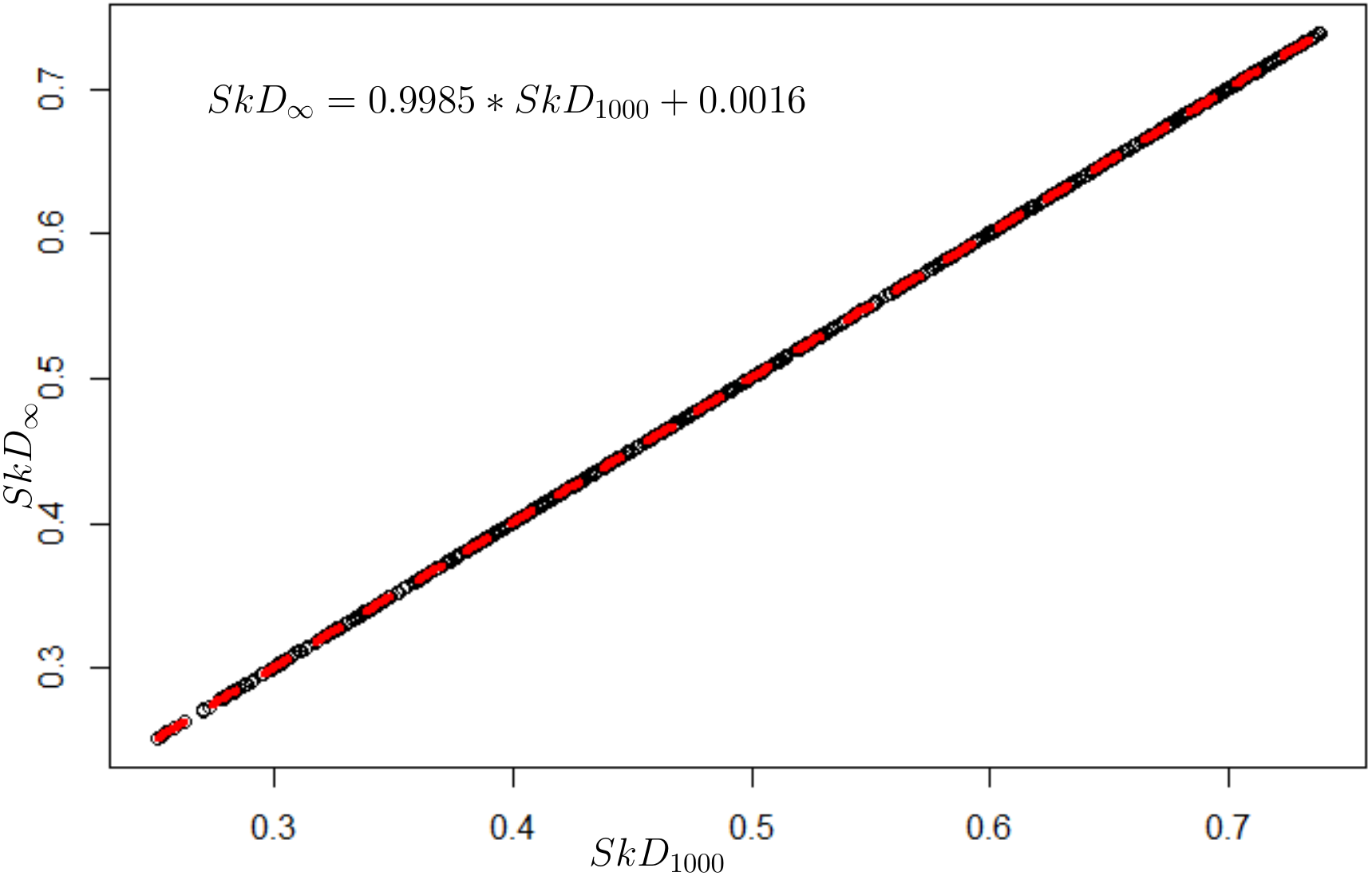}
  \caption{Fitting a linear model to $\SkD_{\infty}$ and $\SkD_{1000}$.}
  \label{fig:1000-inf}
\end{figure}

\begin{figure}[!ht]
  \centering
    \includegraphics[width=0.9\textwidth]{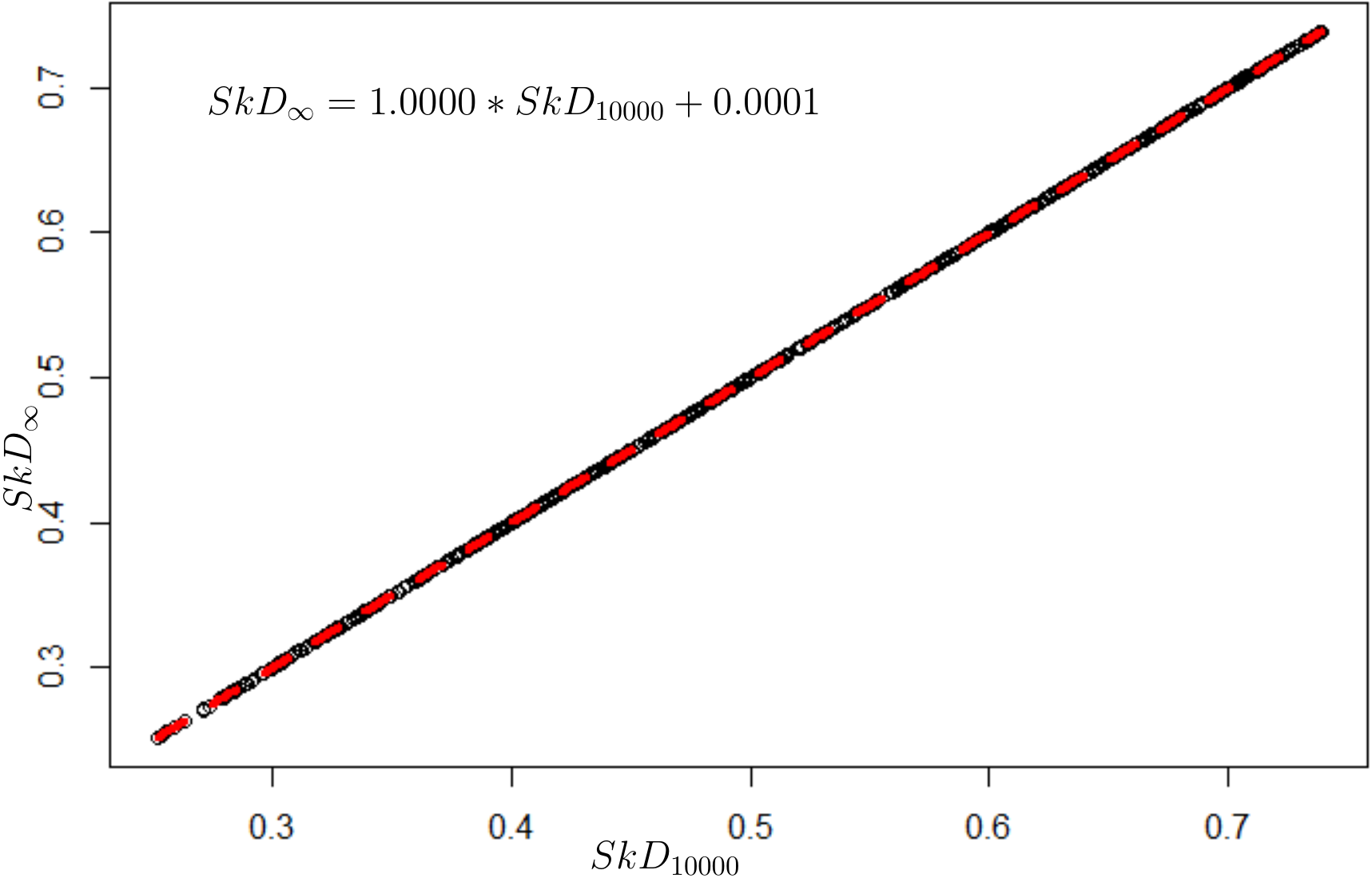}
  \caption{Fitting a linear model to $\SkD_{\infty}$ and $\SkD_{10000}$.}
  \label{fig:10000-inf}
\end{figure}

\begin{figure}[!ht]
  \centering
    \includegraphics[width=0.9\textwidth]{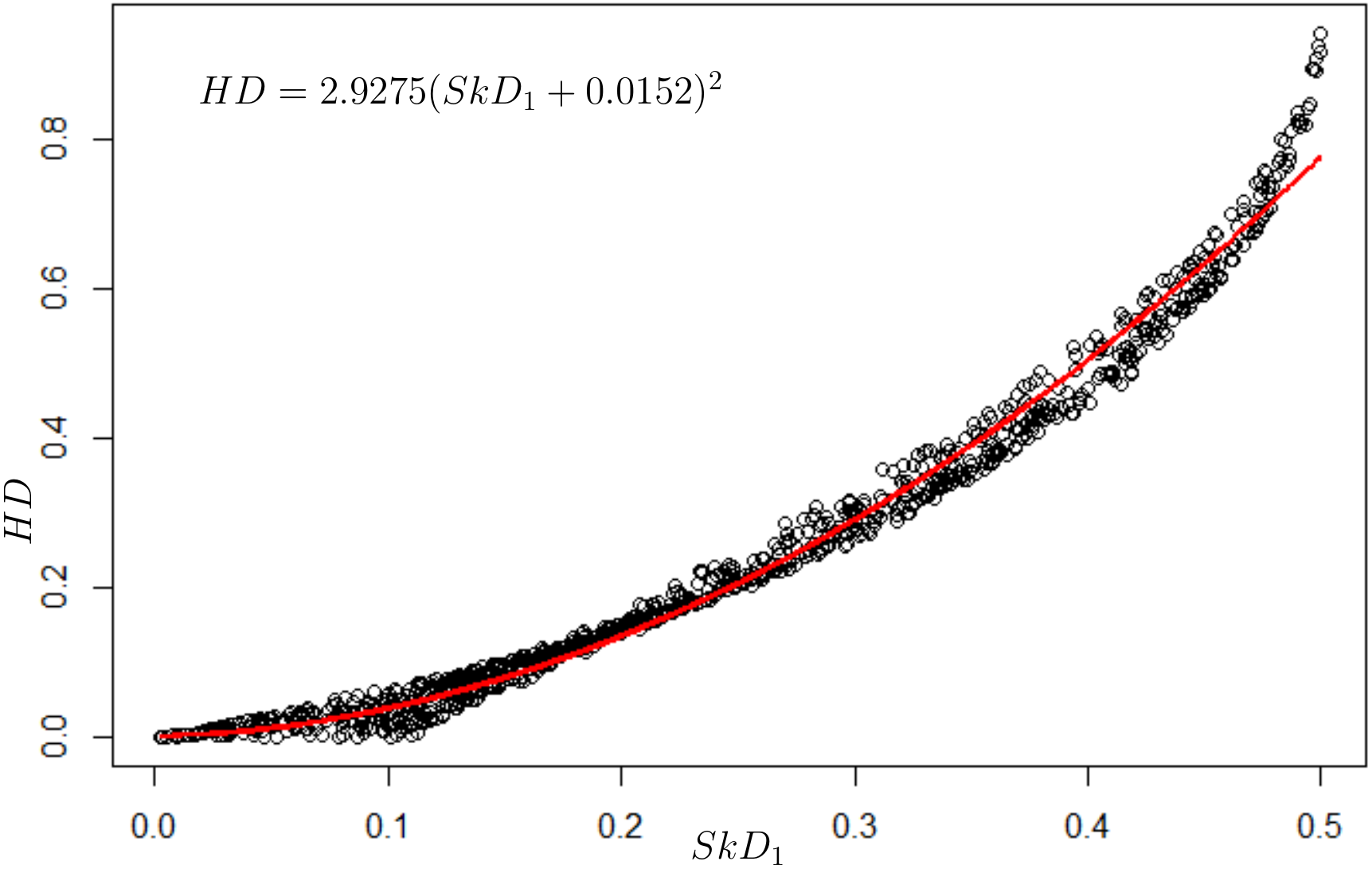}
  \caption{Fitting a quadratic model to $\HD$ and $\SkD_1$.}
  \label{fig:1-h}
\end{figure}

\begin{figure}[!ht]
  \centering
    \includegraphics[width=0.9\textwidth]{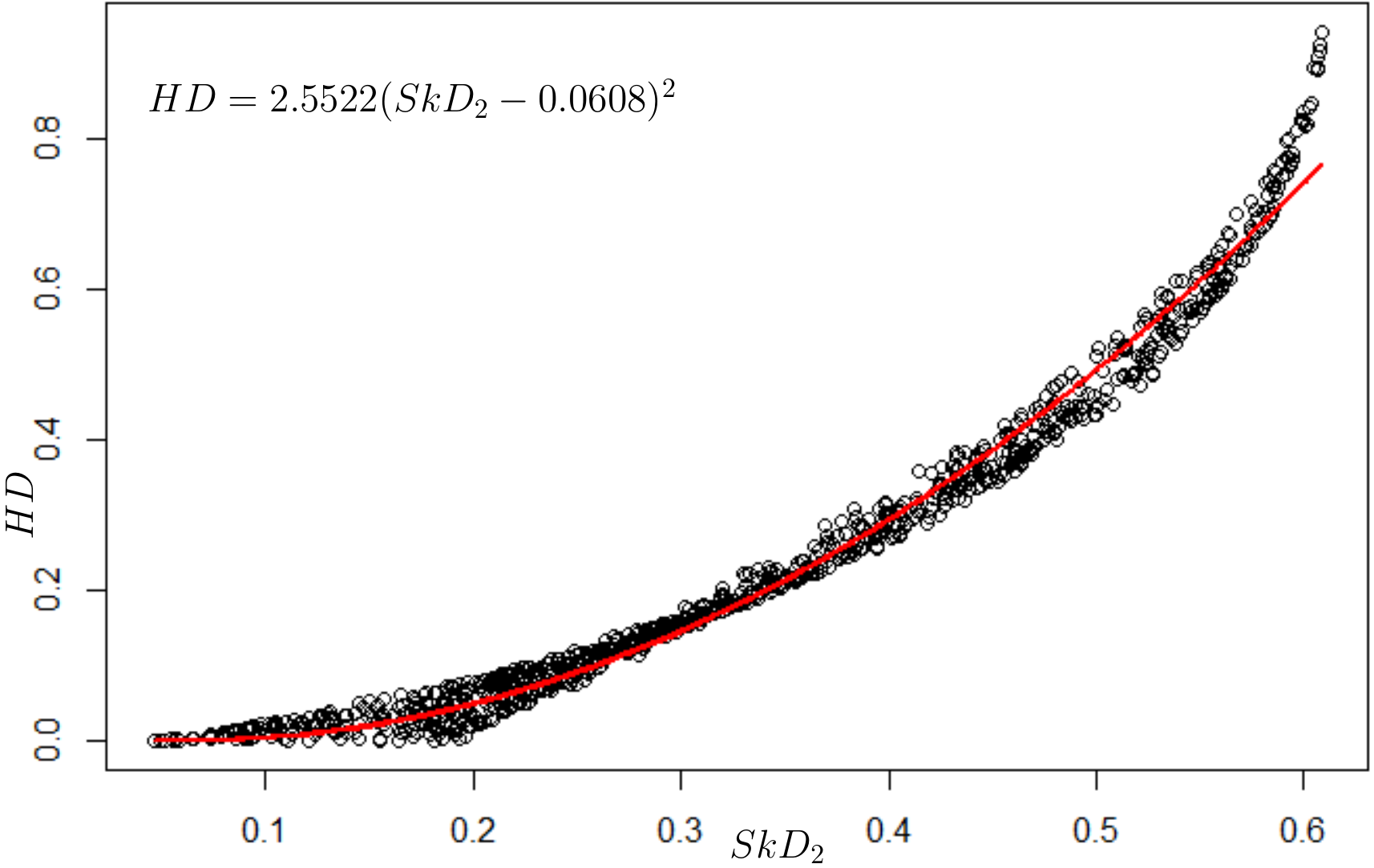}
  \caption{Fitting a quadratic model to $\HD$ and $\SkD_2$.}
  \label{fig:2-h}
\end{figure}

\begin{figure}[!ht]
  \centering
    \includegraphics[width=0.9\textwidth]{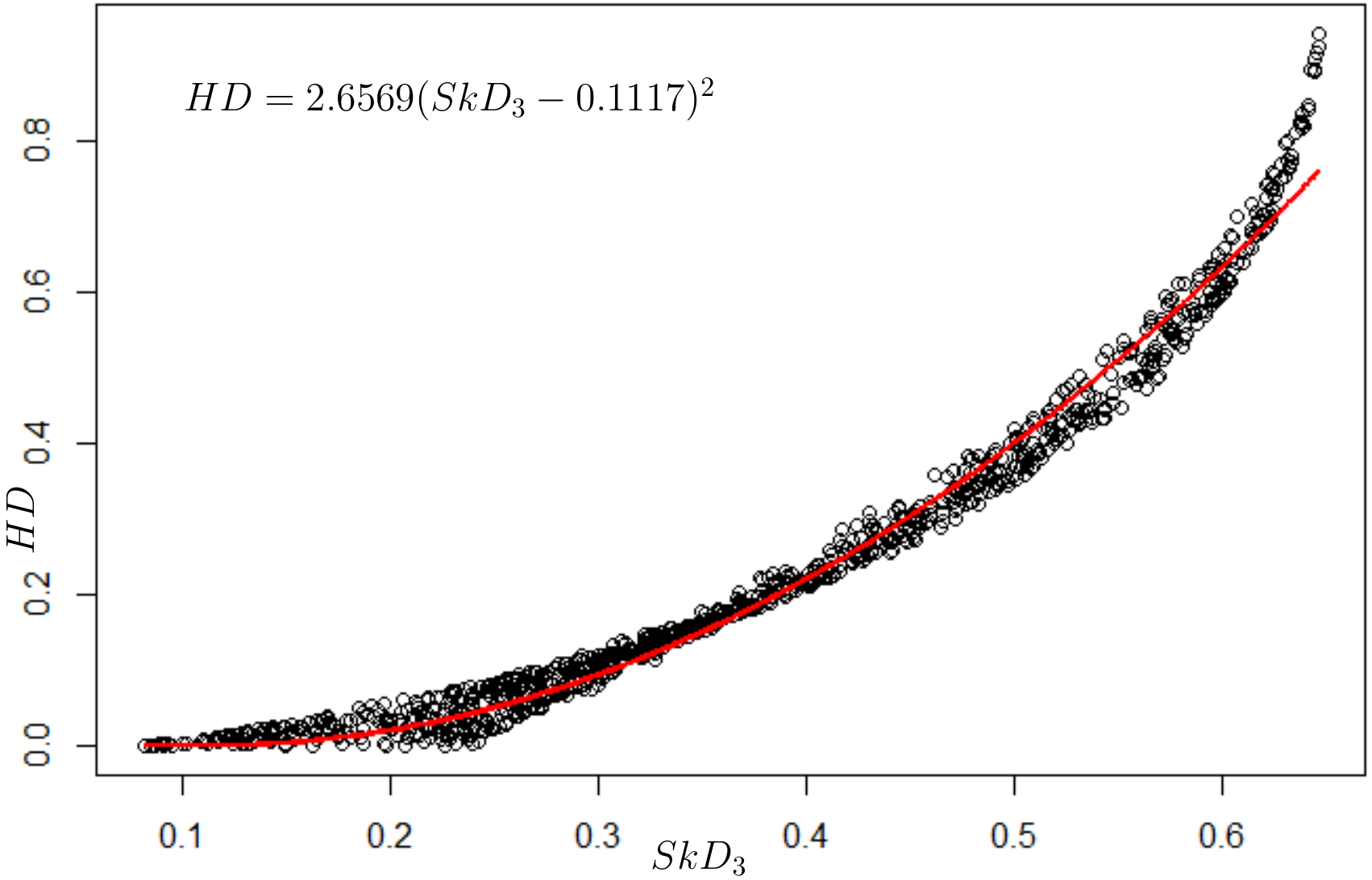}
  \caption{Fitting a quadratic model to $\HD$ and $\SkD_3$.}
  \label{fig:3-h}
\end{figure}

\begin{figure}[!ht]
  \centering
    \includegraphics[width=0.9\textwidth]{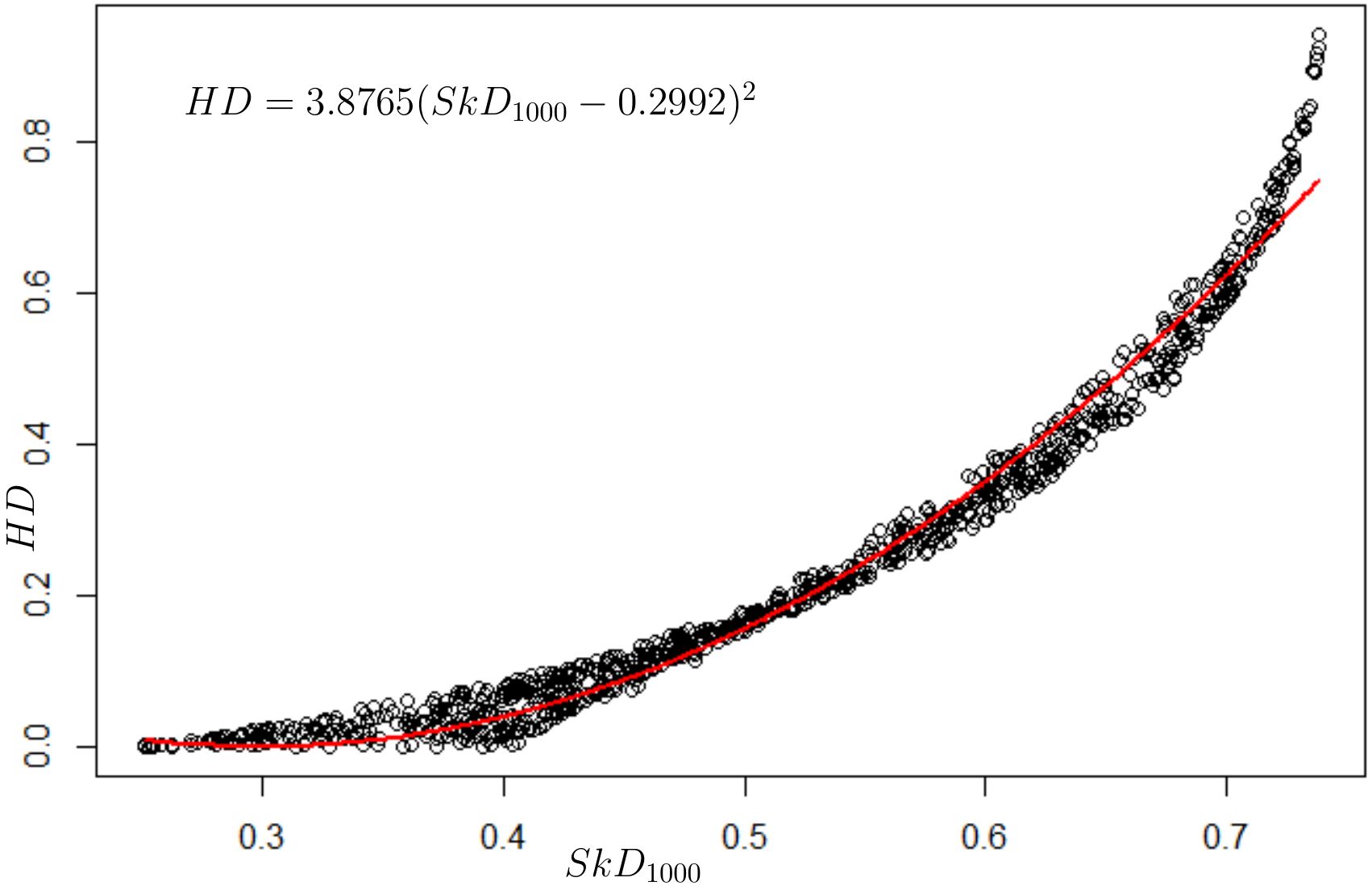}
  \caption{Fitting a quadratic model to $\HD$ and $\SkD_{1000}$.}
  \label{fig:1000-h}
\end{figure}

\begin{figure}[!ht]
  \centering
    \includegraphics[width=0.9\textwidth]{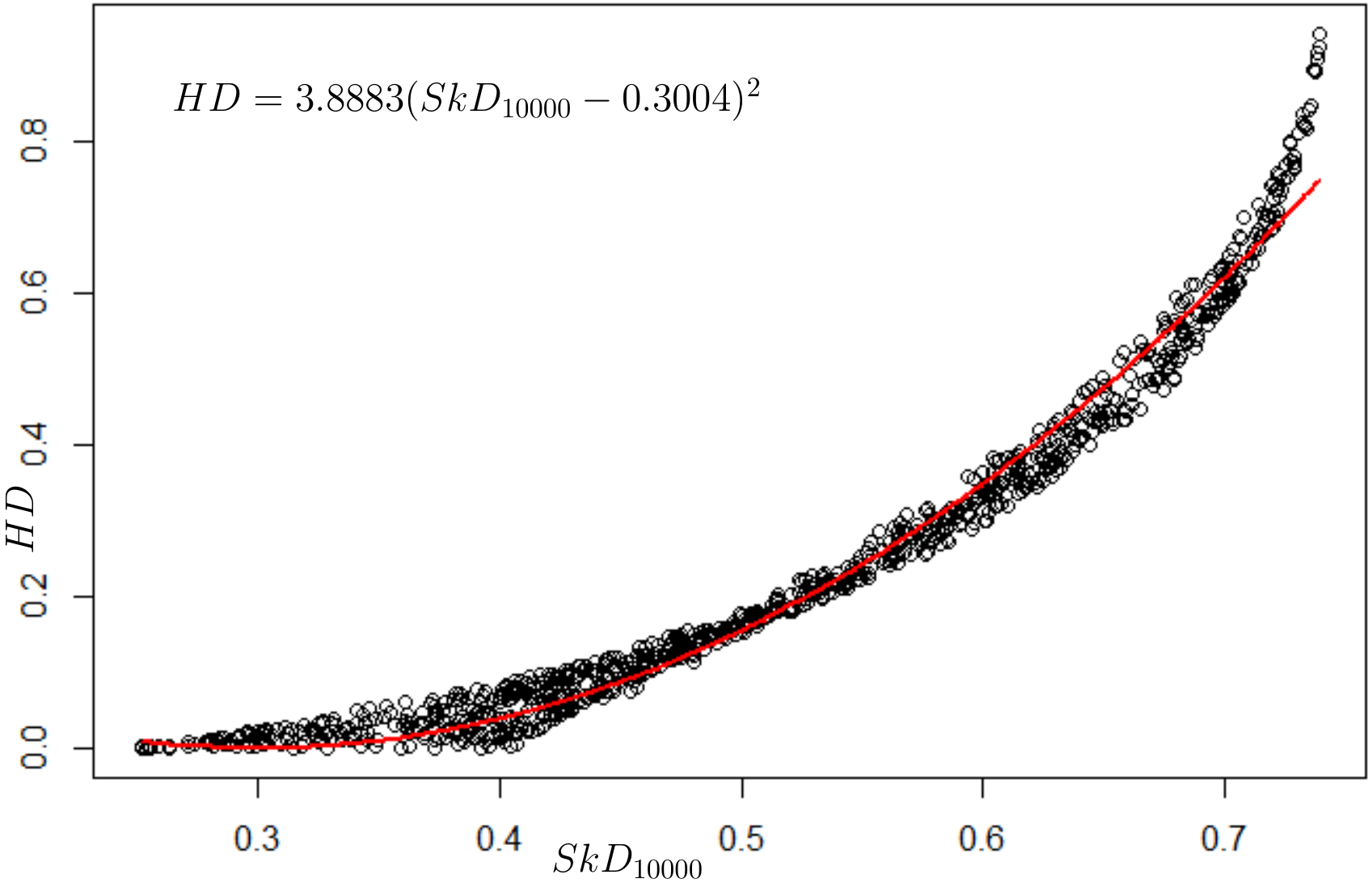}
  \caption{Fitting a quadratic model to $\HD$ and $\SkD_{10000}$.}
  \label{fig:10000-h}
\end{figure}

\begin{figure}[!ht]
  \centering
    \includegraphics[width=0.9\textwidth]{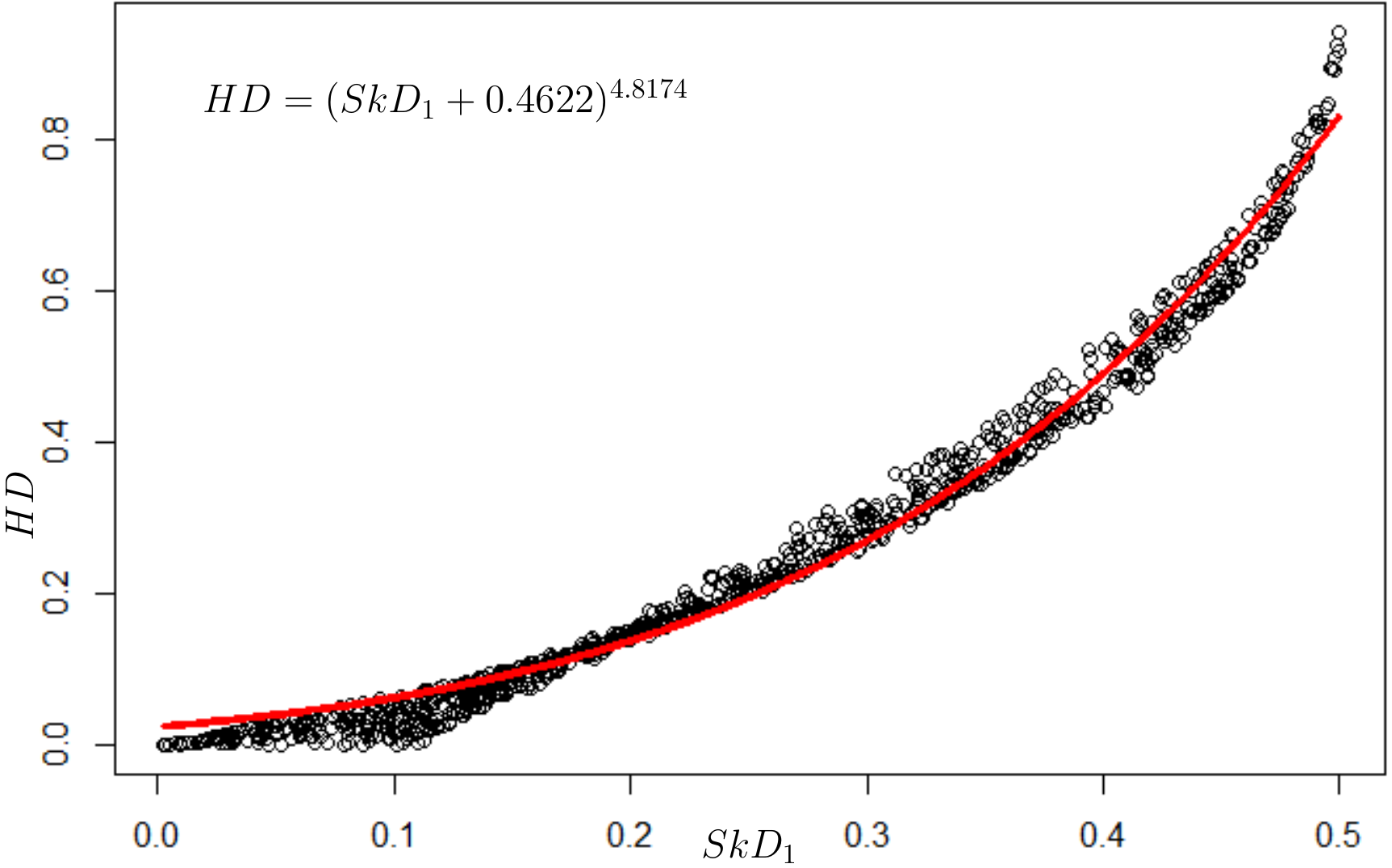}
  \caption{Fitting a power model to $\HD$ versus $\SkD_1$.}
  \label{fig:1-h-pow}
\end{figure}

\begin{figure}[!ht]
  \centering
    \includegraphics[width=0.9\textwidth]{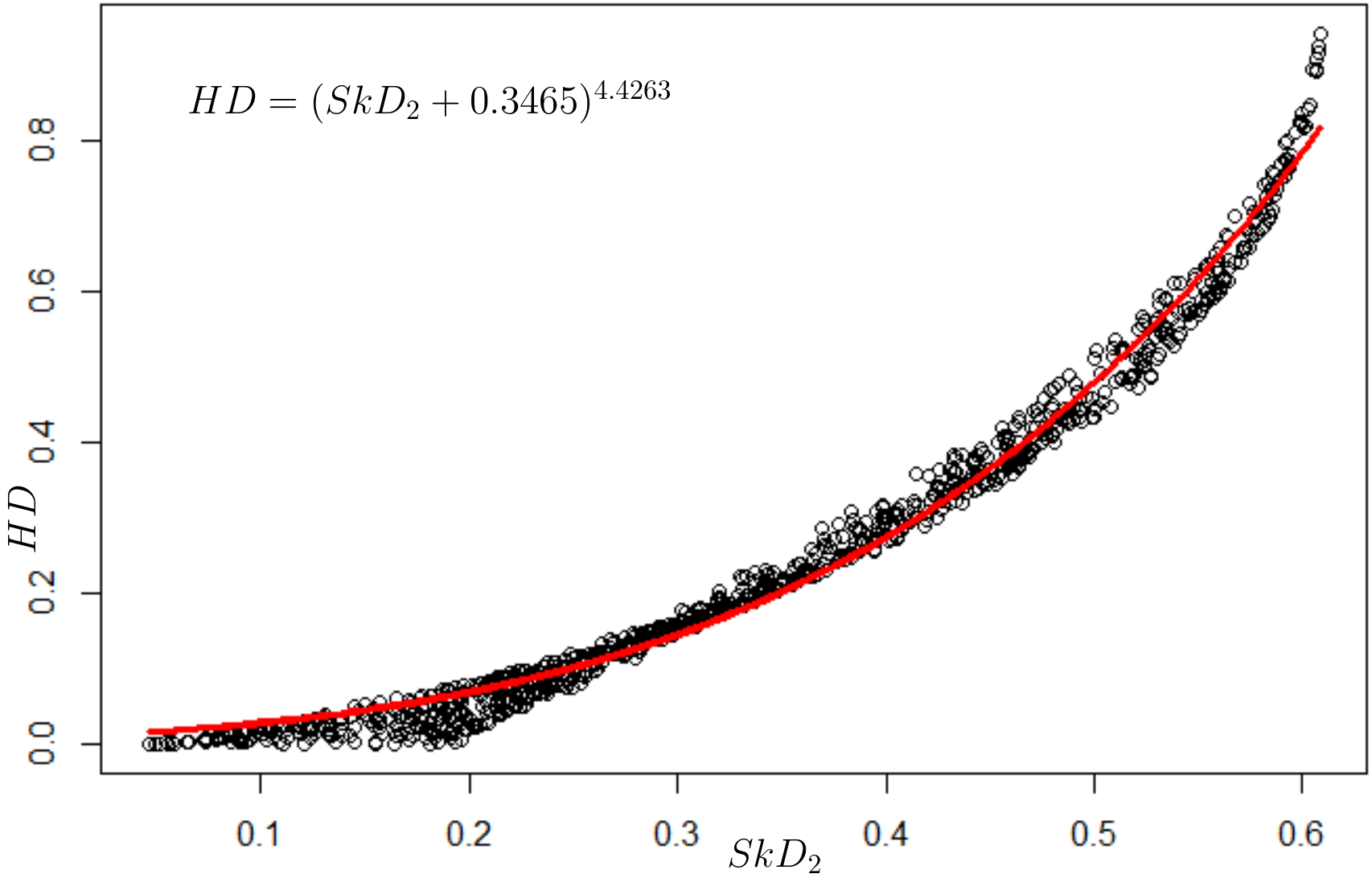}
  \caption{Fitting a power model to $\HD$ and $\SkD_2$.}
  \label{fig:2-h-pow}
\end{figure}

\begin{figure}[!ht]
  \centering
    \includegraphics[width=0.9\textwidth]{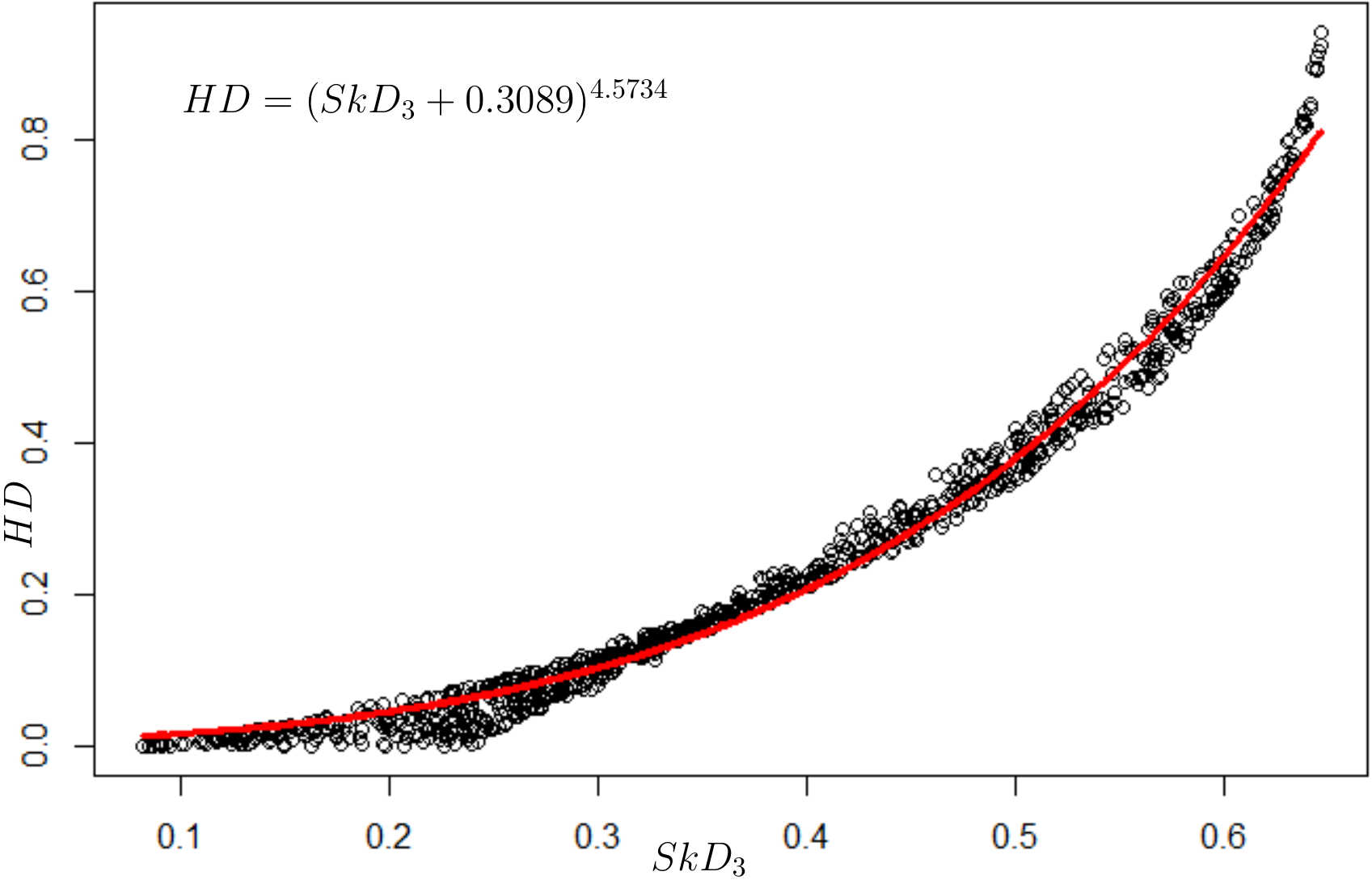}
  \caption{Fitting a power model to $\HD$ and $\SkD_3$.}
  \label{fig:3-h-pow}
\end{figure}

\begin{figure}[!ht]
  \centering
    \includegraphics[width=0.9\textwidth]{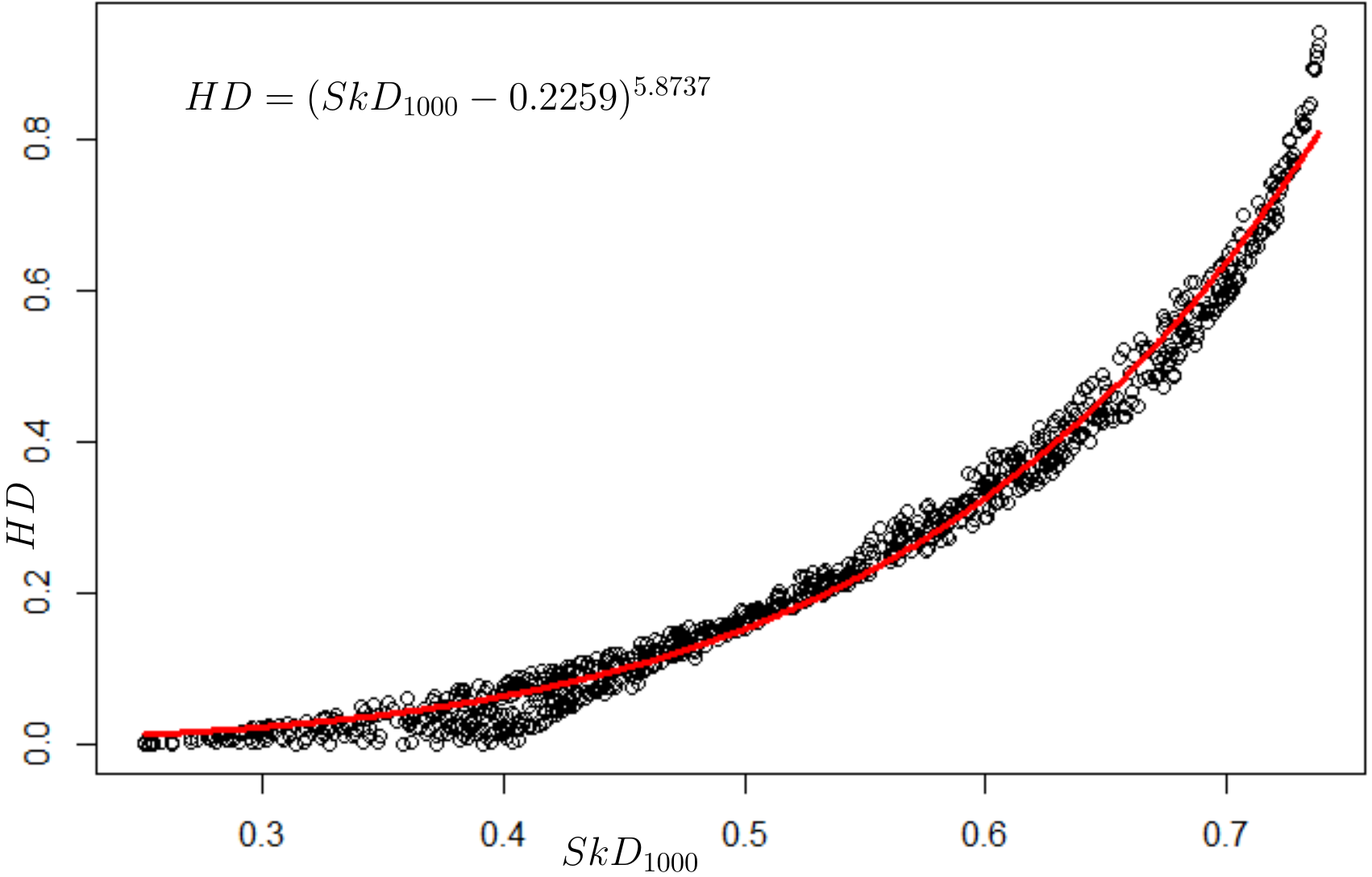}
  \caption{Fitting a power model to $\HD$ and $\SkD_{1000}$.}
  \label{fig:1000-h-pow}
\end{figure}

\begin{figure}[!ht]
  \centering
    \includegraphics[width=0.9\textwidth]{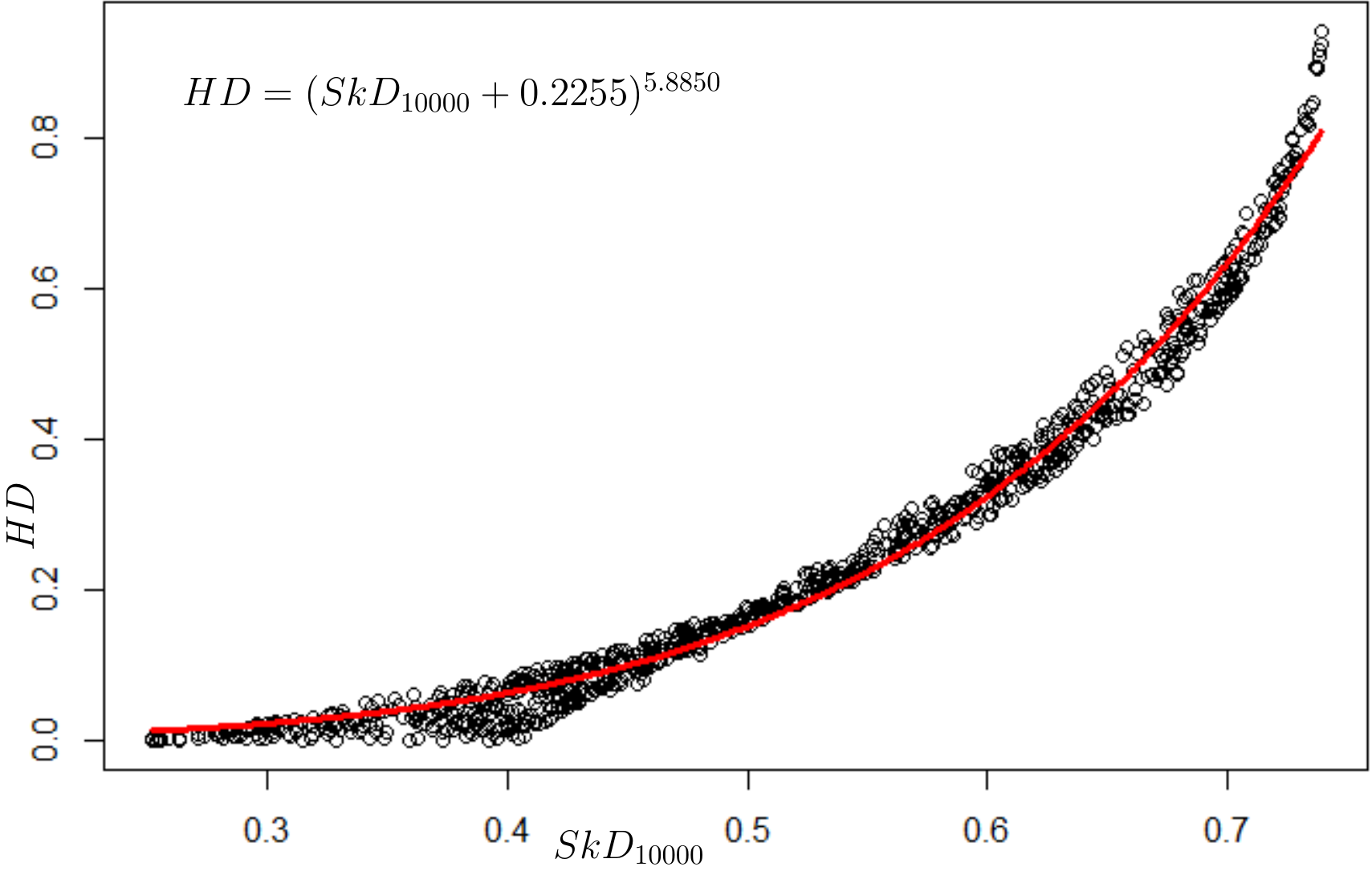}
  \caption{Fitting a power model to $\HD$ and $\SkD_{10000}$.}
  \label{fig:10000-h-pow}
\end{figure}
%relationships&experiments
%%---------------Chapter 8----------------
\chapter{Conclusion}
\label{ch:conclusion}
In this thesis we have presented a study of several depth functions such as halfspace depth, simplicial depth, and $\beta$-skeleton depth.  Emphasis was given to the studying of $\beta$-skeleton depth from both geometric and algorithmic viewpoints in $\mathbb{R}^2$. Furthermore, lower bounds for computing the planar $\beta$-skeleton depth ($1\leq \beta\leq \infty$), and approximation of different depth functions are also studied in this thesis. Finally, we provide some experimental results to support our approximation technique, and to illustrate the relationships among the aforementioned depth functions.
\section{Contributions}
Regarding the geometric perspective, in Chapter \ref{ch:geometric}, we proved that the exact bound $(1/4)n^4-\theta(n^3)$ for the combinatorial complexity of the arrangement of planar $\beta$-influence regions is achievable. The connectivity and convexity of $\beta$-skeleton depth regions are also studied in Chapter \ref{ch:geometric}.
\\\\For the algorithmic part, in Chapters \ref{ch:algorithmic} and \ref{ch:lowerbound}, we presented an optimal algorithm for computing the planar spherical depth (i.e. $\beta$-skeleton depth, $\beta=1$) of a query point. This algorithm takes $O(n \log n)$ time that matches the corresponding lower bound $\Omega (n\log n)$ proved in Chapter \ref{ch:lowerbound}. For the other values of $1<\beta\leq \infty$, employing the results on semialgebraic range counting problems in the literature, we developed an algorithm to compute the planar $\beta$-skeleton depth of a query point. This algorithm takes $O(n^{(3/2)+\epsilon})$ time. In developing this algorithm, we reduced the problem of computing the planar $\beta$-skeleton depth of a query point to a combination of at most $3n$ range counting problems, where $n$ is the size of the data set (Theorem \ref{thrm:q-skeleton}). In Section \ref{sec:tukey-alg}, we presented a simple and optimal algorithm for computing the planar halfspace depth. This algorithm takes $\Theta(n\log n)$ time. There are other optimal algorithms for this problem by Aloupis in \cite{aloupiscomputing}, and by Chan in \cite{chan2004optimal}. We used our specialized halfspace range counting method to obtain the number of data points in each halfspace whereas in Aloupis's algorithm the idea of sweeping halfline is employed to obtain such number.
\\\\The results in Chapter \ref{ch:lowerbound} include proving lower bounds for the complexity of computing planar $\beta$-skeleton depth, $1\leq \beta\leq \infty$. In this chapter, using different reductions, we proved that computing the planar $\beta$-skeleton depth of a query point allows us to answer the problem of Element Uniqueness, which takes $\Omega (n\log n) $ time. As such, computing the planar $\beta$-skeleton depth also requires $\Omega (n\log n) $ time.
\\\\Finally, in Chapter \ref{ch:experiments}, we investigated some relationships among the influence regions of $\beta$-skeleton depth and simplicial depth. We employed these relationships to prove that there exists $\beta^*<\infty$ such that for all $\beta\geq \beta^*$, and all $q$ in any finite range $\mathcal{R}\subset \mathbb{R}^2$ that contains $S$, the values of $\SkD_{\beta}(q;S)$ and $S_{\infty}(q;S)$ are equal. Also, we proved that the $\beta$-skeleton depth has a lower bound in terms of a constant factor of simplicial depth (in particular, $\SkD_{\beta}\geq(2/3)\SD$). In the remainder of Chapter \ref{ch:experiments}, we proposed a method of approximation, using the idea of fitting functions, to approximate one depth function by another one. To support the theoretical results in Chapter \ref{ch:experiments}, we provided some experimental results. As an example, the experimental results suggest that with a reasonable amount of error, the halfspace depth can be approximated by a quadratic function of the $\beta$-skeleton depth ($\HD\approxeq3.8883(\SkD_{\beta}-0.3004)^2$ if $\beta \to \infty$).
\section{Open Problems and Directions for Future Work}
\begin{open problem}
In Chapter \ref{ch:lowerbound}, we proved that planar $\beta$-skeleton depth requires $\Omega(n\log n)$ time. However, our best algorithm for computing the planar $\beta$-skeleton depth ($1<\beta\leq \infty$) takes $O(n^{(3/2)+\epsilon})$ time. Is it possible to develop an $O(n\log n)$ algorithm for this problem. Note that for the case of $\beta=1$, we presented such an algorithm.
\end{open problem}
\begin{open problem}
In studying the relationships among different depth functions, we proved that $\SkD_{\beta}\geq (2/3)\SD$. Is there some constant $c$ such that for all query points within the convex hull $\SkD_{\beta}\leq c\SD$.
\end{open problem}
\begin{open problem}
In Section \ref{sec:beta-alg}, we used the existing results on semialgebraic range counting and developed Algorithm~\ref{Alg:betaskeleton-pseudocode} with $O(n^{(3/2)+\varepsilon})$ query time, $O(n)$ storage, and $O(n\log n)$ expected preprocessing time. Can the query time be reduced to $O(n^{(4/3)+\varepsilon})$ by spending $O(n^{(4/3)+\varepsilon})$ time and space on the preprocessing. Or, is it possible to answer a semialgebraic range query in $O(n^{(1/3)+\varepsilon})$ time by spending $O(n^{(4/3)+\varepsilon})$ time and space on the preprocessing.   
\end{open problem}
\subsection{Future Work}
Some of the directions for future work include:  
\begin{itemize}
\item developing efficient algorithms to compute the $\beta$-skeleton median (central point).
\item studying the $\beta$-skeleton depth in higher dimensions.
\item using the $\beta$-skeleton depth to find outliers in a data set. This direction can be pursued to develop an efficient intrusion detection system. Regarding this application of $\beta$-skeleton depth: Algorithm \ref{Alg:sph-pseudocode} is used in an intrusion detection method proposed in \cite{sharafaldin2018eagleeye}. 
\end{itemize}%conclusion
%\include{bibliography}
% changes default name Bibliography to References
\renewcommand{\bibname}{References}
%\bibliographystyle{amsplain}
%\bibliography{mybibliography}

%%-------------Vita---------------------------
\clearpage
\phantomsection  % correct link in PDF bookmarks.
\addtocontents{toc}{\cftpagenumbersoff{chapter}}  % omit page # in ToC
\addcontentsline{toc}{chapter}{Vita}
\chapter*{Vita}
\pagestyle{empty}
\thispagestyle{empty}
\singlespacing
\textbf{Candidate's full name:} Rasoul Shahsavarifar\\
\textbf{University attended (with dates and degrees obtained):}\\
University of New Brunswick, Fredericton, Canada, 2014-2019, PhD\\
University of Tabriz, Tabriz, Iran, 2006-2008, MSc\\
Razi University, Kermanshah, Iran, 2001-2006, BSc\\

\textbf{Publications:}\\
Hassan Mahdikhani, Rasoul Shahsavarifar, Rongxing Lu, David Bremner, “Achieve Privacy-Preserving Simplicial Depth Query over Outsourced Cloud Platform”, manuscript submitted to Journal of Information Security and Applications, March 2019.\\
Rasoul Shahsavarifar and David Bremner, `` Approximate Data depth Revisited'', CCCG2018, University of Manitoba, Winnipeg, MB, Canada, 7-10 Aug 2018.\\
Rasoul Shahsavarifar and David Bremner, `` Approximate Data Depth Revisited'', arXiv preprint arXiv:1805.07373, 2018.\\
Rasoul Shahsavarifar and David Bremner, `` Computing the Planar $\beta$-skeleton Depth'', arXiv preprint arXiv:1803.05970, 2018.\\
Rasoul Shahsavarifar and David Bremner, `` An Optimal Algorithm for Computing the Spherical Depth of Points in the Plane'', arXiv preprint arXiv:1702.07399, 2017.\\

\textbf{Conference Presentations:}\\
Rasoul Shahsavarifar and David Bremner, `` Approximate Data depth Revisited'', CCCG2018, University of Manitoba, Winnipeg, MB, Canada, 7-10 Aug 2018.\\
Rasoul Shahsavarifar and David Bremner, `` Approximation of Data depth'', CMS2018, University of New Brunswick, Fredericton, NB, Canada, 1-4 June 2018.\\
Rasoul Shahsavarifar and David Bremner, `` On the Planar Spherical Depth and Lens Depth'', CCCG2017, Carleton university, Ottawa, ON, Canada, 2017.\\
Rasoul Shahsavarifar and David Bremner, `` Computing the spherical depth of points in the plane'', FWCG2016, University of New York, New York,USA, 2016.\\
Rasoul Shahsavarifar and David Bremner, `` Geometric and Computational Aspects of Data Depth'', Annual Graduate Research Conference, University of New Brunswick, Frederiction, NB, Canada, April 23, 2015.

\end{document}